\documentclass[10pt,usenames,dvipsnames]{article}
\usepackage[letterpaper, left=1.1truein, right=1.1truein, top = 1in, bottom = 1in]{geometry}

\usepackage{enumerate}
\usepackage{amsmath,amssymb}
\usepackage{natbib}
\usepackage{caption}
\usepackage{enumitem}
\usepackage{comment}
\usepackage[usenames]{color}
\usepackage{bm}
\usepackage{ying}
\usepackage{multirow}
\usepackage{rotating}
\usepackage{fullpage}
\usepackage{authblk}
\usepackage{mathtools}
\usepackage{parskip}
\usepackage{float}
\usepackage{subcaption}
\usepackage{tikz,tabularx}
\usetikzlibrary{patterns}
\usetikzlibrary{arrows}
\usetikzlibrary{tikzmark, positioning, fit, shapes.misc}

\usetikzlibrary{decorations.pathreplacing, calc}
\tikzset{brace/.style={decorate, decoration={brace}},
  brace mirrored/.style={decorate, decoration={brace,mirror}},
}

\newcolumntype{g}{>{\columncolor{red}}c}

\usepackage{graphicx,amssymb}
\usepackage{xcolor}

\usepackage[colorlinks,
            linkcolor=red,
            anchorcolor=blue,
            citecolor=blue
            ]{hyperref}
\usepackage{algorithm}
\usepackage{algorithmic}

\let\emptyset\varnothing 

\def\asto{{\stackrel{\textrm{a.s.}}{\to}}}

\def \iid {\stackrel{\text{i.i.d.}}{\sim}}

\def \calib {\textrm{calib}}

\def \test {\textrm{test}}

\def \hete {\textnormal{hete}}
\def \homo {\textnormal{homo}}
\def \fr {{\textnormal{FR}}}
\def \sdr {{\textnormal{SDR}}}
\def \tdr {{\textnormal{TDR}}}
\def \mdr {{\textnormal{MDR}}}
\def \mF {{\mathrm{F}}}
\def \cs {{\textnormal{CS}}}

\providecommand{\keywords}[1]
{{ 
  \fontsize{9}{12}\selectfont
  \textbf{\textit{Keywords:}} #1
}}

\theoremstyle{plain}

\allowdisplaybreaks

\usepackage{hyperref}
\def\@#1\@{\begin{align}#1\end{align}}
\def\$#1\${\begin{align*}#1\end{align*}}

\usepackage{color-edits}
\addauthor[Ying]{ying}{magenta}
\addauthor[Tian]{tian}{blue}

\usepackage{xcolor}
 
\title{Conformal Selective Prediction with General Risk Control}
\author[1]{Tian Bai}
\author[2]{Ying Jin}%\thanks{Email: yjinstat@wharton.upenn.edu}}
\affil[1]{Department of Statistics, Stanford University}
\affil[2]{Department of Statistics and Data Science, University of Pennsylvania}
\date{}

\begin{document}

\maketitle

\begin{abstract}
    In deploying artificial intelligence (AI) models, selective prediction offers the option to abstain from making a prediction when uncertain about model quality. To fulfill its promise, it is crucial to enforce strict and precise error control over cases where the model is trusted.  
    We propose Selective Conformal Risk control with E-values (SCoRE), a new framework for deriving such decisions for any trained model and any user-defined, bounded and continuously-valued risk. SCoRE offers two types of guarantees on the risk among ``positive'' cases in which the system opts to trust the model.  
    Built upon conformal inference and hypothesis testing ideas, SCoRE first constructs a class of (generalized) e-values, which are non-negative random variables whose product with the unknown risk has expectation no greater than one. Such a property is ensured by data exchangeability without requiring any modeling assumptions. Passing these e-values on to hypothesis testing procedures, we yield the binary trust decisions with finite-sample error control. SCoRE avoids the need of uniform concentration, and can be readily extended to settings with distribution shifts. 
    We evaluate the proposed methods with simulations and demonstrate their efficacy through applications to error management in drug discovery, health risk prediction, and large language models.
\end{abstract}
% \fontsize{10}{12}\selectfont

% !TEX root = main.tex

\section{Introduction}

% \yingcomment{Alternative story: uncertainty quantification for decision making -- one way is to use uncertainty to decide whether to trust a model (abstain or ask for help otherwise, etc.).}

% Trust in artificial intelligence (AI) systems is essential for their safe and effective deployment across science, healthcare, and industry. 
Limiting errors when deploying AI models is an indispensable component of their life cycle~\citep{wiens2019no,kompa2021second}. 
As model prediction errors are inevitable---arising from inadequate modeling, sampling uncertainty, and randomness in training---post-training mechanisms that manage  errors at deployment are especially important. 
A prominent approach is to deploy a model with an \emph{abstention} (or \emph{rejection}) option: a model is used only when it appears reliable and is withheld otherwise~\citep{chow2009optimum,el2010foundations}. 
% Selective prediction offers  the option to abstain from using a model when uncertain about its quality, thereby ensuring only trustworthy outputs are deployed~\citep{chow1957optimum,el2010foundations}. 
This paradigm, often called \emph{selective prediction}, aims to control errors precisely among the predictions we choose to deploy while maintaining high coverage, i.e., deploying as often as possible~\citep{geifman2017selective}.
% As highlighted in~\cite{geifman2017selective}, the ultimate goal is to equip the model with a mechanism to determine when to \emph{trust} it with precise error control, while keeping the coverage (i.e., how often the model is trusted) as high as possible. 
This leads to the  general question:
\vspace{0.25em}
\begin{quote}
\emph{Given a black-box model $f$, labeled data $\{(X_i,Y_i)\}_{i=1}^n$ and a new instance $X_{n+1}$, can we derive a trust decision $\psi_{n+1}\in \{0,1\}$ that controls an unknown risk $L_{n+1}$ among those with $\psi_{n+1}=1$?}
\end{quote}
\vspace{0.25em}

% We shall also leverage a set of labeled data $\{(X_i,Y_i)\}_{i=1}^n$ where $n\in \NN$. 
Most prior work addresses this problem for classifiers $f$ with binary risks $L_{n+1}\in \{0,1\}$, typically offering either asymptotic control of a selective error rate or  finite-sample bound based on  uniform concentration of empirical classification errors. 
Recent extensions of conformal prediction provide finite-sample, distribution-free guarantees for selective tasks with binary risks~\citep{vovk2005algorithmic,jin2023selection,jin2023model}, and have been used to identify trustworthy AI outputs in applications such as compound screening~\citep{Bai_Tang_Xu_Svetnik_Khalili_Yu_Yang_2024},large language models~\citep{gui2024conformal,jung2024trust}, and medical foundation models~\citep{jin2026act}. 
However, many high-stakes applications demand control of \emph{continuously-valued} risks, where a principled and powerful ``trust'' mechanism remains underdeveloped:
% \vspace{0.25em}
\begin{itemize}[left=0.5em]
    \item In \emph{drug discovery}, the early screening phase uses AI models to identify drug candidates with high binding affinities for follow-up experiments. False leads waste resources, and a natural quantitative risk is a (continuous) development cost incurred by pursuing an inactive candidate~\citep{jin2023selection,Bai_Tang_Xu_Svetnik_Khalili_Yu_Yang_2024}, e.g., $L_{n+1}=\text{cost}\cdot\ind\{Y_{n+1}\leq c\}$ for the unknown affinity $Y_{n+1}$ and a known threshold $c\in \RR$.  
    \item In \emph{radiology report generation}, an AI-generated report  is useful only when it is sufficiently close to expert references~\citep{gui2024conformal}. 
    Here, the risk can be naturally continuous, such as a  semantic distance between the model output $f(X_{n+1})$ and the (unknown) expert-level reference report $Y_{n+1}$. 
    \item In \emph{healthcare management}, hospitals routinely use predictions of continuous outcomes, such as ICU length of stay, to support downstream planning and interventions~\citep{bertsimas2020predictive,marafino2021evaluation,hu2025implementing}. Practitioners may seek to deploy only highly accurate predictions~\citep{jin2026act}, where the risk can often be a continuous metric such as the squared prediction error. 
\end{itemize}
\vspace{0.25em}
Besides the focus on continuous risks, these settings also call for different \emph{notions} of risk control tied to downstream objectives: one may seek to bound the expected total risk accumulated over deployed instances, while another may prioritize the expected risk per deployed instance. As we shall see, these considerations reflect distinct error notions. Ideally, after all, such guarantees should be finite-sample and distribution-free, applying to any black-box model under mild exchangeability assumptions.

\begin{figure} 
    \centering
    \includegraphics[width=\linewidth]{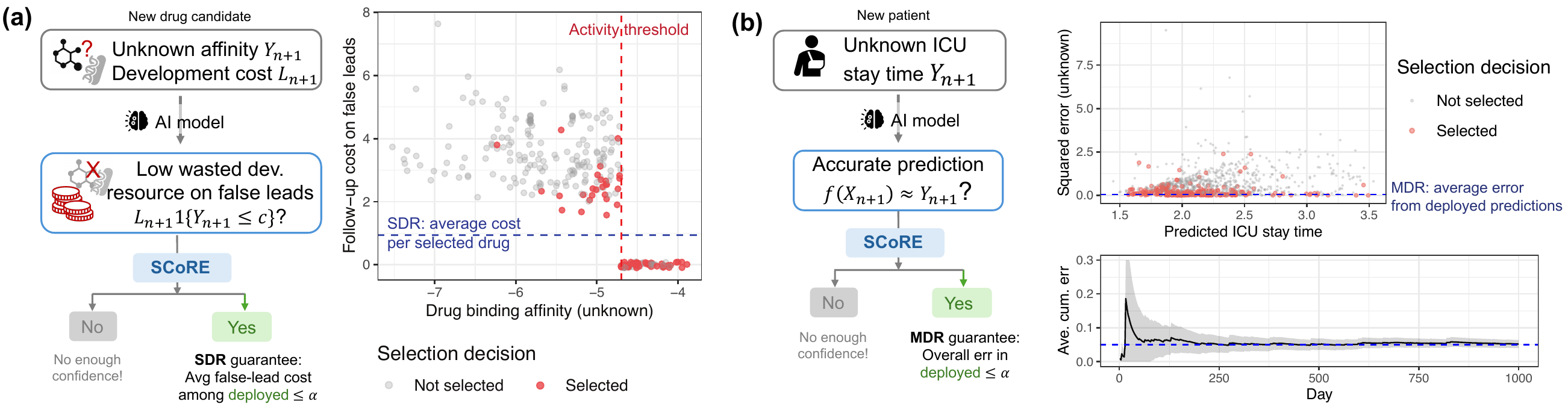}
    \caption{%Application of SCoRE to drug discovery and clinical prediction error management. 
    {\small \textbf{Application of SCoRE. (a) Drug discovery.} Left: given predictions of an unknown drug binding affinity $Y_{n+1}$, SCoRE controls the average cost $L_{n+1}\ind\{Y_{n+1}\leq c\}$ among the selected compounds. 
    Right: in a real drug discovery dataset, the average cost among selected candidates (red dots below activity threshold) is below $\alpha=1$. 
    \textbf{(b) Clinical prediction.}  Left: SCoRE identifies predictions of health outcomes with small error $f(X_{n+1})\approx Y_{n+1}$ with MDR control, ensuring a low total squared error in deployment. Right: selection results in a semi-synthetic dataset (upper), and mean squared error per day  when 50 patients await predictions every day (lower).}}\vspace{-1em}
    \label{fig:intro_example}
\end{figure}

\subsection{Our contributions}

We introduce Selective Conformal Risk control with E-values (SCoRE), a new framework that provides 
finite-sample, distribution-free control of bounded, continuously-valued risks in selectively trusting any model. 
Viewing trust as a binary decision for each test instance, we formalize two criteria:
\vspace{0.25em}
\begin{enumerate}[label=(\roman*)]
    \item Marginal deployment risk (MDR): $\EE[L_{n+1}\psi_{n+1}]$, the expected risk incurred by deployed instances; 
    \item Selective deployment risk (SDR): $\EE [ (\sum_j L_{n+j}\psi_{n+j})/(1\vee \sum_j \psi_{n+j})]$ when  given multiple test instances $\{X_{n+j}\}_{j=1}^m$, which quantifies the average risk \emph{per} deployed instance.
\end{enumerate}
\vspace{0.25em}
See Section~\ref{subsec:def_risk} for formal definitions. 
Both metrics target ``positive'' deployed cases and conceptually parallel  type-I error metrics in hypothesis testing. The SDR, which requires an intrinsically ``selective'' treatment, extends the selective prediction literature beyond binary risks~\citep{chow2009optimum,el2010foundations,geifman2017selective}, while the MDR offers a complementary perspective within our unified framework.

Figure~\ref{fig:intro_example} previews two representative applications. 
In a drug discovery task (panel (a)), our SDR-control procedure selects compounds while controlling the average cost wasted on false leads. In a clinical prediction task (panel (b)), our MDR-control procedure identifies highly accurate predictions and tightly controls the total prediction error accumulated across daily batches  (lower right, divided by 50). 
Achieving such guarantees is nontrivial because the MDR and SDR concern the unknown risk on a data-dependent subset of test instances: we must decide which test instances to deploy using calibration data (and predictions for the risk), yet the deployment risk of each selected instance depends on an unseen outcome.
% In panel (a), we apply SCoRE to a drug discovery task where the goal is to identify compounds with unknown binding affinity $Y_{n+1}>c$ for a pre-specified value $c\in \RR$, while limiting a continuously-valued development cost of inactive compounds. Our SDR-control variant selects compounds (red dots) while ensuring the average wastage of development cost  per selected molecule is controlled below a given threshold. 
% In panel (b), we apply SCoRE to identify accurate predictions for a health outcome. Simulating a scenario where $50$ patients arrive per day, our MDR-control variant selects those with small prediction error (red dots in the upper right), and the total prediction error in deployment among the daily batch of patients is tightly controlled (lower right, divided by 50). 

% 
Methodologically, SCoRE connects selective deployment to hypothesis testing with e-values~\citep{vovk2021values}.  
The key idea is to connect a deploy decision with a reject decision in hypothesis testing.  
We show that applying standard hypothesis testing procedures that threshold a class of (risk-adjusted) e-values obeying $E_{n+j}\geq 0$ and $E[L_{n+j} E_{n+j}] \leq 1$ leads to finite-sample MDR and SDR control. 
For each test point, we use a set of labeled calibration data to construct such an e-value  under standard exchangeability conditions, which then leads to risk control.  Figure~\ref{fig:viz} summarizes the workflow.
While e-values have been used to test (deterministic) hypotheses~\citep{vovk2021values,ramdas2024hypothesis}, their expectation-based validity is a natural match for controlling the expectation of unknown risks.  
Notably, our guarantees  only require  the exchangeability among data. It avoids uniform concentration arguments common in selective prediction~\citep{geifman2017selective}, accommodates dependence among data (e.g., predictions from graphs)~\citep{huang2024uncertainty}, and extends naturally to covariate shift settings~\citep{TibshiraniBCR19}.

% The validity of such e-values, and thus the guarantees of our methods, only requires the exchangeability among data. This brings several advantages. First, inheriting the model-free nature of conformal inference, our method does not require modeling assumptions on the data and is applicable to any  ``black-box'' model. Second, it avoids the need of  uniform concentration inequalities common in selective prediction, which can be conservative when the sample size is limited~\citep{geifman2017selective}. 
% Third, it accommodates dependency such as predictions from graph neural networks (GNNs) trained on test data~\citep{huang2024uncertainty}. 
% Finally,  it
% enables natural extensions to  complex settings such as those with distribution shifts~\citep{TibshiraniBCR19}, which we demonstrate as an extension of our main framework. % and online decisions with streaming data~\citep{xu2024online}. 

Finally, risk control should be balanced with utility~\citep{geifman2017selective}:
a mechanism that abstains too often  limits the model power.
We analyze the power through any user-specified reward of deployment,
% In addition, in certain settings, deploying the model on specific key instances may be preferable: in drug discovery, scientists may prioritize drug candidates whose molecular structures are substantially dissimilar from existing ones.
% We thus propose the notion of deployment reward that describes the value of deploying an instance, and define the power as the total reward in deployed instances. 
leading to a Neyman--Pearson-type characterization of the asymptotically optimal scores that guide selection. 
We also develop practical strategies to achieve this when the risk is consistently estimated.

\begin{figure} 
    \centering
    \includegraphics[width=0.8\linewidth]{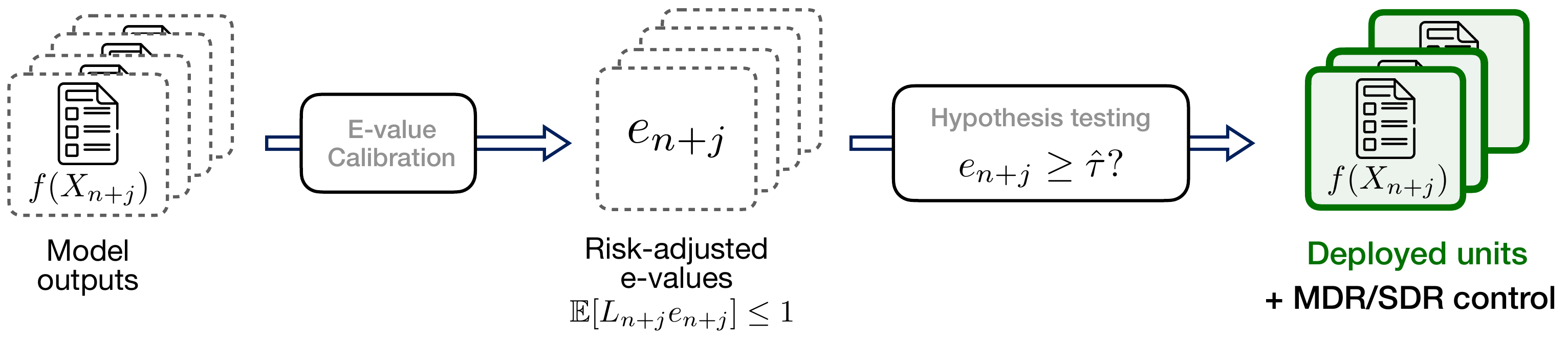}
    \caption{{\small Visualization of the SCoRE workflow. Starting with any model outputs for unlabeled test points and a score that estimates the deployment risks, we use a set of calibration data to construct a risk-adjusted e-value for every test sample, and pass them on to hypothesis testing procedures and select test samples with reliable prediction.}}
    \label{fig:viz}
\end{figure}

\vspace{-0.5em}
\paragraph{Paper outline.} The rest of the paper is organized as follows. Section~\ref{sec:setup} sets up the problem, introducing the two deployment risk metrics with concrete examples. Section~\ref{sec:general_method} introduces the general method of SCoRE, including the notion of risk-adjusted e-values and how they can be used to achieve MDR and SDR control via hypothesis testing. 
Section~\ref{sec:marginal} and Section~\ref{sec:selective} introduce the concrete procedures in constructing these e-values for MDR and SDR control, respectively. 
Section~\ref{sec:weighted} briefly describes a natural extension to the covariate shift setting. 
Demonstration of representative applications and simulations are in Section~\ref{sec:real} and Section~\ref{sec:simu}. 

\vspace{-0.5em}
\paragraph{Data and code.} Reproducibility code for both our simulation and real data experiment can be found at the Github repository \url{https://github.com/Tian-Bai/SCoRE}.

\section{Problem setup}
\label{sec:setup}

\subsection{Defining deployment risk}
\label{subsec:def_risk}

We begin by introducing the setup and our notions of deployment risk. Assume access to a set of labeled (calibration) data $\cD_\calib = \{(X_i,Y_i)\}_{i=1}^n$, 
and a set of unlabeled (test) data $\cD_\test = \{X_{n+j}\}_{j=1}^m$ 
whose labels $\{Y_{n+j}\}_{j=1}^m$ are unobserved. 
For now, we assume  that $\{(X_i,Y_i)\}_{i=1}^{m+n}$ are exchangeable across $i\in [m+n]$; we relax this in Section~\ref{sec:weighted} to covariate shift settings. 
Here $X_i\in \cX$ is the feature and $Y_i\in \cY$ is the label. 
We are interested in deploying a model $f\colon \cX\to \cY$. 
% for convenience, we assume its training process is independent of the calibration and test data. 
% , but all techniques apply under the minimal condition that $\{(f(X_i),Y_i)\}_{i=1}^{n+m}$ are exchangeable, such as with GNNs trained over an entire graph~\citep{huang2024uncertainty}. 
It may be a regression model with $\cY=\RR$, or a classification model with $\cY = \{1,\dots,K\}$, or a language model where $\cY$ is the space of natural language.

We quantify the consequence of erroneously deploying $f$ on a new instance $X$ with unknown outcome $Y$ by a numerical \emph{risk} $\cL(f,X,Y) \in \RR^+$, where $\cL(\cdot)$ is a known mapping. 
Throughout, we work with a bounded risk, and without loss of generality, assume $\cL(f,X,Y)\in [0,1]$. 
Concrete examples of risks are discussed in Section~\ref{subsec:intro_app}. 
The risk for the $j$-th test point is denoted as $L_{n+j} = \cL(f,X_{n+j},Y_{n+j})$, which is \emph{unknown} since the label $Y_{n+j}$ is not observed. 
To formalize optimality of deployment outcomes, we allow a user-specific \emph{reward} for deployment, represented by a random variable $r(f,X,Y)\in \RR^+$, where $r(\cdot)$ is a known mapping. Intuitively, $r$ captures the utility of deploying a model on a test instance, such as the scientific value, operational benefit, or downstream savings in resources. 

We will use a pre-trained score function $s\colon \cX\to \mathbb{R}$ to calibrate the deployment decisions, and our procedure prioritizes instances with smaller scores $s(X)$ (so smaller scores shall indicate preliminary evidence for safer instances). The validity of our procedures does not rely on the choice of of $s$.  
We assume for convenience that both $f(\cdot)$ and $s(\cdot)$ are trained independently of $\cD_\calib$ and $\cD_\test$. More generally, our results apply as long as the triplets $(s(X_i),f(X_i),Y_i)$'s are exchangeable across $i\in [n+m]$, such as with graph neural networks trained over an entire graph with a separate labeled training data and all features of $i\in [n+m]$~\citep{huang2024uncertainty}. 
% (e.g., a graph neural network trained over the entire graph). 
A natural idea is to set $s(X)$ as a prediction for $\cL(f,X,Y)$; we discuss optimal score choice later.

Our goal is to construct binary decisions $\hat\psi_{n+j} \in \{0,1\}$ for all $j\in [m]$, which may depend on both $\cD_\calib$ and $\cD_\test$. 
Here, $\hat\psi_{n+j}=1$ means to deploy/trust the model for $X_{n+j}$, and $\hat\psi_{n+j}=0$ means abstention. 
What \emph{deploying} a model means depends on the context (e.g., sending a compound to wet-lab follow-up, accepting an automated clinical prediction, or releasing an LLM-generated report). 
% In drug discovery where $X_{n+j}$ is a molecule, it means we select $X_{n+j}$ as a promising candidate for further investigation. 
% In cancer subtyping where $X_{n+j}$ is the  scanned image of a tumor, it means we trust the diagnosis given by an AI pathological model. 
Following~\cite{geifman2017selective}, we emphasize risk control over the trusted cases and consider two error metrics. 

\vspace{-0.5em}
\paragraph{Marginal deployment risk (MDR).}
The first error metric concerns the overall (expected) risk. Given a user-specified error rate $\alpha$, we aim to develop $\hat\psi_{n+j}\in \{0,1\}$ such that 
\@\label{eq:def_mgn_risk}
\mdr := \EE[L_{n+1} \cdot \hat\psi_{n+1}  ]
\@
is controlled below $\alpha$. 
This is an analogue for classical type-I error control~\citep{lehmann1986testing} for a random and non-binary risk. 
% This notion is related to  classical hypothesis testing for a deterministic null hypothesis $H_0$. There, the Type-I error is defined as $\EE[\ind\{H_0\text{ is true}\} \cdot \hat\psi ]$ for a test function $\hat\psi\in \{0,1\}$, where $\hat\psi=1$ means rejecting $H_0$~\citep{lehmann1986testing}.  
% In our context, however, the risk is random and non-binary.
% Conformal inference can be viewed as inverting tests with a transformed risk:
% \$
% \EE[ \ind\{ V(X_{n+1},Y_{n+1}) > c\} \ind\{\psi(X_{n+1};c) = 1\} ] \leq \alpha
% \$
% where $\psi(X_{n+1};c)$ is the test statistic for testing a hypothesis $H_0(c)\colon V(X_{n+1},Y_{n+1})\leq c$, and the risk is equal to $1$ whenever $H_0(c)$ is true. 
% 
It is useful to interpret MDR with multiple test points, in which  controlling~\eqref{eq:def_mgn_risk} at $\alpha $ implies control over the \emph{total deployment risk} (TDR): 
\@\label{eq:def_total_risk}
\tdr:= \EE\big[   \textstyle{\sum_{j=1}^m} L_{n+j} \hat\psi_{n+j}\big] \leq \alpha  m.
\@
That is, the total risk accumulated by the deployed instances $\cR = \{j\colon \hat\psi_{n+j}=1\}$ is controlled. 
% In this case, one finds a subset $\cR = \{j\colon \hat\psi_{n+j}=1\}$ of test instances on which the model can be deployed while controlling the total risk below a budget limit.  

\vspace{-0.5em}
\paragraph{Selective deployment risk (SDR).}
The second type of error we study measures the average risk per deployed unit. Formally, letting $\cR = \{j\in [m]\colon \hat\psi_{n+j}=1\}$ be the set of deployed units, we define  
\@\label{eq:def_sel_risk}
\sdr:= \EE\Bigg[ \frac{\sum_{j=1}^m L_{n+j} \cdot \ind\{j\in \cR\}}{1\vee|\cR|}   \Bigg]  .
\@ 
The SDR is motivated by, and generalizes, the false discovery rate (FDR) in classical hypothesis testing~\citep{benjamini1995controlling}. 
In particular, if we set  $L_{n+j} = \ind\{H_{0,j}\text{ is true}\}$ for a set of deterministic null hypotheses $\{H_{0,j}\}_{j=1}^m$, then SDR reduces to the usual FDR. 
In prediction problems, SDR connects to  the model-free selective inference problem~\citep{jin2023selection,jin2023model,gui2024conformal} when $L_{n+j} = \ind\{Y_{n+j}\leq c_{n+j}\}$ represents a binary ``bad event'' (e.g., the outcome is not sufficiently large relative to a cutoff $c_{n+j}\in \RR$), and our SDR-control procedure~\eqref{eq:def_sel_risk}  reduces to the methods studied there. When $m\to \infty$, by the law of large numbers, the SDR is close to the risk conditional on deployment~\citep{geifman2017selective} $\EE[L_{n+1}\given \psi_{n+1}=1]$, and more broadly, to  marginal FDR-type notions~\citep{storey2002direct}. However, our formulation allows the development of effective solutions, whereas those criteria can be difficult to control in finite sample in a model-free fashion.

% Note that these risks are never observed on the test points at deployment time, and the deployment set can be chosen in a data-dependent way; providing distribution-free control of the risk among deployed instances therefore requires careful calibration beyond naive score thresholding.

\vspace{-0.5em}
\paragraph{When to use which?} The two metrics serve different goals. MDR is natural when there is a fixed risk budget and does not require the risk to scale with the number of deployments: a procedure may deploy few but comparatively risky cases yet still controlling the MDR.  
%, while the SDR is suitable %for large-scale problems where the total error can grow with the frequency of deployment.  
On the other hand,  SDR is suitable when one requires that \emph{only low-risk cases are deployed}, so the incurred risks scale with the number of deployed cases. Such distinctions mirror those between the type-I error and FDR, which have been extensively discussed in the statistics literature (see, e.g.,~\cite{ioannidis2005most,benjamini1995controlling}).

%Also, similar to the discussion there, the marginal control of~\eqref{eq:def_total_risk} does not imply the deployed units are of a high quality. Consider an extreme case where $L_{n+1}\equiv 1$, then, one may set $\hat\psi_{n+1}=1$ with probability $\alpha$, in which case the total risk will be controlled, yet every deployed instance incurs a high error. 

% In general, $Y_{n+j}$ may represent some risk associated with 
% each $j$; in some decision-making problems, the risk of making a false positive 
% may depend on the actual outcome. 
% For instance, in early disease diagnosis, 
% a sensible task is to filter out `safe' patients (whose health risk is low) while 
% controlling false negatives; 
% in this case, the risk of any false negative often depends on 
% how large the actual risk is, and counting the number of errors may not be appropriate. 
 
% \tiancomment{Should we introduce reward $r$ earlier, say here?}

\subsection{Examples of application scenarios}
\label{subsec:intro_app}

To further contextualize the discussion, we now give several concrete examples of the deployment risks and how the MDR/SDR translate into practical guarantees in four representative applications. Readers interested in methodology may skip the rest of the section without missing key information.

\vspace{-0.5em}
\paragraph{Drug discovery with low risk.} 
Early stages of drug discovery aims to select promising drug candidates from a large library. While traditional approaches rely on exhaustive physical screening to evaluate their properties~\citep{szymanski2011adaptation,macarron2011impact}, it is increasingly popular to rely on AI predictions to shortlist drug candidates~\citep{carracedo2021review,dara2022machine}. In this case, $X$ is the physical/chemical structure of a drug candidate, and a model $f$ generates an imperfect prediction $f(X)$ for the unknown property of interest $Y$. Here, a decision to trust $f$ for a test instance $X_{n+j}$ means selecting it for future development, where a false positive may incur a waste of subsequent cost $\ell(X_{n+j},Y_{n+j}) \in [0,1]$. In~\cite{jin2023selection,Bai_Tang_Xu_Svetnik_Khalili_Yu_Yang_2024}, the risk is binary $\ell(X_{n+j},Y_{n+j})=\ind\{Y_{n+j}\leq c\}$ where $c$ is a known threshold.
Controlling the TDR~\eqref{eq:def_total_risk} limits the total expected cost of false leads. 
Controlling the SDR~\eqref{eq:def_sel_risk} implies that the average cost per selected compound is limited. 

\vspace{-0.5em}
\paragraph{Finding small-error predictions.} 
For a regression model $f\colon \cX\to \RR$, practitioners may rely on its predictions only when sufficiently accurate, for tasks such as auto-labeling and  decision support. In this case, 
a natural risk is $\cL(f,X,Y) = \ind\{|Y-f(X)|>c\}$ for a fixed tolerance $c>0$, or $\cL(f,X,Y) = |Y-f(X)|^2$ for mean squared error (MSE).  
In the former case, controlling the MDR~\eqref{eq:def_mgn_risk} limits the probability of deploying a high-error case, while controlling the SDR~\eqref{eq:def_sel_risk} limits the fraction of high-error cases among deployed ones. 
With the MSE risk, limiting the SDR~\eqref{eq:def_sel_risk} controls the average MSE among the deployed units.

\vspace{-0.5em}
\paragraph{Deploying LLMs with low semantic error.} 
In using LLMs for radiology report generation, the input $X$ is a medical image, and the output $f(X)$ is a natural-language report describing the findings from the image. Since the report will be handed to clinicians to make medical decisions, it is useful to control risks in cases where LLM reports are adopted. In \cite{gui2024conformal}, the unknown label $Y$ is a human-expert report, and $\cL(f,X,Y)$ is a binary risk which equals $1$ if $f(X)$ differs from $Y$ based on CheXbert labels~\citep{smit2020chexbert}. More generally, $\cL(f,X,Y)$ may measure semantic distances  or number of deviations in key findings between the reports. 
Here, controlling the SDR~\eqref{eq:def_sel_risk} below a expert-defined error rate would be useful for ensuring that the LLM models are deployed only when they are comparable to experts.

 \vspace{-0.5em}
\paragraph{Selecting accurate diagnosis with few follow-ups.} 
For multi-class diagnosis such as a disease subtype $Y\in [K]$, a foundation model $f$ produces probability estimates $f(X,k)$ for each label $k$, leading to a ranking of labels $f(X,[1])\geq f(X,[2])\geq \dots \geq f(X,[K])$, where $([1],[2],\dots,[K])$ is a permutation of $(1,\dots,K)$. 
Clinical workflows may proceed down this list with confirmatory tests until the true label is reached. 
% To use a model in disease diagnosis, clinicians can go down the list of ranked labels, conduct follow-up examinations, and figure out the exact diagnosis. 
To expedite the process, it is useful to only use high-quality predictions where one does not need to go too far down the list to reach the correct label (an extreme case is when the top-1 prediction is correct). 
One may define $\cL(f,X,Y) = \frac{1}{K}\sum_{k=1}^K \ind\{f(X,k)\leq f(X,Y)\}$, the number of steps needed before reaching the true label. Then, controlling the TDR~\eqref{eq:def_total_risk} finds units needing fewer than $\alpha K\cdot m$ follow-up steps in total, while controlling the SDR~\eqref{eq:def_sel_risk} means \emph{each} deployed unit needs $\alpha  K$ follow-up steps. Arguably, the SDR is more sensible for AI integration: we trust AI only when it improves efficiency upon traditional human inspections.

\subsection{Related work}

\paragraph{Selective prediction.} This paper is motivated by the philosophy of selective prediction, that is, we only deploy a model when confident  and control errors on the deployed cases~\citep{chow2009optimum,el2010foundations,geifman2017selective,mozannar2020consistent}. Much of this literature addresses classification settings with asymptotic guarantees for selective risk. This is related to, and expanded by our SDR notion (see~\cite{gui2024conformal} for a discussion on the distinctions for binary risks). We contribute to this literature from the conformal inference perspective. Our methods provide both selective and marginal guarantees, work in finite sample, and address  general, continuously-valued risks. 

\vspace{-0.5em}
\paragraph{Selective conformal inference.} 
Methodologically, SCoRE is closest to the work on selective inference and multiple testing in prediction problems via conformal inference~\citep{jin2023selection,jin2023model,bai2024optimized,huo2024real,lee2025selection,nair2025diversifying,gui2025acs,gazin2025selecting,liu2025online,huang2025selective}. As we shall discuss in Section~\ref{subsec:warmup}, this literature builds on conformal p-values to control a binary risk, adapts them to selective settings. The key technical distinction is that we target continuous risks with e-values instead of p-values. 
Other works using conformal prediction to address selective prediction include~\cite{fisch2022calibrated,sokol2024conformalized}, yet they focus on distinct aspects like calibration or directly using prediction sets, instead of valid error control among selected cases. Finally, our work connects to the line of work on conformal risk control (CRC)~\citep{angelopoulos2022conformal}and learn-then-test (LLT)~\citep{angelopoulos2025learn}. These methods address related marginal or selective risk notions, primarily for binary risks. Our setting differs in targeting continuous risks via exact calibration with e-values. In particular, our SDR variant targets a selective criterion that avoids uniform concentration (over a grid) needed there, while our MDR variant provides an e-value perspective that connects to CRC and enables a unified analysis of finite-sample validity, covariate shift, and asymptotic optimality (see more detailed discussion later). 
% Specifically, they addresses selective risk notions similar to the SDR (for binary risk); with general risks, the selective risk is non-monotone in the cutoff, so applying their LTT approach needs to leverage uniform concentration-type arguments instead of exact calibration via multiple testing ideas. Our MDR variant can be viewed as an instance of  conformal risk control (which deals with marginal risks), whereas our framework provides a unified view from e-values and offers asymptotic and optimality analysis. 

\vspace{-0.5em}
\paragraph{Conformal inference with e-values.} The e-values, as a parallel to p-values, have attracted recent interest in hypothesis testing and related tasks due to advantages such as compatibility with dependence~\citep{vovk2021values,wang2022false,waudbysmith2021estimating,ramdas2024hypothesis}. E-values in conformal prediction date back to~\cite{vovk2025conformal}, and have   attracted recent attention~\citep{balinsky2024enhancing,koning2023post,gauthier2025values,koning2025optimal,gauthier2025adaptive}. Distinct from other works that harness advantages of e-values like any-time validity, we leverage  e-values is to control the expectation of unknown risks (though our construction of risk-adjusted e-values is related to the soft-rank e-values~\citep{gauthier2025adaptive}; see discussion in Section~\ref{subsec:marginal_construction}). Finally, this work generalizes conformal selection methods that can be interpreted via e-values, yet with a different goal of controlling risks (see Section~\ref{subsec:risk_adjusted_e_values} for detailed discussion).

\vspace{-0.5em}
\paragraph{Statistical hypothesis testing.} This work is deeply connected to classical statistical hypothesis testing. While most of the works focus on binary type-I error control of rejecting a null hypothesis, there are methods that incorporate ``weights'' for the hypotheses in defining the type-I error~\citep{benjamini1997multiple,roeder2009genome,basu2018weighted,benjamini2017weighted}. Our risks $L_{n+j}$ in both error metrics can be viewed as unknown, random weights, and we provide a solution with e-values, which might be useful for other problems where a similar structure is present. While the connection is not straightforward, this relates to~\cite{grunwald2024beyond} which uses e-values to control the downstream costs in distinct test decisions. 
Another related line of work considers selecting multiple families of hypotheses, so that the average risk (such as the FDP in the family) is controlled among the selected families~\citep{heller2009flexible,sun2011multiple,benjamini2014selective}; while we use quite different techniques, our methods may be applicable in their setting if knowledge of the risk is available in some ``calibration'' families.

\section{General strategy: testing with risk-adjusted e-values}
\label{sec:general_method}

This section presents the  high-level strategy for controlling the two metrics. Section~\ref{subsec:warmup} warms up via an existing framework with binary risk control. Section~\ref{subsec:risk_adjusted_e_values} introduces the concept of risk-adjusted e-values, and Section~\ref{subsec:general_strategies} shows how any risk-adjusted e-values yield MDR and SDR control. 

\subsection{Warm-up: conformal p-value for binary risk}
\label{subsec:warmup}

We briefly review the binary-risk setting to motivate our framework.  
Conformal selection methods (e.g.,~\citet{jin2023selection,jin2023model,bai2024optimized} and references therein) address the problem of identifying sufficiently large outcomes $Y_{n+j}>c$ for a pre-specified constant $c>0$ while controlling a binary error $\ind\{Y_{n+j}\leq c\}$. 
\cite{jin2023selection} formalizes this problem as testing a random hypothesis $H_j\colon Y_{n+j}\leq c$, where rejecting $H_j$ implies declaring a large outcome. 
They leverage conformal prediction~\citep{vovk2005algorithmic} to construct  conformal p-values $\{p_j\}$ obeying 
\$
\PP(Y_{n+j}\leq c,~ p_j\leq t)\leq t,\quad \text{for all} ~ t\in [0,1].
\$
This resembles the null property of valid p-values in classical hypothesis testing. 
However, the null event is random and is not conditioned upon; instead, it appears jointly with the p-value in the probability statement.
With this in hand, rejecting $H_j$ when $p_j\leq \alpha$ naturally leads to the control of the binary MDR, as $\EE[\ind\{Y_{n+j}\leq c\}\ind\{p_j\leq \alpha\}]\leq \alpha$. 
In addition, \cite{jin2023selection} show that, when the calibration and test samples are exchangeable, passing multiple p-values $\{p_j\}_{j=1}^m$ to the Benjamini-Hochberg procedure~\citep{benjamini1995controlling} at level $\alpha\in(0,1)$ produces a selection set $\cR$ with FDR control:
\$
\EE\Bigg[ \frac{\sum_{j=1}^m\ind\{Y_{n+j}\leq c\} \ind\{j\in \cR\}}{1\vee|\cR|} \Bigg] \leq\alpha,
\$
which coincides with~\eqref{eq:def_sel_risk} when taking $L_{n+j}=\ind\{Y_{n+j}\leq c\}$. 

Conformal selection draws upon classical hypothesis testing to control the expectation of a binary risk  by the tail probability of a uniformly distributed p-value. 
However, tail probability is not a natural instrument for quantifying and controlling the expectation  of continuous risks. 
This motivates the use of e-values whose validity is defined through expectation~\citep{vovk2021values}. 
The remaining challenge is then to construct e-values that remain valid when the ``null'' is an unknown random risk.

\subsection{Risk-adjusted e-values}
\label{subsec:risk_adjusted_e_values}

We now introduce our key technical tool inspired by e-values~\citep{vovk2021values}. Specifically, for each test unit, we construct a non-negative random variable obeying the following definition. 

\begin{definition}[Risk-adjusted e-value]\label{def:eval}
    For the random risk $L_{n+j}=\cL(f,X_{n+j},Y_{n+j})$, we say a random variable $E_{n+j}$ is a \emph{risk-adjusted e-value} if $E_{n+j}\geq 0$ almost surely and  $\EE[E_{n+j}L_{n+j}]\leq 1$. 
\end{definition}

The concrete constructions of risk-adjusted e-values based on the scores $s(X_i)$ and observed risks $\cL(f,X_i,Y_i)$ will be tailored to each error metric and introduced later.

Similar to the null property of conformal p-values, the defining property of risk-adjusted e-values characterizes the joint behavior of risks and e-values.
This joint control naturally allows these e-values to be combined with hypothesis testing procedures to produce binary trust decisions $\{\hat\psi_{n+j}\}$. 
Intuitively, a large value of $E_{n+j}$ provides evidence that the risk $L_{n+j}$ is small due to the validity condition $\EE[L_{n+j}E_{n+j}]\leq 1$.  Definition~\ref{def:eval} generalizes the notion of e-values in statistical hypothesis testing: when testing a deterministic null hypothesis $H_0$, a random variable $E\geq 0$ is an e-value if $\EE[E]\leq 1$ under $H_0$ (that is, the risk being $\ind\{H_0\text{ is true}\}$), so that a large value of $E$ suggests evidence against the null~\citep{ramdas2024hypothesis}.

% \yingcomment{mention conformal selection e-value and optCS e-value.}

Even closer to us is the e-value perspective of conformal selection~\citep{jin2023selection}.
While the original method relies on p-values, several works construct e-values $e_{n+j}$ obeying $\EE[e_{n+j}\ind\{Y_{n+j}\leq c\}]\leq 1$ for controlling the FDR in online selection~\citep{xu2024online}, promoting selection diversity~\citep{nair2025diversifying}, and addressing hierarchical data~\citep{lee2025selection}. 
Other works that address other issues in conformal p-values, including covariate shift in~\cite{jin2023model} and the model optimization in \cite{bai2024optimized}, can also be interpreted as implicitly using certain e-values obeying such property. 

\subsection{General strategies for MDR and SDR control}
\label{subsec:general_strategies}

Once any risk-adjusted e-values are available, Theorem~\ref{thm:general_mgn_risk} offers a general strategy for deriving trust decisions that control the MDR~\eqref{eq:def_mgn_risk} in finite samples, and Theorem~\ref{thm:general_sel_risk} provides such a strategy for SDR~\eqref{eq:def_sel_risk}. 
% Their proofs are simple and provided below. 

\begin{theorem}\label{thm:general_mgn_risk}
    Suppose $E_{n+j}$ obeys Definition~\ref{def:eval}. Setting the trust decision as $\hat\psi_{n+j} = \ind\{E_{n+j}\geq 1/\alpha\}$ yields the marginal risk control: $\EE[L_{n+j}\cdot \hat\psi_{n+j}]\leq \alpha$. 
\end{theorem}

\begin{proof}[Proof of Theorem~\ref{thm:general_mgn_risk}]
    Since $L_{n+j}\geq 0$ and $E_{n+j}\geq 0$, we have $L_{n+j} \hat\psi_{n+j} = L_{n+j} \ind\{E_{n+j}\geq 1/\alpha\} \leq L_{n+j}E_{n+j} \cdot \alpha$. 
    Taking the expectation gives $\mdr=\EE[L_{n+j} \hat\psi_{n+j}]\leq \alpha$ due to Definition~\ref{def:eval}. 
\end{proof}

% \subsection{General strategy for SDR control}

We apply the e-BH procedure~\citep{wang2022false} to risk-adjusted e-values to control the SDR.

\begin{theorem}\label{thm:general_sel_risk}
    Suppose $\{E_{n+j}\}_{j=1}^m$ obey  Definition~\ref{def:eval}.  Let $\hat\psi_{n+j}=1$ if and only if $j$ is selected by the e-BH procedure applied to $\{E_{n+j}\}_{j=1}^m$ at level $\alpha\in (0,1)$. 
    That is,  $\hat\psi_{n+j}=\ind\{E_{n+j}\geq m /(\alpha\hat\tau)\}$, where $\hat\tau = \max\{\tau\colon \sum_{j=1}^m \ind\{E_{n+j}\geq m/(\alpha\tau)\}\geq \tau \}$. 
    Then, it holds that $\EE\big[  \sum_{j=1}^m L_{n+j} \cdot \hat\psi_{n+j} / (1\vee \sum_{j=1}^m \hat\psi_{n+j}) \big]  \leq \alpha$. 
\end{theorem}

\begin{proof}[Proof of Theorem~\ref{thm:general_sel_risk}]
    By the definition of $\hat\tau$, we have 
    \$
    \sdr = \EE\Bigg[ \frac{\sum_{j=1}^m L_{n+j} \ind\{E_{n+j}\geq m/(\alpha\hat\tau)\}}{1\vee \sum_{j=1}^m   \ind\{E_{n+j}\geq m/(\alpha\hat\tau)\} } \ind\{\hat\tau>0\}  \Bigg] 
    \leq \sum_{j=1}^m \EE\bigg[ \frac{ L_{n+j} \ind\{E_{n+j}\geq m/(\alpha\hat\tau)\}}{\hat\tau }   \bigg]. 
    \$
    Since $L_{n+j}\geq 0$ and $E_{n+j}\geq 0$, we have $ L_{n+j} \ind\{E_{n+j}\geq m/(\alpha\hat\tau)\} \leq L_{n+j}E_{n+j} \cdot \alpha \hat\tau/m$. Therefore, 
    \$
    \sdr \leq \sum_{j=1}^m  \EE\bigg[\frac{L_{n+j}E_{n+j} \cdot \alpha \hat\tau}{\hat\tau \cdot m}\bigg] \leq \frac{\alpha}{m}\sum_{j=1}^m \EE[L_{n+j}E_{n+j}] \leq \alpha,
    \$
    since $E_{n+j}$ obeys Definition~\eqref{def:eval}.
\end{proof}

The remaining task then reduces to constructing valid risk-adjusted e-values; the MDR and SDR control then follow automatically by Theorems~\ref{thm:general_mgn_risk} and~\ref{thm:general_sel_risk}.

The strategy in Theorem~\ref{thm:general_mgn_risk} for controlling the MDR is related to \cite{grunwald2024beyond}, though the connection is not straightforward. There, e-values are used to control risk in classical hypothesis testing where scientists are allowed to derive rules of taking multiple actions---more than just reject or not, each with a known risk. 
In contrast, we use e-values to control unobserved risks in prediction problems. 
Our e-values are also compatible with other techniques for e-values such as multiple testing, and can lead to control of interpretable metrics like the SDR in predictive inference settings.

% !TEX root = main.tex

\section{Marginal risk control with conformal e-values} \label{sec:marginal}

While Section~\ref{sec:general_method} shows that any collection of risk-adjusted e-values can lead to valid, finite-sample control of the MDR or SDR, the power or utility of the procedure depends critically on the quality of the e-values. Poorly designed e-values can result in an excessively large number of unnecessary abstentions.

In this section, we study the concrete construction of risk-adjusted e-values tailored for MDR control based on conformal inference and data exchangeability, thereby completes the strategy in Theorem~\ref{thm:general_mgn_risk}. Owing to the distinct testing structure, the corresponding e-value construction for SDR control differs and is presented in Section~\ref{sec:selective}.
We then discuss an efficient computation shortcut that produces trust decisions directly, bypassing the explicit numerical search for e-values. Finally, we derive optimal choices within the proposed family of e-values.

\subsection{Constructing e-values} \label{subsec:marginal_construction}

Recall that we have a pre-trained score function $s\colon \cX\to [0,1]$ that predicts $\cL(f,X,Y)$ or a related notion of uncertainty. 
The construction below produces an e-value---and hence the deploy/abstain decision---based on the magnitude of score $s(X_{n+j})$. 
Let the observed calibration risks be $L_i = \cL(f,X_i,Y_i)$ for $i\in [n]$. 

Fix any constant $\gamma \in (0,1)$. We define 
\@\label{def:e-value-mgn}
E_{\gamma, n+1} = \inf_{\ell \in[0,1]} ~\Bigg\{ \frac{(n+1)\cdot \ind\{s(X_{n+1})\leq t_\gamma(\ell)\}}{\sum_{i=1}^n L_i \ind\{s(X_{i})\leq t_{\gamma}(\ell)\} + \ell\ind\{s(X_{n+1})\leq t_\gamma(\ell)\} } \Bigg\}.
\@
Here $\ell\in[0,1]$ is a candidate value of the unknown risk $L_{n+1}$, and $t_\gamma(\ell)$ is a data-dependent threshold chosen so that an empirical risk estimate does not exceed $\gamma$. Concretely,
\@\label{eq:def_F_single}
t_{\gamma} (\ell) = \max \big\{t \in \cM\colon \textrm{F}(t;\ell)\leq \gamma \big\},\quad 
\textrm{F}(t;\ell) =  \frac{\sum_{i=1}^n L_i \ind\{s(X_{i})\leq t\} + \ell\ind\{s(X_{n+1})\leq t\} }{ n+1 }.
\@
Here we define $\cM:=\{s(X_i)\}_{i=1}^{n+1}$. By convention, $\max \emptyset = -\infty$, and $E_{\gamma, n+1} = 0$ when $\inf_{\ell\in[0,1]}t_{\gamma}(\ell) = -\infty$.
Put differently,  $E_{\gamma,n+1} = \inf_{\ell\in [0,1]} \{ \ind\{s(X_{n+1})\leq t_\gamma(\ell)\} /  \textrm{F}(t_\gamma(\ell);\ell)  \}$. 

\begin{remark}\label{rem:range_marginal}
In~\eqref{def:e-value-mgn}, we take the infimum over the entire range $[0,1]$. In principle, this search domain can be reduced to the values that can be attained, i.e.,   $\ell\in  \RR^+\cap \{\cL(X_{n+1},y)\colon y\in \cY\}$. While our computation strategies and numerical experiments are tied to~\eqref{def:e-value-mgn}, such a replacement may lead to larger e-values and faster computation. For binary risk, this reduction allows computing $E_{\gamma,n+1}$ by simply plugging in $\ell=1$. 
\end{remark}

The SCoRE procedure for MDR control is  summarized in Algorithm~\ref{alg:mdr}. (In practice, we recommend setting $\gamma=\alpha$; see Section~\ref{subsec:mgn_computation}.)

\begin{algorithm}
    \small
    \captionsetup{font=small}
    \caption{SCoRE-MDR}
    \label{alg:mdr}
\begin{algorithmic}[1]
    \REQUIRE{Labeled data $\{(X_i, Y_i)\}_{i=1}^n$, test data $X_{n+1}$, pre-trained score function $s(\cdot)$, MDR target $\alpha \in (0,1)$.} \\[0.5ex]
    \STATE Compute calibration risks $L_i = \cL(f, X_i, Y_i)$ for $i = 1, \dots, n$.
    \STATE Obtain the scores $\cM:=\{s(X_i)\}_{i=1}^{n+1}$.
    \STATE Compute $E_{\alpha,n+1}$ as in~\eqref{def:e-value-mgn}.
    \STATE Compute $\hat\psi_{n+1} = \ind\{E_{\alpha,n+1}\geq 1/\alpha\}$.
    \ENSURE{Deployment decision $\hat\psi_{n+1}$.}
\end{algorithmic}
\end{algorithm}

Theorem~\ref{thm:single} confirms that $E_{\gamma,n+1}$ obeys Definition~\ref{def:eval}, whose proof is in Appendix~\ref{app:subsec_thm_single}. Consequently, Algorithm~\ref{alg:mdr} controls the MDR below $\alpha$ in finite samples. Importantly, Theorem~\ref{thm:single} only relies on exchangeability among data, without requiring  the score function $s$ to accurately predict the risk.

\begin{theorem}\label{thm:single}
    Suppose $\{(X_i,Y_i)\}_{i=1}^{n+1}$ are exchangeable. Then, $\EE[L_{n+1} E_{\gamma,n+1}  ]\leq 1$ for any fixed $\gamma \in(0,1)$.
\end{theorem}

The intuition of~\eqref{def:e-value-mgn} is as follows. 
Should $L_{n+1}$ be known, any random variable of the form 
\$
\frac{(n+1)\cdot L_{n+1} A_{n+1}}{\sum_{i=1}^n L_i A_i + L_{n+1}A_{n+1}},
\$ 
has expectation equal to $1$ if $\{(L_i,A_i)\}_{i=1}^{n+1}$ are exchangeable. 
Thus, we can define $E_{n+1}:= \frac{(n+1)\cdot A_{n+1}}{\sum_{i=1}^n L_i A_i + L_{n+1}A_{n+1}}$ for some  $\{(L_i,A_i)\}_{i=1}^{n+1}$ that are exchangeable, which is a risk-adjusted e-value obeying $\EE[E_{n+1}L_{n+1}]\leq 1$.  

While the choice of  $\{A_{i}\}$ can be quite flexible, we set $A_{i} = \ind\{s(X_{i})\leq T\}$, where $T$ is a random variable that is permutation invariant to $\{(X_i,Y_i)\}_{i=1}^{n+1}$. 
This is because in applying Theorem~\ref{thm:general_mgn_risk} to obtain MDR control, 
a crucial inequality is $ \ind\{E_{n+1}\geq 1/\alpha\}\leq \alpha \cdot E_{n+1}$, which is tight only if $E_{n+1}$ takes value in $\{0,1/\alpha\}$. This motivates the ``one-hot'' form of the e-value. 
Since $L_{n+1}$ is unobserved, 
we construct a conservative e-value by taking the smallest value over all the possible values of $L_{n+1}$ via $\ell \in [0,1]$. 
Finally, the $t_\gamma(\ell)$ in~\eqref{eq:def_F_single} can be viewed as an empirical calibration: we note that $\textrm{F}(t;L_{n+1})$ estimates $\EE[L\ind\{s(X)\leq t\}]$ in a way that preserves exchangeability.

\begin{remark}
  We may generalize $s(x)$ to any label-dependent scores  
   $V\colon \cX\times\cY\to \RR$. Define
    \@\label{def:e-value-mgn-y}
E_{\gamma, n+1}^{\textnormal{general}} = \inf_{y\in \cY} ~\Bigg\{ \frac{(n+1)\cdot \ind\{V(X_{n+1},y)\leq t_\gamma(y)\}}{\sum_{i=1}^n L_i \ind\{V(X_{i},Y_i)\leq t_{\gamma}(y)\} + \cL(X_{n+1},y)\ind\{V(X_{n+1},y)\leq t_\gamma(y)\} } \Bigg\},
\@
where $t_{\gamma} (y) = \max \big\{t \in \cM\colon \mF^{\textnormal{general}}(t;y)\leq \gamma \big\}$, and 
\@\label{eq:def_F_single_y}
\mF^{\textnormal{general}}(t;y) =  \frac{\sum_{i=1}^n L_i \ind\{V(X_{i},Y_i)\leq t\} + \cL(X_{n+1},y)\ind\{V(X_{n+1},y)\leq t\} }{ n+1 }.
\@
Then, the definition in~\eqref{def:e-value-mgn} is a special case with $V(x,y)=s(x)$. One can still follow the proof idea of Theorem~\ref{thm:single} outlined above to show that $\EE[E_{\gamma,n+1}^{\textnormal{general}}L_{n+1}]\leq 1$ under exchangeability. However, $E_{\gamma,n+1}^{\textnormal{general}}\geq 1/\alpha$ requires $V(X_{n+1},y)\leq t_\gamma(y)$ for all $y\in \cY$, which might be harder to satisfy in general. The computational and statistical benefits of this definition are beyond the scope of the current work. 
\end{remark}

\subsection{Efficient computation}
\label{subsec:mgn_computation}

The definition of $E_{\gamma,n+1}$ in~\eqref{def:e-value-mgn} involves an infimum over a continuous  variable $\ell \in [0,1]$. 
Fortunately, for MDR control we only need the thresholding decision $\ind\{E_{\gamma,n+1}\geq 1/\alpha\}$, not the exact value of  $E_{\gamma,n+1}$. 
The next proposition shows how to streamline the computation, whose proof is in Appendix~\ref{app:subsec_single_equiv}.  

\begin{prop}\label{prop:single_equiv}
    For $\gamma \leq \alpha$, we have 
    \$
    \ind\{E_{\gamma,n+1}\geq 1/\alpha\} = \ind \bigg\{  \frac{1+\sum_{i=1}^n L_i \ind {\{s(X_i)\leq s(X_{n+1})\}}}{n+1} \leq \gamma \bigg\}.
    \$
    For $\gamma > \alpha$, we have
    \$
    \ind\{E_{\gamma,n+1}\geq 1/\alpha\} = \ind \bigg\{  &\frac{1+\sum_{i=1}^n L_i \ind {\{s(X_i)\leq s(X_{n+1})\}}}{n+1} \leq \gamma, ~ \\
    &\quad \text{and} \quad \frac{\ell +\sum_{i=1}^n L_i \ind  \{s(X_i)\leq t\} }{n+1} \notin (\alpha,\gamma ],~~\forall t\in \cM, \ell\in [0,1]
    \bigg\},
    \$
    where $\cM = \{s(X_{i})\}_{i=1}^{n+1}$ is the set of all calibration and test scores.
    % In particular, when $\gamma = \alpha$, we have
    % \$
    % \ind\{E_{\alpha,n+1}\geq 1/\alpha\} = \ind \bigg\{  \frac{1+\sum_{i=1}^n L_i \ind {\{\hat\mu(X_i)\leq \hat\mu(X_{n+1})\}}}{n+1} \leq \alpha \bigg\}.
    % \$
\end{prop}

% \yingcomment{
% Why boosting does not control MDR here?
% Boosting: $\hat\psi_{n+1}=\ind\{E_{\gamma,n+1}/\xi\geq 1/\alpha\} = \ind\{\xi \leq \alpha\cdot E_{\gamma,n+1}\}$. 
% Error control: $\EE[L_{n+1}\hat\psi_{n+1}] = \EE[L_{n+1}\ind\{\xi \leq \alpha\cdot E_{\gamma,n+1}\}] = \EE[L_{n+1}\min\{1, \alpha E_{\gamma,n+1}\}]\leq \alpha \EE[L_{n+1}E_{\gamma,n+1}]$}
% \tiancomment{I think any value of $\gamma$ would work for the proposition (but since our goal is to control at $\alpha$, we should finally set $\gamma = 1/\alpha$); I also tried to simplify the proof. By the way, the LHS in RHS indicator seems also to be the same as the conformal p-value for single selection when the score is clipped. Is this also a specific case under the conformal risk control procedure?} 

\begin{remark}
    Proposition~\ref{prop:single_equiv} justifies setting the parameter $\gamma$ equal to the nominal level $\alpha$. When $\gamma < \alpha$, the proposition implies  $\ind\{E_{\gamma, n+1} \geq 1/\alpha\} \leq \ind\{E_{\alpha, n+1} \geq 1/\alpha\}$, and  SCoRE always selects less frequently than $\gamma = \alpha$.
    On the other hand, if $\gamma > \alpha$, one must impose an extra thresholding condition that almost always fails in practice, yielding asymptotically zero power (Theorem \ref{thm:opt_mdr}) under standard regularity conditions.
\end{remark}

% \tiancomment{TODO: try something like $\gamma = 1.1\alpha$.}

Proposition~\ref{prop:single_equiv} allows us to connect the MDR control instantiation of SCoRE with existing works in conformal inference and selective inference. 
First, SCoRE with binary risks reduces to the conformal selection    framework~\citep{jin2023selection} discussed in Section~\ref{subsec:warmup}. 
To select test instances with responses exceeding a specific threshold $Y_{n+1} > c$, it constructs p-values
\@\label{eq:def_conformal_pval}
p_{n+1} = \frac{1+\sum_{i=1}^n\ind\{V(X_i, Y_i) \leq V(X_{n+1}, c)\}}{n+1},
\@
where $V: \mathcal{X} \times \mathcal{Y} \to \mathbb{R}$, $V(x, y) = \infty\ind\{y > c\} + s(x)$ is the clipped nonconformity score\footnote{We flipped the sign of the scores in~\cite{jin2023selection} to be consistent with the current setup.}. 
One can check that defining the risks as $L_i = \ind\{Y_i \leq c\}$, Proposition~\ref{prop:single_equiv} yields $\ind\{E_{\alpha,n+1}\geq 1/\alpha\}  =  \ind\{p_{n+1} \leq \alpha\}$.
% \$
%     \ind\{E_{\alpha,n+1}\geq 1/\alpha\}  &= \ind \bigg\{  \frac{1+\sum_{i=1}^n \ind\{Y_i \leq c\} \ind {\{s(X_i)\leq s(X_{n+1})\}}}{n+1} \leq \alpha \bigg\}   = \ind\{p_{n+1} \leq \alpha\}.
% \$
Thus, SCoRE is procedurally equivalent to conformal selection for one hypothesis with type-I error control. 
% where $V: \mathcal{X} \times \mathcal{Y} \to \mathbb{R}$, $V(x, y) = \infty\ind\{y > c\} + \hat\mu(x)$ is the clipped nonconformity score\footnote{We note that the original clipped score in CS is $V(x, y) = M \cdot\ind\{y > c\} - \hat\mu(x)$, where $M$ is a large constant serving as the relaxation of infinity. Here, we directly use $M = \infty$, and also flip the sign of the predicted value in the score since the predictor $\hat\mu$ are defined differently in the two methods. In SCoRE, $\hat\mu$ estimates the risk $L_{n+1}$, while in CS, it estimates the probability of $Y_{n+j} > c$, which is $1 - L_{n+1}$.\label{ftnt:nonconformity_score}}, and $p_1$ denotes the conformal p-value in CS. 
% Since only one instance is considered in this setting, the BH procedure in CS selects the test instance if $p_1 \leq \alpha$, which exactly corresponds to the condition $E_{\alpha,n+1}\geq 1/\alpha$ required by SCoRE. Thus, under this setup, the two procedures are functionally equivalent.

Furthermore, our e-value is related to conformal risk control~\citep{angelopoulos2022conformal} with risk functions $L_i(\lambda) := L_i\ind\{-s(X_i) \geq \lambda\}$, and $\lambda\in \Lambda =[-1, 0]$. Given any bounded, non-increasing risk function,  conformal risk control determines a parameter $\hat\lambda \in \Lambda$ so that the test risk $\EE[L_{n+1}(\hat\lambda)]$ is controlled. With this risk, it yields
$
    \hat\lambda = \inf \big\{ \lambda \in \Lambda: \frac{1 + \sum_{i=1}^n L_i(\lambda)}{n+1} \leq \alpha \big\},
$
We observe that $\hat\psi_{n+1}=1$ given by Theorem~\ref{thm:general_mgn_risk} (with $\gamma=\alpha$) is equivalent to  $-s(X_{n+1}) \geq \hat\lambda$. That is, SCoRE opts to deploy a unit if and only if it is a risk-controlled decision. We defer the detailed explanation of this fact to Appendix~\ref{app:subsec_connect_crc}. 

\subsection{Asymptotics and optimality}
\label{subsec:asymptotics_mdr}

While the MDR control holds regardless of the score function $s(\cdot)$, 
the usefulness of the procedure depends on the average number of deployable instances, and more generally, on the downstream reward from deploying the model.  
To navigate this choice, we define a general notion of power:
\@\label{eq:def_mdr_power}
\text{Power} := \EE\big[ r(X_{n+1},Y_{n+1})\hat{\psi}_{n+1}  \big],
\@
where $r\colon \cX\times\cY\to [0,1]$  encodes a bounded ``reward'' of deploying the model on a test instance which may depend on the unknown label. 
It may also depend on the model $f$ but we omit this   for simplicity.
When $r(x,y)\equiv 1$, the power is  the probability of deployment. 
This flexibility allows practitioners to prioritize deployment on more valuable instances. For example, drug discovery scientists may assign a high reward to ``novel'' instances and maximize the reward in the selected candidates  while controlling the total wastage.

Theorem~\ref{thm:opt_mdr} establishes a ``Neyman-Pearson lemma''-like rule~\citep{lehmann1986testing} for asymptotically optimal scoring functions that maximizes~\eqref{eq:def_mdr_power} subject to MDR control.  Its proof is in Appendix~\ref{app:subsec_opt_mdr}. Throughout, we treat $f(\cdot)$ and $s(\cdot)$ as fixed while taking the calibration sample size $n$ to infinity.
 
\begin{theorem}\label{thm:opt_mdr}
    Suppose  $\{(X_i,Y_i)\}_{i=1}^{n+1}$ are i.i.d.~from some unknown distribution $P$. Define $\mathrm{F}^*(t) := \EE[\cL(f,X,Y)\ind\{s(X)\leq t\}]$ for an independent copy $(X,Y)\sim P$, and $f,s$ are viewed as fixed. Define $t^*:= \sup\{t\in [0,1]\colon \mathrm{F}^*(t)\leq \gamma\}$. 
    % Assume $\mathrm{F}^*(t)$ is continuous in $t\in [0,1]$ and non-constant in a small neighborhood around $t^*$, and the distribution of $\hat\mu(X)$ is non-atomic. 
    Suppose the distribution of $s(X)$ is non-atomic, and $\mathrm{F}^*(t)$ is strictly increasing at  $t^*$. 
    Then  the following holds:
    \begin{enumerate}[label=(\roman*).]
        \item As  $n\to \infty$, $\sup_{\ell\in [0,1]}|t_\gamma(\ell) - t^*|~\asto~ 0$.
        \item $\lim_{n\to\infty}\textnormal{Power}=\EE[r(X_{n+1},Y_{n+1})\ind\{s(X_{n+1})\leq t^*\}]$ if $\gamma\leq\alpha$, and $\lim_{n\to\infty}\textnormal{Power}=0$ if $\gamma>\alpha$. Furthermore, for a fixed $s(\cdot)$, the asymptotic power is optimized at $\gamma=\alpha$. 
        \item Fix $\gamma=\alpha$. Define $l(x):=\EE[\cL(f,X,Y)\given X=x]$ and $r(x):=\EE[r(X,Y)\given X=x]$. Suppose $r(X)>0$ a.s., and the distribution of $l(X)/r(X)$ is non-atomic. Then, the asymptotic power is optimized at any $s(x)$ that is strictly increasing in $l(x)/r(x)$. 
    \end{enumerate} 
\end{theorem}

With a constant reward, Theorem~\ref{thm:opt_mdr} suggests using standard estimators of the conditional prediction error. For example, in multi-class classification, one may be interested in whether the top-1 prediction (i.e., the label with the highest predicted probability) equals the true class, thereby defining $\cL(f,x,y) = \ind\{ y\neq \argmax_{y'} f(x,y')\}$ where $f(x,y)$ is the predicted probability of label $y$. Letting $\hat{y} = \argmax_{y'} f(x,y')$, a natural estimator for $l(x)$ is then $\sum_{y'\neq \hat{y}}f(x,y') = 1- f(x,\hat{y})$. In regression tasks with point prediction $f(x)$, it is natural to consider the mean squared error (MSE) $\cL(f,x,y) = (y-f(x))^2$, in which case  $s(x)$ should estimate the conditional MSE  $\EE[(Y-f(x))^2\given X=x]$. 

When the reward is non-constant,  Theorem~\ref{thm:opt_mdr} implies that the score function $s(X)$ should aim to preserve the ranking of the risk-to-reward ratio $l(x)/r(x)$. It changes the choice of the optimal score (compared with that for a constant reward) only when dividing by $r(x)$ substantially changes the ranking of $l(x)$ alone. We shall see  that this seems to rarely happen in real datasets, but our simulations do find some settings where the optimal scores under constant/non-constant rewards make a difference in the final decisions. 

% this is intuitively very similar to a continuous knapsack problem, where a greedy approach monotonous in l/r gives optimal solution
% just a random thought: when the distribution of test points are different, the optimal strategy (max reward subject to TDR) should assign different alpha to each test point? first sort by l/r, and
% when the budget is below alpha*m, just assign 1 (always deploy). otherwise assign the remaining value (use sdr control)?

% !TEX root = main.tex

\section{Selective risk control with conformal e-values} \label{sec:selective}

This section provides a construction of risk-adjusted e-values tailored to SDR control, which completes the procedure in Theorem~\ref{thm:general_sel_risk}. 
The key distinction from MDR control is that SDR concerns the average risk among selected instances (this notion is closer to the standard ideas in selective prediction~\citep{geifman2017selective}). Accordingly, the e-values are designed to integrate effectively with the e-BH filter. Section~\ref{subsec:construct_eval} presents the construction, along with an efficient algorithm for e-value computation that avoids grid search and runs in quadratic time. For multiple testing with the e-BH filter, Section~\ref{subsec:boosting} introduces a boosting strategy. Finally, Section~\ref{subsec:asymptotics_sdr} characterizes the asymptotically optimal choice of score.

% \subsection{General method}

% Assume access to non-negative variables $\{E_{\gamma, n+j}\}_{j=1}^m$ obeying $\EE[Y_{n+j}E_{\gamma, n+j}]\leq 1$ for all $j\in [m]$. Let 
% $\cR$ as the output of the eBH procedure~\citep{wang2022false} at nominal level $\alpha$, i.e., 
% \$
% \cR = \{j\colon E_{\gamma, n+j} \geq m/(\alpha k^*) \},\quad \textrm{where}\quad 
% k^* = \max \Big\{k\colon \sum_{j=1}^m \ind\{E_{\gamma, n+j} \geq m/(\alpha k)\} \geq k\Big\}.
% \$
% The validity of this procedure is in Theorem~\ref{thm:sel_general}, whose proof is in Appendix~\ref{app:subsec_thm_sel_general}.  

% \begin{theorem}\label{thm:sel_general}
% If $\EE[Y_{n+j}E_{\gamma, n+j}]\leq 1$ for all $j\in [m]$,  then $\cR$ obeys selective risk control~\eqref{eq:def_sel_risk}. 
% \end{theorem}

% While any e-values obeying the condition $\EE[Y_{n+j}E_{\gamma, n+j}]\leq 1$ achieves selective risk control, the specific construction of e-values is important for achieving satisfactory power. 

\subsection{Construction of e-values} \label{subsec:construct_eval}

We construct e-values for SDR control using the same exchangeability idea as in Section~\ref{sec:marginal}, but with a thresholding rule calibrated to approximate the SDR incurred by selecting low-score test points.

As before, let the calibration risks be $L_i=\cL(f,X_i,Y_i)$ for $i=1,\dots,n$, and let $s\colon \cX\to [0,1]$ be any pre-trained score. Fixing any constant $\gamma>0$, we define 
\@\label{eq:def_eval_sel}
E_{\gamma,n+j} = \inf_{\ell \in [0,1]} \bigg\{ \frac{(n+1)\cdot \ind\{s(X_{n+j}) \leq t_{\gamma,n+j}(\ell)\} }{\ell \ind\{s(X_{n+j}) \leq t_{\gamma,n+j}(\ell)\}+  \sum_{i=1}^n  L_i\ind\{s(X_{i}) \leq t_{\gamma,n+j}(\ell)\}} \bigg\}.
\@
The threshold   $t_{\gamma,n+j}(\ell) = \max\{t \in \cM\colon  \fr_{n+j}(t;\ell) \leq \gamma \}$ is chosen as the largest score cutoff such that a plug-in estimate of the SDR does not exceed $\gamma$, and 
\$
\fr_{n+j}(t;\ell) = \frac{\ell \ind\{s(X_{n+j})\leq t\} + \sum_{i=1}^n L_i\ind\{s(X_i)\leq t\}}{1+\sum_{k\neq j} \ind\{s(X_{n+k}) \leq t\}}\cdot \frac{m}{n+1}.
\$
Here  $\cM = \{s(X_i)\}_{i=1}^{m+n}$ is the empirical calibration and test scores, $\max \varnothing=-\infty$, and we set $E_{\gamma, n+j} = 0$ when $\inf_{\ell \in [0,1]} t_{\gamma, n+j}(\ell) = -\infty$. 
A slightly more conservative yet computationally  efficient version is discussed in Appendix~\ref{app:subsec_alt_evalue}. 
We summarize the entire procedure in Algorithm~\ref{alg:sdr}.

\begin{algorithm}
    \small
    \captionsetup{font=small}
    \caption{SCoRE-SDR}
    \label{alg:sdr}
\begin{algorithmic}[1]
    \REQUIRE{Labeled data $\{(X_i, Y_i)\}_{i=1}^n$, test data $\{X_{n+j}\}_{j=1}^m$, pre-trained score $s$, SDR target $\alpha \in (0,1)$, constant $\gamma>0$.} \\[0.5ex]
    \STATE Compute calibration risks $L_i = \cL(f, X_i, Y_i)$ for $i = 1, \dots, n$.
    \STATE Obtain the scores $\cM:=\{s(X_i)\}_{i=1}^{n+m}$.
    \STATE Compute $E_{\gamma,n+j}$ as in~\eqref{eq:def_eval_sel} (or the conservative version in Appendix~\ref{app:subsec_alt_evalue}) for $j=1,\dots,m$.
    \STATE Compute $\cR$ as the selection set of the eBH procedure applied to $\{E_{\gamma,n+j}\}_{j=1}^m$ at level $\alpha$.
    \ENSURE{Deployment decision $\hat\psi_{n+j} = \ind\{j\in \cR\}$.}
\end{algorithmic}
\end{algorithm}

Theorem~\ref{thm:evalue_sel} establishes the validity of $E_{\gamma,n+j}$ as a risk-adjusted e-value, whose proof is in Appendix~\ref{app:subsec_thm_sel_eval}. As a consequence, the output of Algorithm~\ref{alg:sdr} achieves finite-sample SDR control per Theorem~\ref{thm:general_sel_risk}.

% \tiancomment{$\fr$ is similar to $\textrm{F}^{-1}$ in the last section. Therefore should we restrict $\gamma 
% \in (0,1)$ here to match the condition $\gamma > 1$ in the last section? (Also for convenience, take the reciprocal of one to match the other one.)}

\begin{theorem}\label{thm:evalue_sel}
    Assume $\{(X_i,Y_i)\}_{i=1}^{n+m}$ are exchangeable. Then $E_{\gamma,n+j}$ defined in~\eqref{eq:def_eval_sel} obeys $\EE[L_{n+j}E_{\gamma,n+j}]\leq 1$ for any fixed $\gamma > 0$. 
\end{theorem}

Similar to Remark~\ref{rem:range_marginal}, the infimum over $\ell\in [0,1]$ in~\eqref{eq:def_eval_sel} can be restricted to attainable risk values, i.e., $\ell\in \RR^+\cap \{\cL(X_{n+j},y)\colon y\in \cY\}$, leading to sharper e-values when the range of risk is narrower. However, for unified statements, we keep the current definition throughout.

The high-level intuitions of~\eqref{eq:def_eval_sel} are as follows.  $E_{\gamma,n+j}$   conservatively approximates $\frac{(n+1)A_{n+j} }{\sum_{i=1}^n A_iL_i + A_{n+j}L_{n+j}}$ where $A_i = \ind\{s(X_i)\leq T\}$ are random variables such that the $(A_i,L_i)$'s are exchangeable.
Here the ``stopping time'' $T$ is approximated by $t_{\gamma,n+j}(\ell)$, which is carefully designed to align with the e-BH filter. This choice is inspired by the stopping-time interpretation of the BH procedure~\citep{benjamini1995controlling,storey2002direct}  as inverting an empirical-process estimate of the false discovery proportion (FDP). In our context, $\fr_{n+j}(t;\ell)$ estimates the SDR when selecting test units with $s(X_{n+j})\leq t$. Specifically, note 
\$
\fr_{n+j}(t;\ell) \approx \frac{\sum_{i=1}^n L_i\ind\{s(X_i)\leq t\}/n}{\#\{\ell\in[m]\colon s(X_{n+\ell})\leq t\} / m}  \approx \frac{\sum_{\ell=1}^m L_{n+\ell}\ind\{s(X_{n+\ell})\leq t\}/m}{\#\{\ell\in[m]\colon s(X_{n+\ell})\leq t\} / m}  
\$ 
due to exchangeability among data, 
where the right-handed side approximates the SDR for $\psi_{n+j}=\ind\{s(X_i)\leq t\}$. 
Indeed, with binary risk and $\ell=1$, our $\fr_{n+j}(t;\ell)$ reduces to the FDP estimator in~\cite{storey2002direct} in the context of conformal selection~\citep{jin2023selection}.

% \subsection{Efficient computation}
% \label{subsec:sel_eval_compute}

\paragraph{Efficient computation.} Computing $E_{\gamma,n+j}$ in~\eqref{eq:def_eval_sel} necessitates a search over $\ell\in[0,1]$ which can be computationally prohibitive. We develop an efficient computation of $E_{\gamma,n+j}$ in Algorithm~\ref{alg:sel} that avoids such a search. The key idea is to reduce the continuous search over $\ell \in [0,1]$ to a search over the finite set of values attained by $t_{\gamma,n+j}(\ell)$ in $\cM \cup \{-\infty\}$. The proof of Proposition~\ref{prop:mult_equiv} is deferred to Appendix~\ref{app:subsec_mult_equiv}. 

\begin{prop} \label{prop:mult_equiv}
    The output of Algorithm~\ref{alg:sel} equals $E_{\gamma,n+j}$ defined in~\eqref{eq:def_eval_sel}, whose computation complexity is at most $\cO((n+m)m + (n+m)\log(n+m) )$.
\end{prop}

\begin{algorithm}
    \small
    \captionsetup{font=small}
    \caption{Efficient computation of e-values for SDR control}
    \label{alg:sel}
\begin{algorithmic}[1]
    \REQUIRE{Labeled data $\{(X_i, Y_i)\}_{i=1}^n$, test data $\{X_{n+j}\}_{j=1}^m$, pretrained score $s$.} \\[0.5ex]
    \STATE Compute  calibration risks $L_i = \cL(f, X_i, Y_i)$ for $i = 1, \dots, n$.
    \STATE Compute the scores for calibration and test data $\cM:=\{s(X_i)\}_{i=1}^{n+m}$.
    \FOR{$j=1, \dots, m$}  
    \STATE Compute $\bar{\ell}(t)=\frac{\gamma(n+1)}{m} \big(1 + \sum_{\ell \neq j}\ind\{s(X_{n+\ell }) \leq t\} \big) - \sum_{i=1}^n L_i \ind\{s(X_{i}) \leq t\}$ for $t\in \cM$.
        \STATE Compute the thresholds $t_{\gamma, n+j}(0)$ and $t_{\gamma, n+j}(1)$.
        \IF{$s(X_{n+j}) > t_{\gamma, n+j}(1)$}
            \STATE Set $E_{\gamma, n+j} = 0$.
        \ELSIF{$t_{\gamma, n+j}(0) = t_{\gamma, n+j}(1)$}
            \STATE $\displaystyle
            \textnormal{Set } E_{\gamma, n+j} = \frac{n+1}{1 + \sum_{i=1}^n L_i \ind\{s(X_i) \leq t_{\gamma, n+j}(1)\} }.
            $
        \ELSE
            \STATE Initialize the set
            $\displaystyle 
                \cM^* = \{t \in \cM: t \geq s(X_{n+j})~\text{and}~ \fr_{n+j}(t; 0) \leq \gamma\} \cap [t_{\gamma, n+j}(1), t_{\gamma, n+j}(0)].
            $
            \STATE Remove all element $t \in \cM^*$ if there exists any $t' \in \cM, t' > t, \fr(t';0) \leq \gamma$ such that $\ell(t') > \ell(t)$.
            \STATE $\displaystyle
            \textnormal{Set } E_{\gamma, n+j} = \inf_{t \in \cM^*}\frac{n+1}{\bar\ell(t) + \sum_{i=1}^n L_i \ind\{s(X_i) \leq t\} }.
            $
        \ENDIF
    \ENDFOR  
    \ENSURE{E-values $\{E_{\gamma, n+j}\}_{j=1}^m$.}
\end{algorithmic}
\end{algorithm}

\paragraph{Connection to conformal selection.} Since  SDR extends FDR to quantitative risks, it is helpful to connect the SDR-controlling procedure of SCoRE---which combines Theorem~\ref{thm:general_sel_risk} and~\eqref{eq:def_eval_sel}---and the conformal selection procedure of~\cite{jin2023selection}. 
Define conformal p-values $p_{n+1},\dots,p_{n+m}$ for all test points analogously to~\eqref{eq:def_conformal_pval}. Conformal selection  applies the BH procedure to $\{p_{n+j}\}_{j=1}^m$ and controls the FDR---which equals the SDR~\eqref{eq:def_sel_risk} under the binary risk $\cL(f,X,Y)=\ind\{Y\leq c\}$---at nominal level $\alpha\in(0,1)$. 
It can be shown that the conformal selection set, denoted $\cS^{\text{CS}}$,  is equivalent to the e-BH output applied to $\{e_{n+j}\}_{j=1}^m$ at level $\alpha$, where we $e_{n+j} = \ind\{p_{n+j} \leq t\} / t$  with $t = \alpha|\cS^{\text{CS}}| / m$; see, e.g.,~\cite{wang2022false}. Thus, conformal selection can also be interpreted as an e-value-based selection method.

Our first result relates the SCoRE e-values to the conformal selection e-values $\{e_i\}$. Its proof is in Appendix~\ref{app:subsec_thm_sdr_conn_CS}. For convenience, write $E_{\gamma,n+j}(\ell)$ for the quantity inside the infimum in~\eqref{eq:def_eval_sel}.

\begin{prop} \label{thm:sdr_conn_CS}
    Assume a binary risk function $\cL(f, X, Y) = \ind\{Y \leq c\}$, where $c \in \mathbb{R}$ is a constant. Then,  $E_{\alpha, n+j}(1) \geq e_{n+j}$ deterministically for any $j\in [m]$. Furthermore, $e_{n+j} = 0$ implies $E_{\alpha, n+j} = 0$.
\end{prop}

Proposition~\ref{thm:sdr_conn_CS} shows that when evaluated at a positive risk level ($\ell=1$), SCoRE is no more conservative than conformal selection. 
On the other hand,  the SCoRE selection set cannot be larger than the conformal selection set, as any $j\notin \cS^{\cs}$ must obey $e_{n+j}=0$, which implies $E_{\alpha,n+j} =0$. 
In general, since SCoRE takes an infimum over $\ell\in [0,1]$, the comparison between $E_{\alpha,n+j}$ and $e_j$ is not immediate.
Nevertheless, next we show that with binary risk, slightly modifying the SCoRE procedure---by tightening the range the infimum is taken over---recovers the conformal selection procedure. Its proof is in Appendix~\ref{app:cor_compare_cs}.  
% \yingcomment{Just notice that we sometimes use risk/loss interchangeably. To be consistent we shall use ``risk''.}

\begin{corollary}\label{cor:compare_cs}
    Under the conditions in Proposition~\ref{thm:sdr_conn_CS}, define $E_{\gamma, n+j}' = E_{\gamma, n+j}(1)$ and let $\cS'$ be the output of eBH applied to $\{E_{\gamma, n+j}'\}_{j=1}^m$ at nominal level $\alpha\in(0,1)$. Then the following holds.
    \begin{enumerate}
        \item[(i)]  $\cS'$ achieves finite-sample selective risk control below $\alpha$.
        \item[(ii)] If we set $\gamma = \alpha$ in defining the \textnormal{SCoRE} e-values, then $\cS' = \cS^{\cs}$, where $\cS^{\cs}$ is the output of conformal selection at level $\alpha$ using p-values defined similar to~\eqref{eq:def_conformal_pval}.
    \end{enumerate}
\end{corollary}

% shall we use R' / R^CS instead of S?

% By Corollary~\ref{cor:compare_cs}, with a careful definition of the risk-adjusted e-values, SCoRE can be viewed as a generalization of the conformal selection procedure, just as we generalize the binary risk to continuous ones. 

% Although the theorem explicitly compares only the \textit{oracle} SCoRE e-values with the conformal e-values, a simple adjustment in the binary risk case extends its applicability to the empirical counterparts. Specifically, by defining $E_{\alpha, n+j}' := E_{\alpha, n+j}(1)$ and directly apply eBH to this set of e-values $\{E_{\alpha, n+j}'\}_{j=1}^m$, we also obtain valid SDR control. To verify this, note that for each $j$, we have $L_{n+j}E_{\alpha, n+j}(L_{n+j}) = L_{n+j}E_{\alpha, n+j}'$, which implies that $E_{\alpha, n+j}'$ are also valid risk-adjusted e-values by Theorem~\ref{thm:evalue_sel}. With this new set of e-values, it holds that $E_{\alpha, n+j}' = E_{\alpha, n+j}(L_{n+j}) \geq e_j$ for any $j$ with $Y_j \leq c$. In other words, our selection procedure is less conservative on the null, thereby potentially achieving higher power in identifying desirable candidates. In fact, the following result guarantees that SCoRE with this modified set of e-values achieves performance no worse than conformal selection.

\subsection{Improving power by boosting e-values}
\label{subsec:boosting}

We can further enhance the power of SCoRE-SDR without sacrificing SDR control,
inspired by the pruning technique in~\cite{jin2023selection, bai2024optimized,fithian2022conditional} and the strategies in~\cite{xu2024online} designed for FDR control.
For notational simplicity, in this section we write $E_{n+j}=E_{\gamma,n+j}$. 

The first variant, heterogeneous boosting, generates $\xi_{n+j}\iid \textrm{Unif}([0,1])$ independent of everything else, and set
\$
\cR_{\hete} = \{j\colon E_{n+j}/\xi_{n+j} \geq m/(\alpha k_\hete^*) \},~~ \textrm{where}~~ 
k_\hete^* = \max \Big\{k\colon \textstyle{\sum_{j=1}^m }\ind\{E_{n+j}/\xi_{n+j} \geq m/(\alpha k)\} \geq k\Big\}.
\$
% Here we write the e-values $E_{\gamma, n+j}$ as $E_{n+j}$ throughout this section for notational simplicity. 
Alternatively, homogeneous boosting generates $\xi_{n+j} \equiv \xi\sim  \textrm{Unif}([0,1])$, and set 
\$
\cR_{\homo} = \{j\colon E_{n+j}/\xi \geq m/(\alpha  k_{\homo}^*) \},~~ \textrm{where}~~ 
k_{\homo}^* = \max \Big\{k\colon \textstyle{\sum_{j=1}^m }\ind\{E_{n+j}/\xi \geq m/(\alpha k)\} \geq k\Big\}.
\$

It has been shown~\citep{bai2024optimized} that both $\cR_{\text{hete}}$ and $\cR_{\text{homo}}$ are supersets of the selection set of BH applied to $\{E_{n+j}\}$, and the next theorem states that SDR control is preserved with proof in Appendix~\ref{app:subsec_boost}.

\begin{theorem}\label{thm:eval_boost}
    Suppose the e-values $\{E_{n+j}\}_{j=1}^m$ satisfy Definition~\ref{def:eval}. Then, $\cR_\hete$ and $\cR_\homo$ run at level $\alpha\in (0,1)$ control the SDR below  $\alpha$.
\end{theorem}

\begin{remark}
    For a set of standard e-values in classical hypothesis testing, the boosting strategy described above remains valid. Specifically, given e-values $\{e_j\}_{i=1}^m$ and independent boosting factors $\{\xi_j\}_{j=1}^m$, the application of the eBH procedure to the adjusted inputs $\{e_j / \xi_j\}_{j=1}^m$ ensures valid FDR control. This procedure can be interpreted as a special case of the e-weighted p-testing framework~\citep{ramdas2019unified, ramdas2024hypothesis,xu2024online}, where the e-values are $\{e_j\}_{i=1}^m$ and the p-values are vacuously defined as the boosting factors $\{\xi_j\}_{j=1}^m$. Accordingly, Theorem~\ref{thm:eval_boost} can be viewed as a generalization of this result, extending from standard e-values to risk-adjusted conformal e-values.
    % On connection to e-weighted p-testing procedures.
\end{remark}

\begin{remark}
    Since MDR coincides with SDR when $m = 1$, the boosting strategy can, in principle, also be applied to the MDR setting introduced in Section~\ref{sec:marginal}. However, Proposition~\ref{prop:single_equiv} shows boosting brings little benefit. Let $\xi \sim \textnormal{Unif}([0,1])$ be independently generated and set $\gamma = \alpha$. Then we have
    \begin{align*}
        \ind\{E_{\gamma,n+1}/\xi \geq 1/\alpha\} = \ind\{E_{\gamma,n+1} \geq 1/(\alpha/\xi)\} = \ind\{E_{\gamma,n+1} \geq 1/\alpha\}
    \end{align*}
    where the second equality follows from Proposition~\ref{prop:single_equiv} and the fact that $\gamma = \alpha \leq \alpha/\xi$. Hence, the test function $\hat\psi_{n+1}$ will remain unaffected after the boosting operation.
    % On why boosting is not meaningful for MDR case.
\end{remark}

% \tiancomment{Can we define the risk-adjusted p-value to be something like $\PP(p_{n+j} \leq \alpha \mid L_{n+j}) \leq \alpha /L_{n+j}$...} \yingcomment{this property might lead to guarantees we like, but for continuous risk, it seems hard to find such a p-value?}

\subsection{Asymptotics and optimality}
\label{subsec:asymptotics_sdr}

To complete the picture, we now study the asymptotic behavior of our SDR-controlling procedure to gain insights on the choice of $s(\cdot)$. Again, we view the model $f$ and the score function $s$ as fixed. 
We define
\@\label{eq:def_power_sdr}
\text{Power} := \EE\bigg[\frac{1}{m} \sum_{j=1}^m r(X_{n+j},Y_{n+j}) \hat\psi_{n+j} \bigg],
\@
where $r\colon \cX \times\cY\to [0,1]$ is a user-specified reward function. Intuitively, this notion of the power captures the total reward in the selectively deployed units (scaled by $1/m$), such as the expected rewards in investing in promising drugs. 
The asymptotic behavior of our SDR-controlling procedure, as well as the optimal choice of the score function, are characterized in Theorem~\ref{thm:asymptotic_sdr}, whose proof is in Appendix~\ref{app:subsec_proof_asymptotic_sdr}. 

\begin{theorem}\label{thm:asymptotic_sdr}
    Assume the distribution of $s(X)$ has no point mass. Define 
    \$
    \fr(t) =  \frac{\EE[L\ind\{s(X)\leq t\}]}{\PP(s(X)\leq t)},\quad \text{and}\quad t_\gamma^* = \max\{t\colon \fr(t)\leq \gamma\}.
    \$
    We further assume that for any sufficiently small $\delta>0$, we have $\fr(t)<\gamma$ for $t\in (t_\gamma^*-\delta,t_\gamma^*)$. Then the following statements hold:
    \begin{enumerate}[label=(\roman*).]
    \item As $n,m\to \infty$, $\sup_{1\leq j\leq m} \sup_{ \ell\in [0,1]} \big| t_{\gamma,n+j} (\ell) - t_\gamma^* \big|  ~\asto ~ 0$. 
    \item $\lim_{n,m\to\infty}\textnormal{Power} = \EE[r(X_{n+1},Y_{n+1})\ind\{s(X_{n+1})\leq t_\gamma^*\}]$ if $\gamma < \alpha$, and  $\lim_{n,m\to\infty} \textnormal{Power}= 0$ if $\gamma>\alpha$. Thus, for a fixed score function $s(\cdot)$, the asymptotic power is optimized as $\gamma\uparrow \alpha$. 
    \item Let $r(x):=\EE[r(X,Y)\given X=x]$ and $l(x):=\EE[\cL(f,X,Y)\given X=x]$ be the conditional expectation of the reward and risk, and suppose $r(X)>0$ almost surely, and the distribution of $(l(X)-\alpha)_+/r(X)$ has no point mass. Let $\gamma\uparrow \alpha$, then $\lim_{n,m\to\infty}\textnormal{Power}$ is optimized at any score function $s(\cdot)$ such that $s(x)$ is monotone in $(l(x)-\alpha)/r(x)$.  
    \end{enumerate}
\end{theorem}

The conditions in Theorem~\ref{thm:asymptotic_sdr} resemble the standard mild assumptions in~\cite{storey2004strong} to obtain meaningful asymptotic analysis of the FDR. 

Theorem~\ref{thm:asymptotic_sdr} demonstrates that the score $s(X)$ should aim to rank test instances by their \emph{excess risk per unit reward} $(l(x)-\alpha)/r(x)$. 
This is distinct from the optimality result for MDR control (Theorem~\ref{thm:opt_mdr}).  
Intuitively, the optimal procedure explores the cost--benefit tradeoff: it prioritizes instances that achieves high reward per unit of risk. 
Compared with the intuitive choice in which $s(x)$ estimates $l(x)$, this  makes a substantial difference only when dividing by $r(x)$ drastically changes the ranking, such as when $l(x)-\alpha$ and $r(x)$ are very positively correlated. Finally, an intuitive way to implement this is to plug in estimators for the two functions and set $s(x)=(\hat{l}(x)-\alpha)/\hat{r}(x)$.

% !TEX root = main.tex

\section{Extension: SCoRE under distribution shift}
\label{sec:weighted}

The techniques of constructing risk-adjusted e-values based on exchangeability enable  broader methodology. 
Here, we present a natural extension of SCoRE to scenarios where the calibration and test data are only weighted exchangeable, referred to as the covariate shift setting~\citep{TibshiraniBCR19}. Such settings are particularly useful in applications like drug discovery where there is often differences between labeled and unlabeled data~\citep{krstajic2021critical,jin2023model,laghuvarapu2023codrug,laghuvarapu2026confhit}. 

\vspace{-0.5em}
\begin{assumption}\label{assump:cov_shift}
    The labeled data follow $(X_i,Y_i)\iid P$ while the test data follow  $(X_{n+j},Y_{n+j})\iid Q$, and the two distributions obey $\ud Q/\ud P(x,y) = w(x)$ for a known or estimable weight function $w\colon \cX\to \RR^+$.
\end{assumption}

The key strategy is to construct risk-adjusted e-values obeying $\EE_Q[L_{n+j}E_{n+j}]\leq 1$ under the \emph{test} distribution. The MDR and SDR control then follow by the same testing arguments in Section~\ref{sec:general_method}.
We first address the case where $w(\cdot)$ is known, in which case an extension of SCoRE provides finite-sample MDR/SDR control. We then briefly discuss robustness properties with estimated weights where the guarantees become asymptotic to accommodate estimation errors.
% We shall note that this section aims to demonstrate the breadth of the conceptual framework of SCoRE, and we leave more fine-grained analysis of SCoRE under estimated weights for future work.

\subsection{Marginal risk control under covariate shift}

We use the same thresholding rule 
$\hat\psi_{n+1} = \ind\{E_{\gamma, n+1} \geq 1/\alpha\}$, where the \emph{weighted} e-value is defined as 
\@\label{eq:eval_mdr_w}
E_{\gamma, n+1} = \inf_{\ell \in[0,1]} ~\Bigg\{ \frac{  \ind\{s(X_{n+1})\leq t_\gamma(\ell)\} \cdot \sum_{i=1}^{n+1} w_i }{\sum_{i=1}^n w_i\cdot L_i \ind\{s(X_{i})\leq t_\gamma(\ell)\} + w_{n+1}\cdot \ell\ind\{s(X_{n+1})\leq t_\gamma(\ell)\} } \Bigg\}.
\@
Here we set $w_i=w(X_i)$ for $i\in [n+1]$,  $t_\gamma(\ell) = \max \big\{t \in \cM\colon \textrm{F}(t;\ell) \leq \gamma \big\}$, and 
\$ 
\textrm{F}(t;\ell) =  \frac{ \sum_{i=1}^n w_i L_i \ind\{s(X_{i})\leq t\} + w_{n+1} \cdot \ell \ind\{s(X_{n+1})\leq t\} }{\sum_{i=1}^{n+1}w_i  }.
\$
Here we define $\cM:=\{s(X_i)\}_{i=1}^{n+1}$, and again set $E_{\gamma,n+1} = 0$ when $\inf_{\ell \in[0,1]} t_{\gamma}(\ell) = -\infty$. 
The following theorem, whose proof is in Appendix~\ref{app:thm_w_single}, demonstrates the validity of the weighted SCoRE e-value.

\begin{theorem} \label{thm:w_single}
    % Assuming the covariate shift setting introduced in Section~\ref{sec:weighted},
    Under Assumption~\ref{assump:cov_shift}, 
    for any fixed constant $\gamma \in(0,1)$, it holds that $\EE_Q[L_{n+1} E_{\gamma,n+1}  ]\leq 1$. 
\end{theorem}

Extending our discussion below Theorem~\ref{thm:single}, the main idea of~\eqref{eq:eval_mdr_w} is based on the observation that, should $L_{n+1}$ be known, any random variable of the form %(see, e.g., the weighted exchangeability idea in~\cite{TibshiraniBCR19})
\$
\frac{w_{n+1} \cdot L_{n+1} A_{n+1}}{\sum_{i=1}^n w_i\cdot  L_i A_i  + w_{n+1}\cdot L_{n+1}A_{n+1}},
\$ 
has expectation equal to $1$ under the covariate shift assumption~\citep{TibshiraniBCR19}.
 
As in the unweighted case, one can often avoid computing the infimum in~\eqref{eq:eval_mdr_w} explicitly. Proposition \ref{prop:single_equiv} in Appendix~\ref{app:subsec_shortcut_w} presents an equivalent shortcut.

\subsection{Selective risk control under covariate shift} \label{subsec:sdr_control_shift}

For the SDR control, we define the \emph{weighted} e-values as
\begin{align} \label{eq:def_eval_sel_w}
E_{\gamma, n+j} = \inf_{\ell \in [0,1]} \bigg\{ \frac{\ind\{s(X_{n+j}) \leq t_{\gamma,n+j}(\ell)\} \cdot (w_{n+j}+\sum_{i=1}^n w_i) }{w_{n+j} \cdot \ell\ind\{s(X_{n+j}) \leq t_{\gamma,n+j}(\ell)\}+  \sum_{i=1}^n  w_i \cdot L_i\ind\{s(X_{i}) \leq t_{\gamma,n+j}(\ell)\}} \bigg\}, \qquad
\end{align}
where $w_i=w(X_i)$ for $i\in[n+m]$, $t_{\gamma,n+j}(\ell) = \max\big\{t\colon  {\textrm{FR}_{n+j}}(t;\ell) \leq \gamma\big\}$, and
\$
{\textrm{FR}_{n+j}}(t;\ell) = \frac{w_{n+j} \cdot \ell\ind\{s(X_{n+j})\leq t\} + \sum_{i=1}^n w_i \cdot L_i\ind\{s(X_i)\leq t\}}{1+\sum_{k\neq j} \ind\{s(X_{n+k}) \leq t\}}\cdot \frac{m}{w_{n+j}+\sum_{i=1}^n w_i}.
\$
Here $\max\varnothing=-\infty$, and set $E_{\gamma,n+j} = 0$ when $\inf_{\ell \in[0,1]} t_{\gamma,n+j}(\ell) = -\infty$. 

The construction~\ref{eq:def_eval_sel_w} mirrors the construction in Section~\ref{sec:selective} while accounting for covariate shift weights. The proof for the theorem below can be found in Appendix~\ref{app:thm_w_evaluesel}.

\begin{theorem} \label{thm:w_evaluesel}
    % Assuming the covariate shift setting introduced in Section~\ref{sec:weighted}, 
    % i changed to the later version to clarify
    Under Assumption~\ref{assump:cov_shift}, 
    % Let $\{(X_i, Y_i)\}_{i=1}^{n}$ be drawn i.i.d. from a distribution $P$, and let $\{(X_{n+j}, Y_{n+j})\}_{j=1}^n$ be drawn i.i.d. from a distribution $Q$. We assume that $\mathrm{d}Q/\mathrm{d}P(x,y) = w(x)$. Then, 
    $\EE_Q[L_{n+j} E_{\gamma,n+j}  ]\leq 1$ for any fixed $\gamma \in(0,1)$ and $j \in [m]$.
\end{theorem}

As in the unweighted case, SCoRE under covariate shift also admits a computational shortcut. We outline the algorithm in Algorithm~\ref{alg:w_sel} and prove its equivalence to~\eqref{eq:def_eval_sel_w} in Appendix~\ref{app:subsec_mult_equiv_w}.

% repeat the proof in appendix or not..?

% \yingcomment{decide later: whether put weighted in appendix or add estimated weight stuffs}

% \paragraph{Efficient computation.}
% Following the unweighted case, the computation shortcut still applies, yet with modified definitions of intermediate quantities. 

\subsection{Robustness to estimated weights}

When the weight function $w(\cdot)$ is unknown, it is natural to first obtain an estimator $\hat{w}(\cdot)$ and compute the MDR/SDR e-values in~\eqref{eq:eval_mdr_w} and~\eqref{eq:def_eval_sel_w} with $w_i=\hat{w}(X_i)$. 
Our analysis shows that the SCoRE procedure asymptotically controls the MDR and SDR provided that the estimated weight function asymptotically converges to the true weight function. 
% The proof can be found in Appendix~\ref{app:w_mdr_asymp} and \ref{app:w_sdr_asymp}.

% the estimated weight can depend on the n,m calib/test data in arbitrary way, right?

\begin{theorem} \label{thm:w_mdr_asymp}
    Under Assumption~\ref{assump:cov_shift}, assume we have access to a sequence of random weight estimates $\{\bar{w}_{n}(\cdot)\}$ trained independent of $\{(X_i,Y_i)\}_{i=1}^{n+1}$ obeying $\|\bar{w}_{n}(\cdot)-w(\cdot)\|_{L_2(\PP_X)}=o_P(1)$ as $n \to \infty$. In addition, assume the function $F(t) = \EE_P[w(X) l(X) \ind\{s(X) \leq t\}]/\EE_P[w(X)]$ is continuous and strictly increasing at $t^* = \sup\{t: F(t) \leq \alpha\}$. Set  $\gamma = \alpha$ and denote by $\textnormal{MDR}_{n}$ the MDR of SCoRE using the e-values~\eqref{eq:eval_mdr_w} with $\bar{w}_{n}(\cdot)$ in place of $w(\cdot)$. 
    Then, we have  $\limsup_{n \to \infty} \textnormal{MDR}_{n} \leq \alpha$.
\end{theorem}

\begin{theorem} \label{thm:w_sdr_asymp}
    Under Assumption~\ref{assump:cov_shift}, assume we have access to a sequence of random weight estimates $\{\bar{w}_{n,m}(\cdot)\}$ trained independent of $\{(X_i,Y_i)\}_{i=1}^{n+m}$ obeying $\|\bar{w}_{n,m}(\cdot)-w(\cdot)\|_{L_2(\PP_X)}=o_P(1)$ as $n,m \to \infty$. Assume that the distribution of $s(X)$ is non-atomic, and the function
    $ 
        F(t) = \frac{\EE_P[w(X) L \ind\{s(X) \leq t\}]}{\PP_Q(s(X) \leq t) \cdot \EE_P[w(X)]}
    $ 
    is continuous and strictly increasing at $t^* = \sup\{t: F(t) \leq \alpha\}$. Set  $\gamma = \alpha$ and denote by $\textnormal{SDR}_{n,m}$ the SDR of SCoRE using the e-values~\eqref{eq:def_eval_sel_w} with $\bar{w}_{n,m}(\cdot)$ in place of $w(\cdot)$.  Then, we have  $\limsup_{n,m \to \infty} \textnormal{SDR}_{n,m} \leq \alpha$.
    % Under Assumption~\ref{assump:cov_shift}, we assume the true weight  $\mathrm{d}Q/\mathrm{d}P(x,y) = w^*(x)$ is not observed. Assume we have access to a sequence of weight estimates $\{\hat{w}_n(\cdot)\}_{n\geq 1}$ obeying $\|\hat{w}_n(\cdot)-w^*(\cdot)\|_{\infty}=o_P(1)$ as $n\to\infty$. In addition, suppose the distribution of $s(X)$ is non-atomic. Then,
    % \begin{enumerate}
    %     \item[(i)] For any fixed $\gamma, j$, the e-value $E_{\gamma,n+j}^{(n)}$ based on $\hat{w}_n(\cdot)$ converges to $E_{\gamma,n+j}^*$ based on $w^*(\cdot)$.
    %     \item[(ii)] Define $ F^*(\eta) = \frac{1}{m} \sum_{j=1}^m \ind\{E_{\gamma,n+j}^* \geq \eta\}$, 
    %     and let the e-BH cutoff be $\eta^* = \inf\{\eta: \eta F^*(\eta) \leq \alpha\}$. Then, if the distribution of $E_{\gamma,n+j}^*$ has no atom at $\eta^*$, we have $\lim_{n \to \infty} \textnormal{SDR}_n \leq \alpha$.
    % \end{enumerate}
\end{theorem}

% To preserve weighted exchangeability, the weight estimates $\hat{w}_k$ must depend on the calibration and test data in a homogeneous manner. A simple approach is to estimate $\hat{w}_k$ using a separate, independent dataset of size $n'$. If the estimator is consistent and $n' \to \infty$, then Theorem~\ref{thm:w_mdr_asymp} and Theorem~\ref{thm:w_sdr_asymp} apply, establishing asymptotically risk-controlled deployments. 
%When data are very limited, one may instead adopt a full training strategy that preserves exchangeability, analogous to that described in~\cite{bai2024optimized}.

% \yingcomment{A potential issue is that we don't know how $k$ and $n,m$ scale with each other, and $\hat{w}_k$ might be estimated with some data. Could we change the theorems to $n,m\to \infty$ only?}

% A double robustness result is in Appendix~\ref{app:subsec_mdr_dr_discussion}. \todo{discussion}

Under mild assumptions, the SCoRE procedure exhibits a \emph{double robustness} property, further relaxing the dependence on accurate weight estimation in the results above. We thus omit the proofs of Theorems \ref{thm:w_mdr_asymp} and \ref{thm:w_sdr_asymp}, as they follow directly from the proofs of these double robustness results (deferred to  Appendices~\ref{app:subsec_mdr_dr_discussion} and~\ref{app:subsec_sdr_dr} for brevity). Those results  mirror the established results in conformal prediction and selection, where valid inference is maintained even if part of the model is misspecified, which we briefly discuss below.

\begin{remark}[Doubly robust calibration]
    A series of research has shown that conformal prediction and selection under covariate shift enjoy ``double robustness'' properties~\citep{lei2020conformal,yang2024doubly,jin2023model} in the sense that they achieve the desired guarantee (coverage or FDR control) when either (i) the estimated weights are consistent or (ii) certain score function converges to an ideal score (conditional quantiles in~\cite{lei2020conformal} or conditional distribution functions in~\cite{jin2023model}). We remark that, with the threshold-based decisions and the expected risk control target, it is nontrivial to prove analogous double robustness results for SCoRE when only plug-in weights are used and no bias-adjustment terms like~\cite{yang2024doubly}. Nevertheless, it is possible to achieve so by calibrating the weights to a finite-sample balancing condition~\citep{hainmueller2012entropy,zubizarreta2015stable,jin2025cross}. As the development of this approach is somewhat technical, we defer the statements and theory of MDR (resp.~SDR) control to Appendix~\ref{app:subsec_mdr_dr_discussion} (resp.~\ref{app:subsec_sdr_dr}). In a nutshell, our results show that, if the estimated weights additionally satisfy a finite-sample balancing condition based on an estimated conditional risk $\hat{l}(x)$ for $l(x) = \EE[\cL(f,X,Y)\given X=x]$, then SCoRE achieves (asymptotic) MDR/SDR control if either (i) the weights are consistent, or (ii) the conditional risk model is consistent. 
\end{remark}

% !TEX root = main.tex

\section{Real data applications}
\label{sec:real}

We apply SCoRE to three applications that require selective deployment with continuous, task-specific risks: drug discovery under covariate shift (Section~\ref{subsec:app_drug}), selective use of ICU length-of-stay predictions (Section~\ref{subsec:app_icu}) and abstention of  radiology report generation with large language models (Section~\ref{subsec:app_llm}). 
Each application specifies a distinct risk function $\cL$ and (optionally) a reward function $r$, and we evaluate  MDR and SDR control together with various notions of selection power. Throughout the applications we focus on SCoRE procedures, and we compare with natural baselines to demonstrate the advantages of SCoRE in Section~\ref{sec:simu}.

\subsection{Application to drug discovery} \label{subsec:app_drug}
% MDR/SDR: similarity to a holdout set; cost $L_i$ of low-activity drugs ($L_i \ind\{Y_i\leq c\}$) (4 subplots)

% We apply SCoRE to virtual screening tasks in drug discovery, where predictions from machine learning models are used to shortlist promising drug candidates. 

We first apply SCoRE to drug discovery to select promising drug candidates while controlling the wasted resources. 
Since wet-lab essays for drug properties (e.g.,  activity against a disease target) are expensive~\citep{macarron2011impact}, ML models are often used to prioritize candidates for follow-up experiments. Existing conformal selection methods in this area typically control the fraction of false leads~\citep{Bai_Tang_Xu_Svetnik_Khalili_Yu_Yang_2024,bai2024optimized,gui2025acs,huo2024real}, which is appropriate when each false lead incurs a similar downstream cost. In practice, however, follow-up costs can vary substantially across molecules, and one may also want to encourage secondary objectives such as diversity~\citep{nair2025diversifying}.   

\vspace{-0.5em}
\paragraph{Risk and reward functions.}
Each sample is a drug candidate with features $X \in \cX$ and a biological property $Y \in \cY \subseteq \mathbb{R}$. 
We aim to control the expected \emph{wasted resources} among \emph{false leads}. Consider a pre-determined threshold $c\in \RR$, and a general cost of development $L(X)\in \RR$. We define the risk   $\cL(f,X,Y) = L(X)\cdot \ind\{Y\leq c\}$. 
Here we use the  synthetic accessibility (SA) score \citep{ertl2009estimation}, denoted as $\text{SA}(x)$, as a proxy for cost (difficulty of development), which is fully determined by its chemical structure.  
Here, MDR control implies limited total wastage of resources, while SDR control implies limited average wastage among selected candidates, which is more appropriate when wastage is allowed to scale with the number of follow-ups.
Throughout, we run SCoRE after normalizing the risks to $[0,1]$ and report results on the original scale.
To reflect secondary factors, we consider three rewards:

\vspace{0.25em}
\begin{enumerate}[label=(\alph*)]
    \item \emph{Diversity.} To encourage the selection of diverse molecules by setting a reward function as the dissimilarity to a hold-out reference set. Here, we use $r_1(X, Y) = 1 - \mathrm{AvgTanimoto}(X)$ where $\mathrm{AvgTanimoto}(X)$ is $X$'s mean Tanimoto coefficient with respect to molecules in the reference set $\cD_{\text{train}}$. 
    \item \emph{Activity.} To  prioritize candidates with exceptional activity, we set the reward as $r_2(X, Y) = Y$. 
    \item Finally, we can set a constant reward $r_0(X,Y) = 1$ to promote more discoveries.
\end{enumerate} 
\vspace{0.25em}

\vspace{-0.5em}
\paragraph{Datasets and models.} We apply SCoRE to four drug property prediction tasks with data from Therapeutic Data Commons \citep{huang2021therapeutics}. Since it is common to observe distribution shift in the drug discovery setting, we apply an artificial shift defined by $w(X) = \text{sigmoid}(|\text{mw}(X) - 400|/400$), where $\text{sigmoid}(z) = 1/(1+e^{-z})$ and $\text{mw}(X)$ denotes the molecular weight of the molecule $X$. This distribution shift is unknown to the learner yet may be learned by deep learning models. %Accordingly, we employ the weighted SCoRE procedure, with weights learned via a random forest classifier.

\begin{figure}
    \centering
    \includegraphics[width=\linewidth]{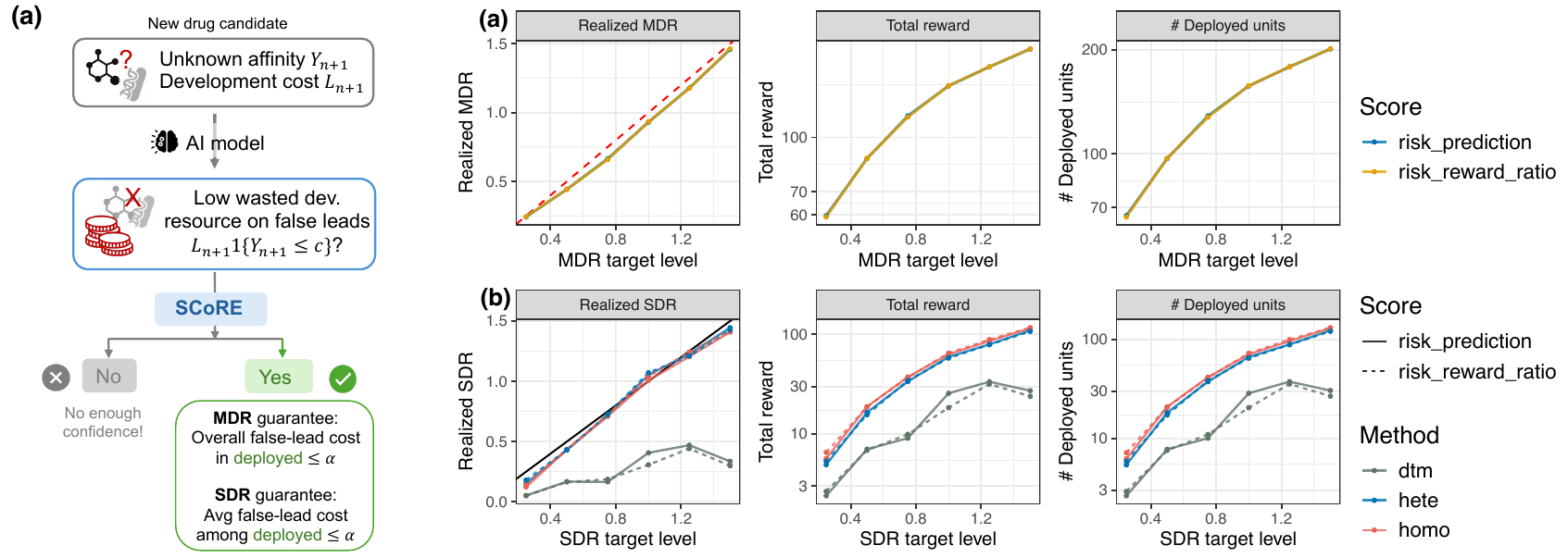}
    \caption{{\small SCoRE for selecting drugs with cost efficiency under covariate shift. \textbf{(a) Overview}: Given predicted drug activities, the goal is to identify highly active drugs with cost wastage control; SCoRE provides MDR and SDR guarantees among shortlisted drug candidates. \textbf{(b) MDR control}: realized MDR at various target levels in the original scale (left), total reward of selected drugs, number of selected drugs (right). \textbf{(c) SDR control}: realized SDR at various target levels (left), total reward of selected drugs (middle), number of selected drugs (right).}}
    \label{fig:real-drug} 
    % what was the XX (that i deleted) after the total reward? "total reward (XX)"
\end{figure}

Each dataset is randomly split into training ($\cD_{\text{train}}$, 40\%), calibration ($\cD_{\text{calib}}$, 30\%) and test ($\cD_{\text{test}}$, 30\%) folds, and the artificial shift is applied to draw the test data $\cD_{\text{test}}$ using rejection sampling. 
The training fold is used to train the risk and reward predictors using the \texttt{DeepPurpose} Python library \citep{huang2020deeppurpose} with the \texttt{DGL\_AttentiveFP} molecule embedding. 
We also set aside a subset of shifted data to train the covariate shift weights via probabilistic classification.
% We also estimate the covariate shift weight function via probabilistic classification on two datasets derived from $\cD_{\text{train}}$, one shifted and one unshifted. \yingcomment{set aside some shifted data to train the weights via classification}
Given the predictors and estimated weights, we apply SCoRE  to $\mathcal{D}_{\text{calib}}$ and $\mathcal{D}_{\text{test}}$.  
For each reward function, we use two score choices suggested by our optimality analysis, a \texttt{risk\_prediction} score $s(x) = \hat{l}(x)$ and a \texttt{risk\_reward\_ratio} score  $s=\hat{l}(x)/\hat{r}(x)$ (MDR case), $s = (\hat{l}(x) - \alpha) / \hat{r}(x)$ (SDR case), where $\hat{l}(\cdot)$ and $\hat{r}(\cdot)$ denotes the learned risk and reward functions. 
We repeat the whole pipeline for $N=100$ independent runs.
For the SCoRE-MDR procedure (Figure~\ref{fig:real-drug}b), the average realized MDR, reward and deployed units are computed by averaging $\psi_{n+1} L_{n+1}$, $\psi_{n+1}r_{1,n+1}$ and $\psi_{n+1}$ over the test data and $N=100$ independent runs.  
For the SCoRE-SDR procedure (Figure~\ref{fig:real-drug}c), these metrics are computed as $\frac{1}{1\vee |\cR|} \sum_{j=1}^m \psi_{n+j} L_{n+j}$, $\frac{1}{m}\sum_{j=1}^m \psi_{n+j} r_{1,n+j}$ and $|\cR|$ respectively, averaged over $N=100$ runs.

\vspace{-0.5em}
\paragraph{Results.} Figure~\ref{fig:real-drug} illustrates the pipeline and results on the \texttt{caco2\_wang} dataset ($906$ drug candidates in total) with the diversity reward $r_1$; see Appendix~\ref{app:drug_more} for additional  results for all the datasets and reward functions.
SCoRE achieves robust MDR and SDR control with useful selection power, even when the covariate shift weights are estimated.  
Among boosting strategies, consistent to earlier observations in~\cite{jin2023model, bai2024optimized}, homogeneous e-value boosting typically achieves the highest selection power. While theory (Theorem~\ref{thm:opt_mdr} and Theorem~\ref{thm:asymptotic_sdr}) suggests a tradeoff between selecting more units (\texttt{risk\_prediction} score) and accumulating higher total reward (\texttt{risk\_reward\_ratio} score), we see only a small empirical difference, likely because dividing by the reward function does not drastically change the priority of candidates in SCoRE.
% \todo{add more datasets in appendix}

\subsection{Application to clinical prediction error management} \label{subsec:app_icu}

% $L_2$ loss for regression (ICU stay) \todo{find the dataset}

Our second applications concerns the management of predictive error in clinical settings where resource allocation relies on noisy model predictions. We focus on selecting accurate predictions for the length of stay for patients in the Intensive Care Unit (ICU) using the \texttt{MIMIC-IV} dataset~\citep{PhysioNet-mimiciv-3.1}. Each data point corresponds to a patient, whose features $X$ include relevant personal and clinical information such as ethnicity, diagnoses, and medications. The response  $Y\in \RR^+$ is the patient's length of stay in the ICU. 

\vspace{-0.5em}
\paragraph{Risk and reward functions.} The primary objective is to select test cases for which a trained stay length predictor $f(X)$ are sufficiently close to the ground truth thus reliable for clinical deployment.
We define the risk function as the $\ell_2$ loss of prediction, $\cL(f, X, Y) = (Y - f(X))^2$. Besides the constant reward $r_0(X,Y) = 1$, we use  $r_1(X,Y) = Y$ to prioritize reliable predictions of patients with long ICU stays. Again, we rescale the outcomes so the boundedness conditions apply. 

\vspace{-0.5em}
\paragraph{Dataset and models.}
The ICU stay data from the \texttt{MIMIC-IV} dataset is pre-processed with an adapted version of the pipelines developed by~\cite{gupta2022extensive}. 
After processing, we  subsample $10000$ observations, half of which are used to train the length of stay predictor, $f$, which was instantiated as a random forest model without tuning. The remaining data are then split into the training subset $\cD_{\text{train}}$, the calibration subset $\cD_{\text{calib}}$, and the test subset $\cD_{\text{test}}$ in a $3:1:1$ ratio. 
We train the risk predictor using a random forest model on $\mathcal{D}_{\text{train}}$, and reuse $f$ as the reward predictor. No  covariate shift was imposed on the dataset for this task, and all the other setups are the same as in Section~\ref{subsec:app_drug}.

\vspace{-0.5em}
\paragraph{Results.}
Figure~\ref{fig:icu} presents the results for this application. Again, SCoRE achieves tight MDR and SDR control in selecting error-controlled predictions without observing the true labels, while exhibiting good selection power. 
In the SDR-controlling variants, homogeneous and heterogeneous boosting leads to comparable power as the deterministic version, yet with realized error closer to the target level. 

\begin{figure}
    \centering
    \includegraphics[width=\linewidth]{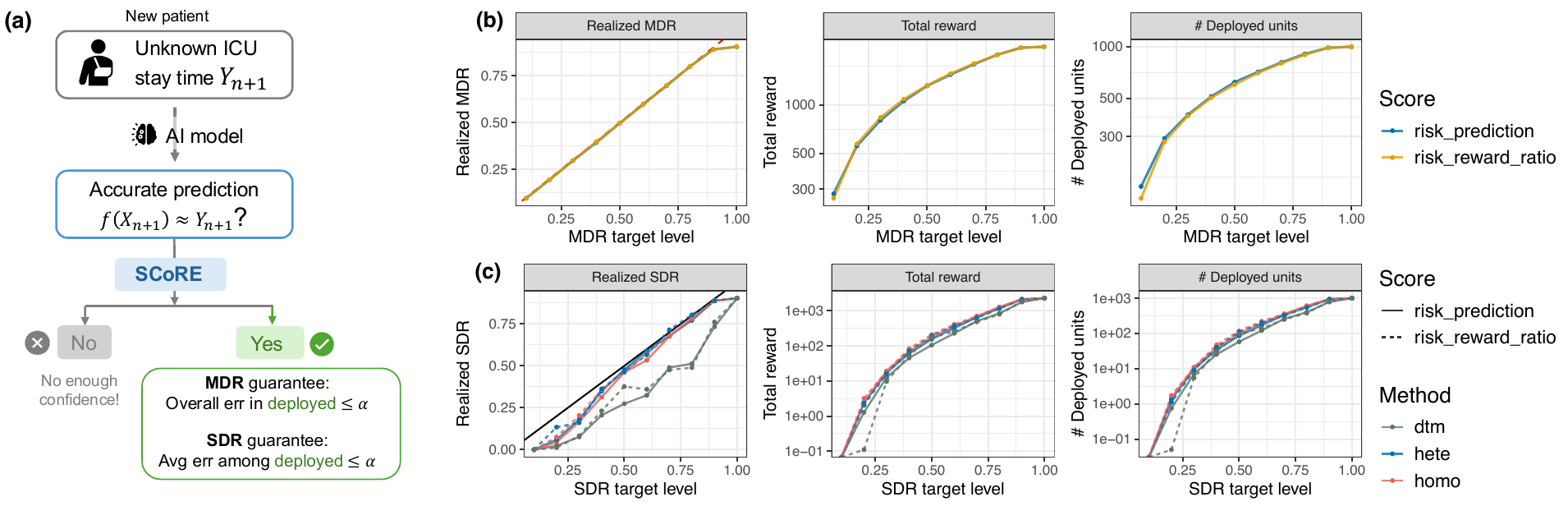}
    \caption{{\small SCoRE for identifying accurate ICU stay time prediction. \textbf{(a) Overview}: Given model predictions, the goal is to identify predictions that are close to the unknown ICU stay time; SCoRE provides MDR and SDR guarantees among identified cases. \textbf{(b) MDR control}: realized MDR at various target levels (left), total reward (stay time) of deployed units, scaled by $1/m$ (middle), number of deployed units (right). \textbf{(c) SDR control}: realized SDR at various target levels (left), total reward of deployed units (middle), number of deployed units (right).}}
    \label{fig:icu}
\end{figure}

\subsection{Application to LLM abstention} \label{subsec:app_llm} 

Finally, we apply SCoRE to the task of aligning large language models for automated chest X-ray radiology report generation (Figure~\ref{fig:llm}). Given a collection of machine-generated diagnoses, the objective in this setting is to select a subset for deployment where the reports are both factually accurate and clinically valuable. 

\vspace{-0.5em}
\paragraph{Datasets and models.}
Following~\cite{gui2024conformal,bai2024optimized,gui2025acs}, each feature $X \in \cX$ is a radiology image serving as a ``prompt''. A vision-to-language model $f: \cX \to \cY$ processes this image to generate a report summarizing its findings, where $\cY$ denotes the space of reports. The ground-truth response $Y$ represents ``gold-standard'' for each image, such as a report authored by human experts.  
We use a subset of the \texttt{MIMIC-CXR} dataset \citep{johnson2019mimic}, with the vision-to-language model $f$ is an encoder-decoder model identical to the one fine-tuned in \cite{gui2024conformal}. 

\vspace{-0.5em}
\paragraph{Risk and reward functions.}  Our risk and reward functions rely on the 14-dimensional label vectors produced by CheXbert~\citep{smit2020chexbert} based on any report, where each vector indicates the status of a specific finding---categorized as present, absent, uncertain, or unmentioned. 
We define the risk function $\cL(f,X,Y)$  as a weighted sum of the false negatives and false positives when comparing $f(X)$ and $Y$ across the CheXbert labels, which measures the alignment between generated reports and human-quality reports in continuous spectrum. 
% For the evaluation of factuality, we use CheXbert \citep{smit2020chexbert} to convert both the machine-generated and the ground-truth reference reports into 14-dimensional label vectors where  
We consider two reward functions: a constant reward $r_0$, and a confidence-weighted reward $r_1$ that assigns higher values to reports that have more correct labels for findings that are definitively present or absent (instead of uncertain or unmentioned). 
For the prediction of risk and reward, we extract 12 distinct numerical features from each report that heuristically measure the uncertainty of LLM-generated outputs similar to prior works. The risk is rescaled to $[0,1]$ before running SCoRE, and the results are reported on the original scale. 
 Further details on the dataset, model,  the specific formulations of the risk and reward functions, and risk/reward prediction models are provided in Appendix~\ref{app:llm_detail}. 
 
We sample 600 observations from the dataset, using 100 to fine-tune the hyper-parameters in the uncertainty features. The remaining observations are uniformly split into three folds of sizes $|\cD_{\text{train}}| = 200$, $|\cD_{\text{calib}}| = 100$, and $|\cD_{\text{test}}| = 200$ in each run of the experiments. Both the risk and reward predictors are implemented as random forest models without parameter tuning. As such, our experiments evaluate the application of SCoRE in scenarios with limited labeled data. The results are averaged over $N=100$ independent runs.

\vspace{-0.5em}
\paragraph{Results.}  Figure~\ref{fig:llm} presents the results for this task, which again demonstrate that SCoRE achieves tight risk control and satisfactory selection power, offering reliable guarantees when detecting high-quality radiology reports with continuous risk control. 
The SDR-control variants with homogeneous and heterogeneous boosting yield slightly higher reward  and closer-to-target SDR, yet the deterministic variant seems to achieve similar power (in terms of both of the two reward functions) with lower error. We also did not observe a significant difference in rewards when using the reward-aware scores, likely due to the fact that dividing the risk by the reward did not really change the ranking of units a lot. 

\begin{figure}
    \centering
    \includegraphics[width=\linewidth]{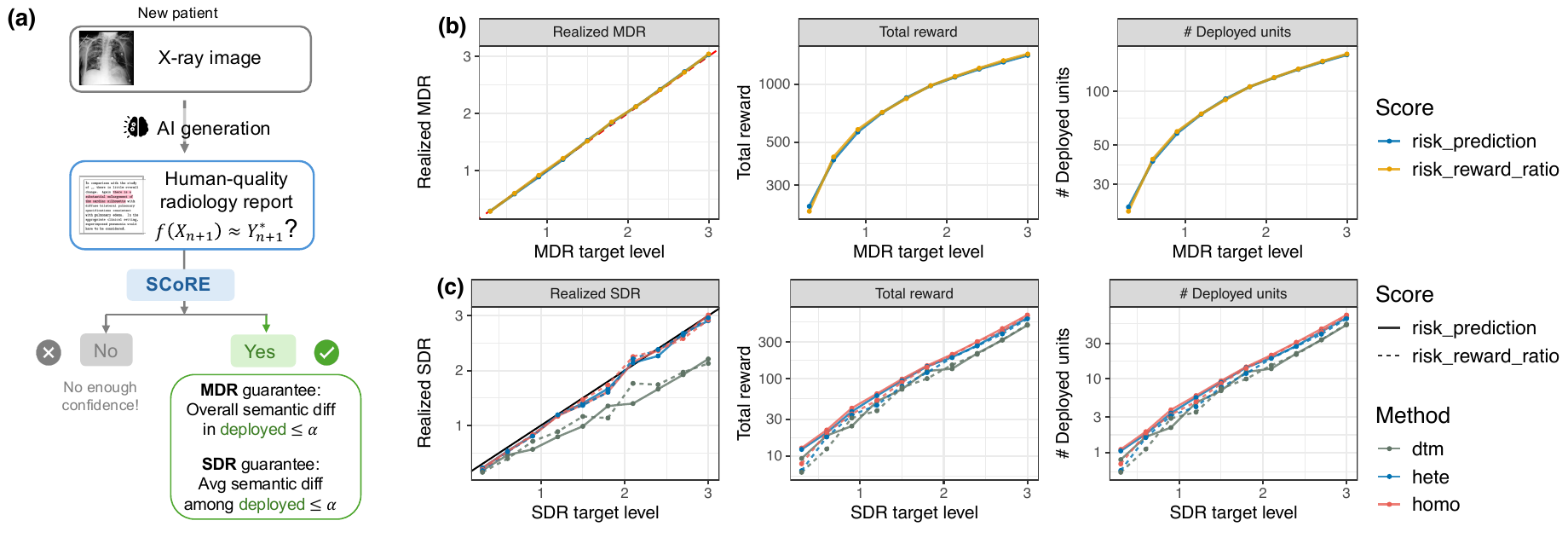}
    \caption{{\small SCoRE for identifying semantically coherent AI-generated radiology report. \textbf{(a) Overview}: The goal is to identify reports close to human-expert reports; SCoRE provides MDR and SDR guarantees among identified reports. \textbf{(b)  MDR control}: realized MDR at various target levels (left), total quality-based reward of deployed units, scaled by $1/m$ for readability (middle),  number of deployed units (right). \textbf{(c) SDR control}: realized SDR at various target levels (left), total quality-based reward of deployed units (middle), number of deployed units (right).}}
    \label{fig:llm}
\end{figure}

% !TEX root = main.tex
\section{Simulations} \label{sec:simu}

In addition to the real data applications, we conduct a series of simulation studies to comprehensively evaluate the SCoRE procedures. We focus on examining (i) validity of risk control under various settings, (ii) factors affecting the tightness of risk control, and (iii)  robustness under estimated weight functions.

\subsection{Simulation settings}

For both MDR and SDR control, we consider 2 distinct data generation processes (DGPs) adapted from \cite{jin2023selection} with nonlinear relationships, where the response variable is properly scaled to fit the current formulation. 
Each data generation process is assessed in two scenarios:
\vspace{0.25em}
\begin{enumerate}[label=(\roman*).]
    \item Exchangeable: both calibration and test samples are independently drawn from the same process.
    \item Covariate shift: the calibration data are drawn from the original process, while test data are generated from a reweighted version of the same process according to an unknown weight function $w$. 
\end{enumerate}
\vspace{0.25em}
Scenario (ii) is studied in Section~\ref{subsec:simu_weight}, where we use an estimator $\hat{w}$ in the SCoRE procedures. %Our simulation results thus assess the finite-sample performance with estimated weights to examine the theory in Theorem~\ref{thm:w_mdr_asymp} and Theorem~\ref{thm:w_sdr_asymp}.
Echoing the practical objectives in applied settings, we examine SCoRE with distinct risk functions:
\vspace{0.25em}
\begin{itemize}
    \item Excess risk: $\cL(f, X, Y) = Y \cdot \ind\{Y > c\}$ where $c$ is a pre-defined threshold;
    \item L2 risk: $\cL(f, X, Y) = (Y - f(X))^2$;
    \item Sigmoid risk: $\cL(f, X, Y) = \sigma(-\tau Y)$ where $\sigma(z) = 1/(1+e^{-z})$, and $\tau \in \RR^+  $ is a temperature parameter.
\end{itemize} 
\vspace{0.25em}

The excess risk is closely related to the expected shortfall~\citep{rockafellar2000optimization} which reflects the tail behavior of $Y$. The L2 risk, also used in Section~\ref{subsec:app_icu}, mirrors   selective prediction where a model $f$ should  be deployed in cases with sufficiently low   expected prediction error. 
The sigmoid risk can be viewed as a smooth relaxation of the indicator function $\ind\{Y > 0\}$ in \cite{jin2023selection}. Later, by varying the temperature parameter $\tau$, we examine how the distribution of the risk affects tightness of risk control. The details on the DGPs, weight functions,  predictive models, and score functions are  in Appendix~\ref{app:simu_detail}. 
% In the experiments, we fix the sample sizes as $n = 1000$ and vary the nominal SDR level $q \in \{0.05, 0.1, \dots, 0.5\}$.
% \todo{sample size ($n_\test=100$ for covariate shift?)}
% i wrote it in the 'marginal/selective risk control' paragraph

We consider two reward functions for each of the six combinations of the DGP and risk function: constant reward $r_0(X, Y) = 1$ and the squared reward $r_1(X, Y) = Y^2$. 
Given the risk estimator $\hat{l}(x)$  the reward estimator $\hat{r}_1(x)$ similar to the real data applications, we set the score function as either the predicted risk or the risk-reward ratio. Similar to Section~\ref{sec:real}, we  refer to the corresponding SCoRE procedures as the \texttt{risk\_prediction} and \texttt{risk\_reward\_ratio} variants respectively. The baseline methods under comparison are described in Section~\ref{subsec:simu_risk_control} for MDR and SDR control, respectively.
% According to the theoretical results, the \texttt{risk\_prediction} variant maximizes the average number of selection, whereas the \texttt{risk\_reward\_ratio} variant maximizes the average $y$-squared reward. 
% When the two score functions induce markedly different rankings of data points, we expect visible differences in their probability of deployment and reward.
% \subsection{Marginal risk control}

% \yingcomment{maybe add histograms of risk and reward? (how to display? maybe just for case 3 and vary $\tau$)}
 
\subsection{Risk control and power comparison}
\label{subsec:simu_risk_control}

We first verify the risk control of SCoRE procedures  without covariate shift, as well as validating that the design of the two score functions indeed perform as claimed in our theory.

\vspace{-0.5em}
\paragraph{Marginal risk control.}
We first evaluate the performance of SCoRE in MDR control tasks, using a calibration sample size of $n = 1000$ and averaging results over $m=100$ test samples in $N = 100$ independent runs. 
Besides SCoRE, we evaluate baselines based on uniform concentration inequalities for $\mdr(t):=\EE[\cL(f,X,Y)\ind\{s(X)\leq t\}]$. Namely, we set $\hat\psi_{n+1}= \ind\{s(X_{n+1})\leq \hat{t}\}$ for $\hat{t} = \max\{t\in \cG\colon  \widehat\mdr(t)+\epsilon_n\leq \alpha\}$  where $\widehat\mdr(t)=\frac{1}{n}\sum_{i=1}^n L_i\ind\{s(X_i)\leq t\}$, and $\epsilon_n$ is a slack computed by uniform concentration inequalities (\texttt{Hoeffding} and \texttt{Rademacher}) and $\cG$ is a search range; see Appendix~\ref{app:simu_baselines} for details. Strictly speaking, these baselines do not control MDR in theory  since the upper bound on $\mdr(t)$ holds with high probability, though we anticipate them to be overly-conservative.

\begin{figure} 
    \centering
    \includegraphics[width=\linewidth]{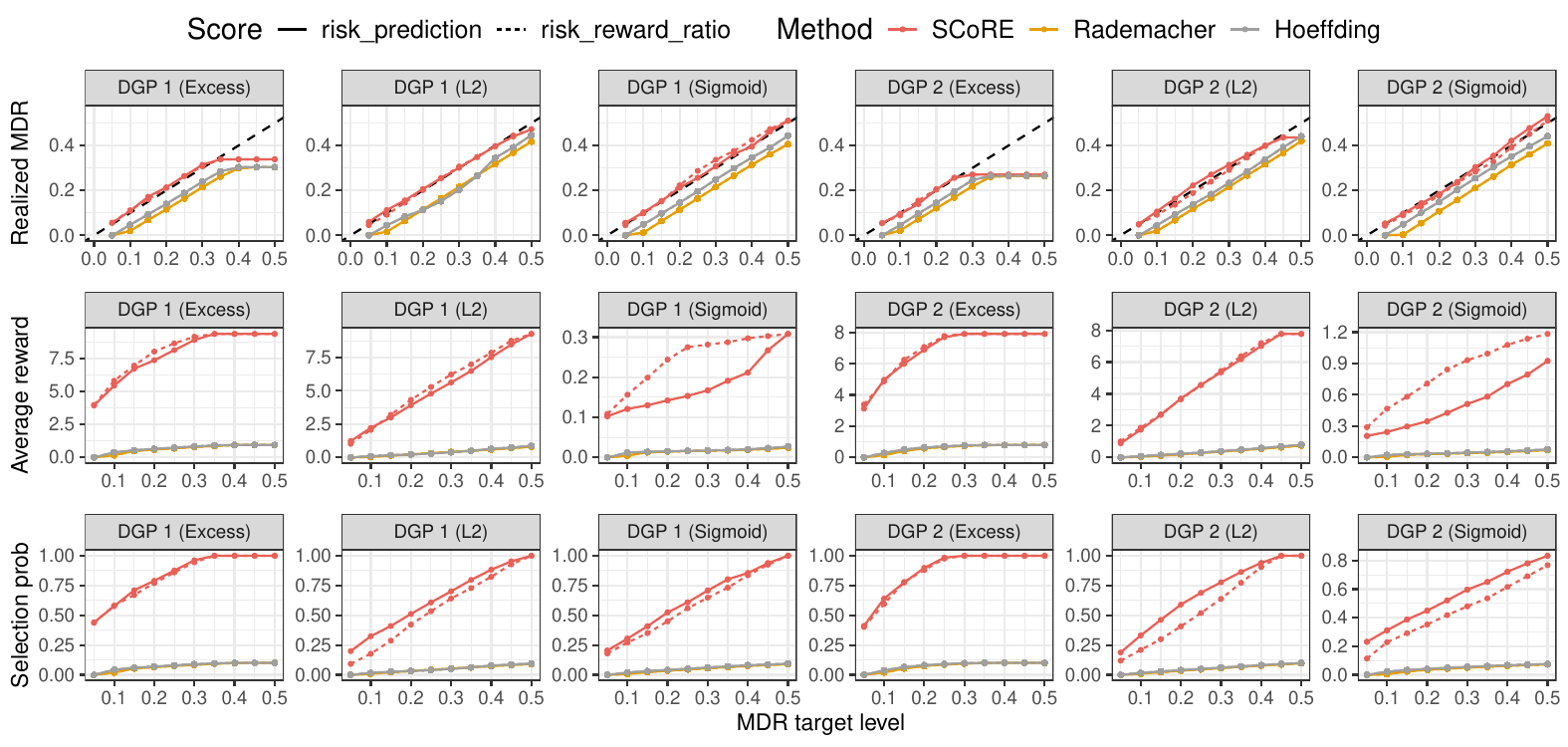}
    \caption{{\small Realized MDR, average reward and fraction of selection for varying nominal MDR levels under two DGPs and three risk functions. Each column corresponds to one pair of DGP and risk function. The dashed black line in the first row is $y=x$.}}  
    \label{fig:simu_mdr}
\end{figure}

Figure~\ref{fig:simu_mdr} presents the average realized MDR, average reward, and fraction of selection for both score function variants, as the nominal MDR level $q$ varies from 0.05 to 0.5 in increments of 0.05. Across all settings, both SCoRE variants demonstrate valid and tight MDR control. 
As anticipated, the \texttt{risk\_reward\_ratio} variant tends to achieve  a higher average reward, whereas the \texttt{risk\_prediction} variant yields a larger number of selections. The contrast between the two variants is most pronounced under the sigmoid loss function (where dividing by the predicted reward changes the ranking of units).  
These findings align with our theory in Theorem~\ref{thm:opt_mdr}. Compared with real applications, we conjecture the signal in the simulations is stronger, so the reward-aware score function makes a visible difference. Finally, the baseline methods based on concentration inequalities empirically control the MDR, yet  yield much lower power, showing the benefit of finite-sample exact MDR control via conformal calibration.

\vspace{-0.5em}
\paragraph{Selective risk control.}
For SDR control, the two choices of score functions are paired with three distinct e-value boosting methods, resulting in six variants in total. Besides SCoRE, we evaluate a baseline with $\hat\psi_{n+j} = \ind\{s(X_{n+j})\leq \hat{t}\}$ for $\hat{t}=\max\{t\in \cG\colon \widehat{\sdr}^+(t) \leq \alpha \}$, where $\widehat{\sdr}^+(t)$ is a uniformly valid upper bound on $\sdr^*(t):=\EE[\cL(f,X,Y)\given s(X)\leq t]$ with high probability, derived from uniform concentration inequalities (\texttt{Hoeffding} and \texttt{Rademacher}) detailed in Appendix~\ref{app:simu_baselines}. Again, these baselines provide high-probability, instead of exact, SDR control. We fix  $n = 1000$ and $m = 100$, and vary nominal level $q$    from $0.05$ to $0.5$ in increments of $0.05$. The results are averaged over $N = 100$ independent runs.

Figure~\ref{fig:simu_sdr} demonstrates that all of the six SCoRE variants maintain valid SDR control. While the deterministic boosting variants (\texttt{dtm}) tend to be overly conservative and fail to fully utilize the SDR budget, both the heterogeneous (\texttt{hete}) and homogeneous (\texttt{homo}) boosting variants achieve tight SDR control and higher power; the power is similar across \texttt{hete} and \texttt{homo}. %Consequently, the two boosting variants substantially outperform the deterministic method in both the number of selections and the achieved reward. 
Consistent with Theorem~\ref{thm:asymptotic_sdr}, the \texttt{risk\_prediction} and \texttt{risk\_reward\_ratio} variants outperform each other at their corresponding maximization targets, with the gap being most pronounced again under the sigmoid risk setting. Finally, the baselines using concentration inequalities are  conservative, leading to very low power and reward. This again shows the benefit of (near-) exact calibration via conformal inference. 

\begin{figure}[htbp]
    \centering
    \includegraphics[width=\linewidth]{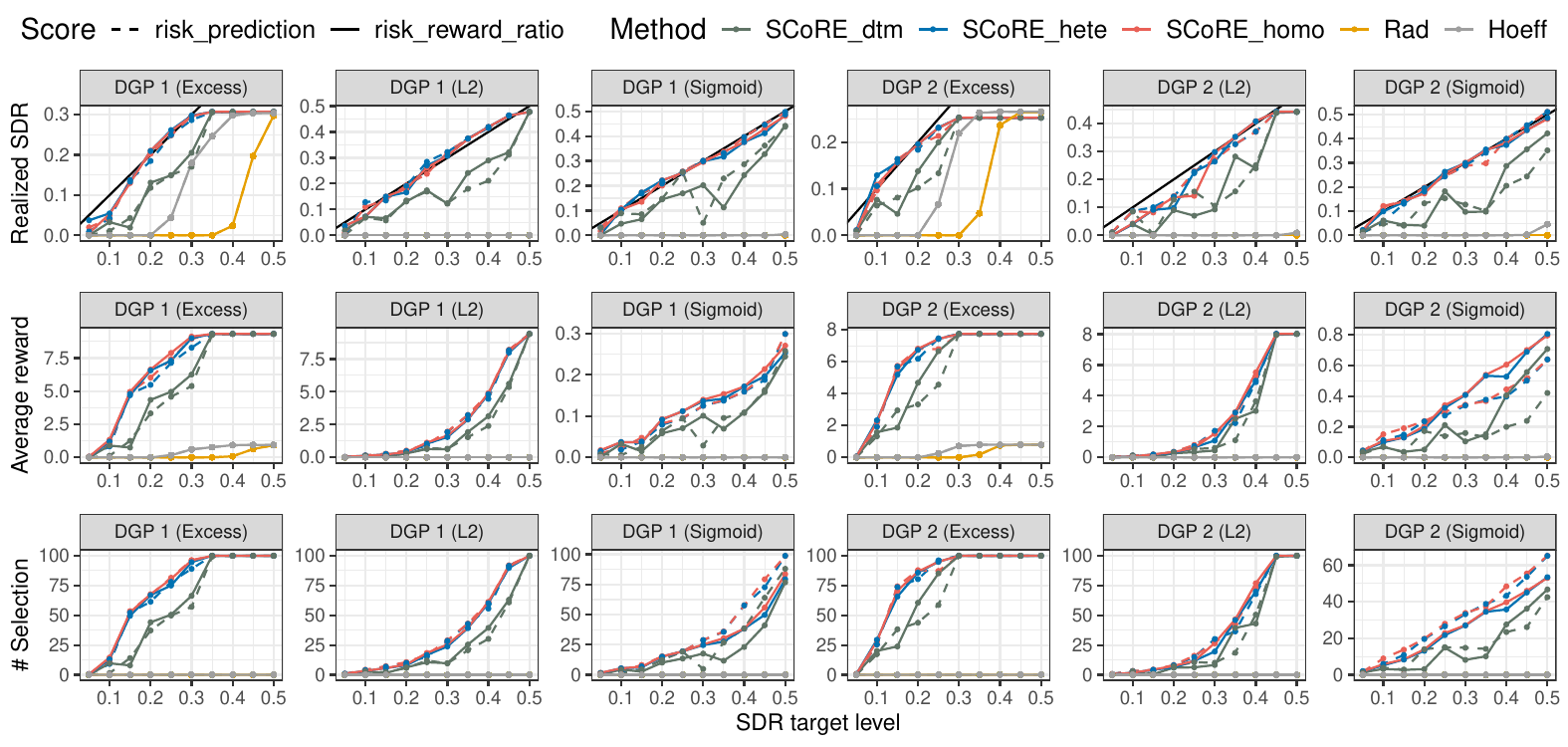}
    \caption{{\small Realized SDR, average reward and number of selection for varying nominal SDR levels. Each column corresponds to one pair of DGP and risk function. For subplots in the first row, the black line is $y=x$.}}
    \label{fig:simu_sdr}
\end{figure}

\subsection{Impact of risk distribution on tightness} 

Our MDR and SDR e-values takes the infimum over the unknown label value, and thus the MDR and SDR control may be slightly conservative since the inequality $\EE[L_{n+1}E_{n+1}]\leq 1$ may not be tight.  
By definition, the conservativeness of our e-values relies on whether the unknown $L_{n+1}$ attains the infimum, and Proposition~\ref{thm:sdr_conn_CS} confirms that it is the case for a binary risk function.  
On the other hand, if the calibration size $n$ is large enough, such conservativeness should be washed away by the law of large numbers. 
Our experiments vary these two aspects to study the tightness of SCoRE's error control.

In specific, we adopt the sigmoid risk function $\cL(f,x,y)=\sigma(-\tau y)$ while varying $\tau\in \{1,2,5,10,30\}$ to yield close approximation to the binary risk function $\ind\{y<0\}$ when $\tau$ is of a larger value. We also vary the calibration size $n\in \{100,300, 1000\}$ under the two DGPs.  
We evaluate our methods with two choices of score functions, with all the other details as before.

Figure~\ref{fig:simu_tau} reports the realized MDR (Panel a) and SDR (Panel b) for the two variants, respectively, averaged over  $N=100$ independent runs under each configuration. While maintaining the desired error control across all the settings, the conservativeness exhibits distinct patterns. 
In panel (a), we see that the MDR control is tight across settings, and the sample size and closeness to binary risk have no visible impact on the tightness. 
In panel (b), in contrast, increasing the value of $\tau$ or $n$ indeed tightens the error control. 
This could be attributed to the inherent structure of the eBH procedure, whose step-up rule induces interactions among the e-values.

\begin{figure}[htbp]
    \centering
    \includegraphics[width=\linewidth]{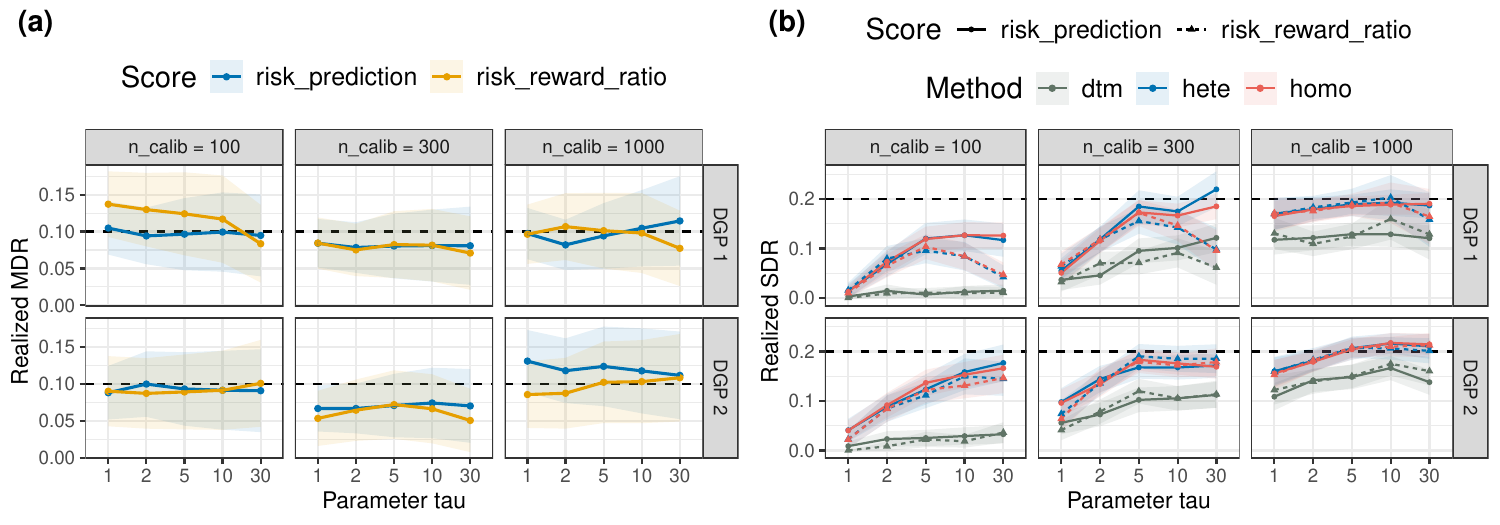}
    \caption{{\small MDR (left) and SDR (right) control when varying the parameter $\tau$ and calibration sample size $n$. Each row is a DGP and each column is a sample size. The nominal level is 0.1 for MDR and 0.2 for SDR. Details are otherwise the same as Figures~\ref{fig:simu_mdr} and~\ref{fig:simu_sdr}.}}
    \label{fig:simu_tau}
\end{figure}

\subsection{Robustness under covariate shift estimation}
\label{subsec:simu_weight}
Finally, we evaluate the robustness of the weighted variant of SCoRE with estimated weights when varying the complexity of weight models. 
We follow exactly the same evaluation procedures as before (using the homogeneous boosting variant for SDR control for conciseness), except that we employ  rejection sampling to create three unknown covariate shifts: (i) logistic model $w_1(x) = \text{sigmoid}(\theta^\top x)$ with $\theta_i=0.1 \cdot \ind\{i\leq 5\}$, (ii) a non-linear function with interactions $w_2(x) = \text{sigmoid}(0.5 (x_1x_2+ x_2x_3+ x_3x_4) + 0.3 \sin(x_1+x_2))$, and (iii) a multi-modal shift $w_3(x) = \text{sigmoid}(3\exp({-\|x'-a_1\|^2}) + 2.1 \exp({-\|x'-a_2\|^2}) - 2)$, where $x' = (x_1,x_2,x_3)$ denotes the first three entries of $x$ and $a_1 = (2,-1,1)$, $a_2 = (-2,1,-1)$. 
All the weight functions are estimated using probabilistic classification as in the previous settings.

% I just found that in (i) I was using weight 0.1 * \{i \leq 5\} .. that could be easier to predict (anyways (ii) and (iii) demonstrates the robustness)

The risk control of SCoRE is presented in Figure~\ref{fig:weight_mdr_sdr}. 
We observe robust MDR and SDR control with estimated weights when the true weights are of various complexity. 
For conciseness, we defer additional results on the power (number of selection and reward) of SCoRE to Appendix~\ref{app:simu_weight_res} and report the main messages here. Consistent with earlier observations, the \texttt{risk\_prediction} score leads to higher number of selections while the \texttt{risk\_reward\_ratio} score leads to higher total reward in deployed units (Appendix~\ref{app:simu_weight_res}). For the SCoRE-SDR variant, with covariate shifts, both homogeneous and heterogeneous boosting lead to comparable power, with substantial improvement over the deterministic version.

\begin{figure}[htbp]
    \centering
    \includegraphics[width=0.9\linewidth]{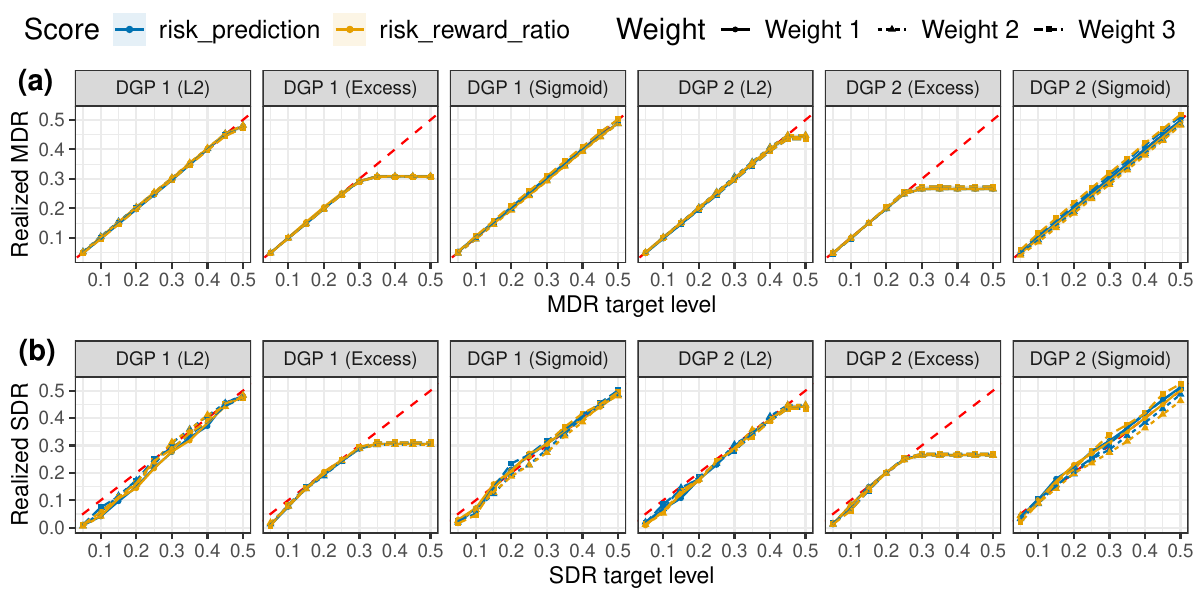}
    \caption{{\small The MDR (a) and SDR (b) control of SCoRE with estimated weights with two score functions under three weight models. Details are otherwise the same as Figure~\ref{fig:simu_mdr}.}}
    \label{fig:weight_mdr_sdr}
\end{figure}

% !TEX root = main.tex
\section{Discussion} \label{sec:discussion}

In this paper, we present SCoRE, a framework based on conformal inference and e-values to derive a selective trust mechanism for any prediction model with precise control of risks among trusted instances. We propose two complementary risk metrics, and show how each can be controlled by applying standard testing procedures to any ``risk-adjusted e-values''. We then propose concrete constructions of the e-values for each metric and analyze the optimal choice of scoring functions in these e-values. SCoRE's principles can be readily extended to settings with covariate shift. 
We demonstrate the utility of SCoRE in several real applications with diverse risk metrics, and conduct simulations to investigate factors that affect its performance. 

Several interesting directions remain open. First, while the asymptotic analysis offers guidance on the choice of score functions that determines which instances may be more trustworthy, it is naturally desirable to use data to optimize the scores. It is thus interesting to develop methods that allow rigorous risk control with data-driven score choices~\citep{bai2024optimized}. However, compared with the binary setting, maintaining validity with an unknown continuous test risk is substantially more challenging. 
% Second, the statistical properties of SCoRE under covariate shift may be worth further study. In both conformal prediction and conformal selection (a special case of SCoRE), existing research has established certain ``double-robustness'' properties, meaning that under covariate shift, the guarantees of these methods are approximately true when either the weights are consistently estimated or the scores converge to appropriate limits~\citep{lei2020conformal,jin2023model}. However, the thresholding form of SCoRE e-values makes such analysis less straightforward, which we leave for future research. 
In addition, the ideas of SCoRE may  extend  to  richer scenarios such as online settings where test instances arrive sequentially and real-time decisions need to be made, where the e-values might  be a useful tool~\citep{xu2024online}. It would also be interesting to apply SCoRE to selectively automate   workflows with tailored risks.

\section*{Acknowledgments}

The authors thank Ruth Heller for pointing out the connection to the multiple-family hypotheses testing problem and helpful discussions on the topic.

% \newpage
%%%%%%%%%%%%%%%%%%%%
%%% Bibliography %%%
%%%%%%%%%%%%%%%%%%%%
\bibliographystyle{apalike}
\bibliography{reference}

\newpage 
\appendix 
% !TEX root = main.tex

\section{Additional discussion}

\subsection{Connection between SCoRE and conformal risk control}
\label{app:subsec_connect_crc}

Here we continue the discussion on the equivalence of SCoRE and conformal risk control in Section~\ref{subsec:mgn_computation}. 
To see this, note that if SCoRE deploy the test instance, we have
\$
    \frac{1+\sum_{i=1}^n L_i \ind {\{s(X_i)\leq s(X_{n+1})\}}}{n+1} = \frac{1+\sum_{i=1}^n L_i(-s(X_{n+1}))}{n+1} \leq \alpha,
\$
so $-s(X_{n+1}) \geq \hat\lambda$ by definition of $\hat\lambda$. Conversely, $-s(X_{n+1}) \geq \hat\lambda$ implies that 
\$
    \frac{1+\sum_{i=1}^n L_i(-s(X_{n+1}))}{n+1} \leq \frac{1+\sum_{i=1}^n L_i(\hat\lambda)}{n+1} \leq \alpha
\$
by monotonicity of $L_i$, and the test instance is deployed by SCoRE.

\subsection{A simpler construction of e-values}
\label{app:subsec_alt_evalue}

Following conformal inference ideas, 
we will use the exchangeability among data to construct the e-values. At the same time, we use the estimation idea in the Benjamini-Hochberg procedure to ensure tight selective risk control. The idea is to set 
\$
e_j = \frac{\ind\{s(X_{n+j}) \leq \hat{t}_j\} }{\sum_{\ell=1}^m \ind\{s(X_{n+\ell}) \leq \tilde{t}_j\}}\cdot \frac{m}{\alpha}.
\$
where $\hat{t}_j$ and $\tilde{t}_j$ are stopping times chosen such that $\EE[L_{n+j}e_j]\leq 1$. 

Specifically, we set $\hat{t}_j\leq t_j(L_{n+j}) \leq \tilde{t}_j$ for some function $t_j(y)$ obeying 
\$
\EE\Bigg[ \frac{L_{n+j}\ind\{s(X_{n+j}) \leq t_j(L_{n+j})\} }{\sum_{\ell=1}^m \ind\{s(X_{n+\ell}) \leq t_j(L_{n+j})\}}   \Bigg] \leq \frac{\alpha}{m}. 
\$
Specifically, we set 
\$
\hat{t}_j = \max\Big\{t\colon \hat{\textrm{FR}}(t) \leq \alpha\Big\}, \quad \textrm{where}\quad \hat{\textrm{FR}}(t) = \frac{\ind\{s(X_{n+j})\leq t\} + \sum_{i=1}^n L_i\ind\{s(X_i)\leq t\}}{\sum_{\ell=1}^m \ind\{s(X_{n+\ell}) \leq t\}} \cdot \frac{m}{n+1}.
\$
On the other hand, 
\$
\tilde{t}_j = \max\Big\{t\colon \tilde{\textrm{FR}}(t) \leq \alpha\Big\}, \quad \textrm{where}\quad \tilde{\textrm{FR}}(t) = \frac{\sum_{i=1}^n L_i\ind\{s(X_i)\leq t\}}{ \sum_{\ell=1}^m \ind\{s(X_{n+\ell}) \leq t\}}\cdot \frac{m}{n+1}.
\$
For the purpose of proof, we define 
\$
t_j(\ell) = \max\Big\{t\colon  {\textrm{FR}}(t;\ell) \leq \alpha\Big\}, \quad \textrm{where}\quad {\textrm{FR}}(t;\ell) = \frac{\ell \ind\{s(X_{n+j})\leq t\} + \sum_{i=1}^n L_i\ind\{s(X_i)\leq t\}}{\sum_{\ell=1}^m \ind\{s(X_{n+\ell}) \leq t\}}\cdot \frac{m}{n+1}.
\$
By definition, we know that for any $\ell\in [0,1]$, 
\$
\hat{\textrm{FR}}(t) \geq  {\textrm{FR}}(t;\ell) \geq \tilde{\textrm{FR}}(t) ,
\$
and hence $\hat{t}_j \leq t_j(\ell) \leq \tilde{t}_j$. 
Therefore, writing $T_j = t_j(L_{n+j})$, we have 
\$
\EE[L_{n+j}e_j] &= \frac{m}{\alpha}\cdot \EE\Bigg[L_{n+j} \frac{\ind\{s(X_{n+j}) \leq \hat{t}_j\} }{\sum_{\ell=1}^m \ind\{s(X_{n+\ell}) \leq \tilde{t}_j\}}\Bigg]  \\ 
&\leq  \frac{m}{\alpha}\cdot\EE\Bigg[L_{n+j} \frac{\ind\{s(X_{n+j}) \leq  {t}_j(L_{n+j})\} }{\sum_{\ell=1}^m \ind\{s(X_{n+\ell}) \leq {t}_j(L_{n+j})\}}\Bigg] \\ 
&\leq  \frac{m}{\alpha}\cdot\EE\Bigg[ \frac{L_{n+j}\ind\{s(X_{n+j}) \leq T_j\} }{L_{n+j}\ind\{s(X_{n+j}) \leq  T_j\} 
 + \sum_{i=1}^n L_i\ind\{s(X_{i}) \leq T_j\}}\Bigg] \cdot \frac{\alpha(n+1)}{m}.
\$ 
Note that $T_j$ is permutation invariant to $(Z_1,\dots,Z_n,Z_{n+j})$ where $Z_{n+j}=(X_{n+j},Y_{n+j})$. Thus, by exchangeability, we have 
\$
\EE[L_{n+j}e_j] =  \frac{m}{\alpha}\cdot \frac{1}{n+1}  \cdot \frac{\alpha(n+1)}{m} \leq 1.
\$

\subsection{Computation shortcuts of SCoRE under distribution shift}
\label{app:subsec_shortcut_w}

Proposition~\ref{prop:w_single_equiv} presents the computation shortcuts for MDR control in SCoRE under covariate shift, parallel to Proposition~\ref{prop:single_equiv}. 

\begin{prop} 
\label{prop:w_single_equiv}
    For $\gamma \leq \alpha$, we have 
    \$
    \ind\{E_{\gamma,n+1}\geq 1/\alpha\} = \ind \bigg\{  \frac{w(X_{n+1}) +\sum_{i=1}^n w(X_i) L_i \ind {\{s(X_i)\leq s(X_{n+1})\}}}{\sum_{i=1}^{n+1} w(X_i)} \leq \gamma \bigg\}.
    \$
    For $\gamma > \alpha$, we have
    \$
    \ind\{E_{\gamma,n+1}\geq 1/\alpha\} = \ind \bigg\{  &\frac{w(X_{n+1}) +\sum_{i=1}^n w(X_i) L_i \ind {\{s(X_i)\leq s(X_{n+1})\}}}{\sum_{i=1}^{n+1} w(X_i)} \leq \gamma,~ \\
    &\quad \text{and} \quad \frac{\ell \cdot w(X_{n+1}) +\sum_{i=1}^n w(X_i) L_i \ind  \{s(X_i)\leq t\} }{\sum_{i=1}^{n+1} w(X_i)} \notin (\alpha,\gamma\big], ~\forall t\in \cM, \ell\in [0,1]
    \bigg\}.
    \$
\end{prop}

Similarly, Proposition~\ref{prop:mult_equiv_w} gives the covariate shift analogue of Proposition~\ref{prop:mult_equiv} and Algorithm~\ref{alg:sel}. 

\begin{algorithm}
    \small
    \captionsetup{font=small}
    \caption{Efficient computation of e-values for SDR control under covariate shift}
    \label{alg:w_sel}
\begin{algorithmic}[1]
    \REQUIRE{Labeled data $\{(X_i, Y_i)\}_{i=1}^n$, test data $\{X_{n+j}\}_{j=1}^m$, pretrained predictor $s$, covariate shift weights $w$.} \\[0.5ex]
    \STATE Compute the true calibration risks $L_i = \cL(f, X_i, Y_i)$ for $i = 1, \dots, n$.
    \STATE Obtain the predicted risks for calibration and test data $\cM:=\{s(X_i)\}_{i=1}^{n+m}$.
    \FOR{$j=1, \dots, m$}  
        \STATE For all $t \in \cM$, compute 
    $$ \bar{\ell}(t)=\frac{\gamma}{m} \frac{\sum_{i=1}^n w(X_i) + w(X_{n+j}) }{w(X_{n+j})} \big(1 + \sum_{\ell \neq j}\ind\{s(X_{n+\ell}) \leq t\} \big) - \sum_{i=1}^n \frac{w(X_i)}{w(X_{n+j})} L_i \ind\{s(X_{i}) \leq t\}.$$
        \STATE Compute the thresholds $t_{\gamma, n+j}(0)$ and $t_{\gamma, n+j}(1)$.
        \IF{$s(X_{n+j}) > t_{\gamma, n+j}(1)$}
            \STATE Set $E_{\gamma, n+j} = 0$.
        \ELSIF{$t_{\gamma, n+j}(0) = t_{\gamma, n+j}(1)$}
            \STATE $\displaystyle
            \textnormal{Set } E_{\gamma, n+j} = \frac{\sum_{i=1}^n w(X_i) + w(X_{n+j})}{w(X_{n+j}) + \sum_{i=1}^n w(X_i)\cdot L_i \ind\{s(X_i) \leq t_{\gamma, n+j}(1)\} }.
            $
        \ELSE
            \STATE Initialize the set
            $\displaystyle 
                \cM^* = \{t \in \cM: t \geq s(X_{n+j})~\text{and}~ \fr_{n+j}(t; 0) \leq \gamma\} \cap [t_{\gamma, n+j}(1), t_{\gamma, n+j}(0)].
            $
            \STATE Remove all element $t \in \cM^*$ if there exists any $t' \in \cM, t' > t, \fr(t';0) \leq \gamma$ such that $\ell(t') > \ell(t)$.
            \STATE $\displaystyle
            \textnormal{Set } E_{\gamma, n+j} = \inf_{t \in \cM^*}\frac{\sum_{i=1}^n w(X_i) + w(X_{n+j})}{ w(X_{n+j}) \cdot \bar\ell(t) + \sum_{i=1}^n w(X_i)\cdot L_i \ind\{s(X_i) \leq t\} }.
            $
        \ENDIF
    \ENDFOR  
    \ENSURE{The computed e-values $\{E_{\gamma, n+j}\}_{j=1}^m$.}
\end{algorithmic}
\end{algorithm}

\begin{prop} \label{prop:mult_equiv_w}
    The output of Algorithm~\ref{alg:w_sel} equals $E_{\gamma,n+j}$ defined in~\eqref{eq:def_eval_sel_w}, whose computation complexity is at most $\cO((n+m)m + (n+m)\log(n+m) )$.
\end{prop}

The proof of the Propositions can be found in Appendix Section~\ref{sec:prop_w_single_equiv} and Section~\ref{app:subsec_mult_equiv_w} respectively.

\subsection{Doubly robust calibration of MDR under covariate shift}
\label{app:subsec_mdr_dr_discussion}

In this part, we present a general approach to achieve double robustness in MDR control under unknown covariate shift when multiple samples from the test distribution $Q$ are available. 
The key idea is to use an estimate $\hat{l}(x)$ for the conditional risk $l(x):=\EE[\cL(f,X,Y)\given X=x]$ and calibrate the weights to satisfy finite-sample balance. 

The following assumption posits that the estimated weights must enforce a finite-sample balance condition in the thresholded, estimated loss, serving to protect against weight misspecification.  

\begin{assumption}\label{assump:balance_mdr}
    $\{\hat{w}_i\}_{i=1}^{n+m}$ obey 
    % $\sum_{i=1}^n \hat{w}_i=O_P(n)$ and % this is implied by the last condition
    the following approximate balancing condition:
    \$
    &\frac{1}{n}\sum_{i=1}^n \hat{w}_i \hat{l}(X_i)\ind\{s(X_i)\leq \hat{t}\} = \frac{1}{m}\sum_{j=1}^m \hat{l}(X_{n+j}) \ind\{s(X_{n+j})\leq \hat{t}\} + o_P(1),\quad \frac{1}{n}\sum_{i=1}^n \hat{w}_i  = 1+o_P(1),
    \$
    where $\hat{t} = \sup\{t\colon \frac{1}{m}\sum_{j=1}^m \hat{l}(X_{n+j}) \ind\{s(X_{n+j})\leq  {t}\} \leq \alpha\}$.  
\end{assumption}

Intuitively, the ``correct'' weights $w(X_i)$ should balance the empirical mean of any function across the two groups. While this might be difficult to achieve especially with misspecified weights, Assumption~\ref{assump:balance_mdr}  enforces this balance at a specific function: the estimated weights to have an equal (reweighted) mean between two groups at the estimate cutoff $\hat{t}$. Asymptotically, this ensures the unknown MDR under the test distribution is well approximated by the reweighted calibration data even though the weights may be misspecified.  

Given certain preliminary estimators $\hat{w}(\cdot)$ and $\hat{l}(\cdot)$, Assumption~\ref{assump:balance_mdr} can be fulfilled by weight calibration via an efficient covariate balancing procedure, see, e.g.,~\cite{zubizarreta2015stable,jin2025cross}. 

Theorem~\ref{thm:dr_mdr} ensures the double robustness in MDR control, in the sense that as long as either the weight or the conditional risk model is correct, we obtain asymptotic MDR control. Its proof is in Appendix~\ref{subsec:dr_mdr}.

\begin{theorem}\label{thm:dr_mdr}
    Take $\gamma=\alpha$, and suppose $\hat{l}(\cdot)$ is trained independent of the calibration and test data, and $s(X)$ has no point mass. Suppose Assumption~\ref{assump:balance_mdr} holds, and assume $\frac{1}{m}\sum_{i=1}^m (\hat{w}_i-\bar{w}(X_i))^2=o_P(1)$ %\tiancomment{$n$ or $m$ here?} \yingcomment{m; thanks!}
    and $\|\hat{l} - \bar{l}\|_{L_2(\PP_X)} = o_P(1)$ for some fixed functions $\bar{w}\colon \cX\to \RR$ and $\bar\ell \colon \cX\to \RR$, and $\sup_i \hat{w}_i\leq M$ for a fixed constant $M>0$. 
    In addition, denoting $G(t) =\EE_P[\bar{w}(X)l(X) \ind\{s(X)\leq t\}]$, we assume $G(t)$ is continuous and strictly increasing at $t^*:=\sup\{t\colon G(t)/\EE_P[\bar{w}(X)]\leq \alpha\}$.  
    Also assume the mapping $t\mapsto \EE_Q[\bar{l}(X)\ind\{s(X)\leq t\}]$ is continuous and strictly increasing at $t^\dagger:=\sup\{t\colon \EE_Q[\bar{l}(X)\ind\{s(X)\leq t\}]\leq \alpha\}$.
    Let $\textnormal{MDR}_{n,m}$ be the MDR of SCoRE with estimated weights $\hat{w}_i$. Then, we have $\limsup_{n,m\to\infty} \textnormal{MDR}_{n,m} \leq \alpha$ under either of the two conditions:
    \begin{enumerate}[label=(\roman*)]
        \item $\bar{w}(\cdot) = w(\cdot)$, i.e., the weights are consistent. 
        \item $\bar{l}(\cdot) = l(\cdot)$, i.e., the risk model is consistent.
    \end{enumerate}
\end{theorem}

Theorem~\ref{thm:dr_mdr} operates under the model convergence conditions $\frac{1}{n}\sum_{i=1}^m (\hat{w}_i-\bar{w}(X_i))^2=o_P(1)$ and $\|\hat{l} - \bar{l}\|_{L_2(\PP_X)} = o_P(1)$. If the weights $\{\hat{w}_i\}$ are obtained by a balancing approach with the data-driven features $\phi(x) = (\hat{w}(x), \hat{l}(x)\ind\{s(x)\leq \hat{t}\} )$ like in~\cite{jin2025cross}, the first condition would be fulfilled if the preliminary weight function $\hat{w}(\cdot)$ converges to any fixed function. The proof follows exactly the same idea as~\citet[Theorem 3.1]{jin2025cross}, which we omit here for brevity. 

Besides the standard model convergence conditions, Theorem~\ref{thm:dr_mdr} posits a mild condition on the limiting weighted risk function $G(t)$, which ensures $\hat{t}$ stabilizes at a constant to facilitate analysis. This can be ensured if $s(X)$ has continuous support and has no point mass, e.g., by adding tiny random perturbations.

\subsection{Doubly robust calibration of SDR under covariate shift}
\label{app:subsec_sdr_dr}

In this part, we present a strategy for doubly robust calibration of SDR control. To protect against weight misspecification, we impose the following balancing condition on the estimated weights. Compared with the MDR version, here we enforce balance at a distinct cutoff $\hat{t}$ that is relevant to the SDR. 

\begin{assumption}\label{assump:balance_sdr}
    $\{\hat{w}_i\}_{i=1}^{n+m}$ obey the following approximate balancing condition:
    \$
    &\frac{1}{n}\sum_{i=1}^n \hat{w}_i \hat{l}(X_i)\ind\{s(X_i)\leq \hat{t}\} = \frac{1}{m}\sum_{j=1}^m \hat{l}(X_{n+j}) \ind\{s(X_{n+j})\leq \hat{t}\} + o_P(1),\quad \frac{1}{n}\sum_{i=1}^n \hat{w}_i  = 1+o_P(1),
    \$
    where $\hat{t} = \sup\{t\colon\frac{1}{n}\sum_{i=1}^n \hat{w}_i \hat{l}(X_i)\ind\{s(X_i)\leq  {t}\} / (1\vee \sum_{j=1}^m  \ind\{s(X_{n+j})\leq  {t}\}) \leq \alpha\}$.
\end{assumption}

The balancing condition protects against weight misspecification and can ensure SDR control if the conditional risk is consistently estimated. The proof of Theorem~\ref{thm:dr_sdr} is in Appendix~\ref{app:subsec_proof_dr_sdr}. 

\begin{theorem}\label{thm:dr_sdr}
    Take $\gamma=\alpha$, and suppose $\hat{l}(\cdot)$ is trained independent of the calibration and test data, and $s(X)$ has no point mass. Suppose Assumption~\ref{assump:balance_sdr} holds, and assume  $\frac{1}{m}\sum_{i=1}^m (\hat{w}_i-\bar{w}(X_i))^2=o_P(1)$, $\|\hat{l} - \bar{l}\|_{L_2(\PP_X)} = o_P(1)$ for some fixed functions $\bar{w}\colon \cX\to \RR$ and $\bar\ell \colon \cX\to \RR$, and $\sup_i \hat{w}_i\leq M$ for a fixed constant $M>0$. Define
    \$
    \bar{F}(t) = \frac{\EE_P[\bar{w}(X)L\ind\{s(X)\leq t\}]}{ \PP_Q(s(X)\leq t)\cdot \EE_P[\bar{w}(X)]}.
    \$
    % $ \bar{F}(t) = \frac{\EE_P[\bar{w}(X)L\ind\{s(X)\leq t\}]}{ \PP(s(X_{n+j})\leq t)\cdot \EE_P[\bar{w}(X)]}$. 
    Suppose $\bar{F}(t)$ is continuous at $t^* = \sup\{t\colon \bar{F}(t)\leq \alpha\}$, and for any sufficiently small $\epsilon>0$, there exists some $t\in (t^*-\epsilon,t^*)$ such that $\bar{F}(t)<\alpha$. Let $\sdr_{n,m}$ be the SDR of SCoRE with estimated weights $\{\hat{w}_i\}$. Then $\limsup_{n,m}\sdr_{n,m}\leq \alpha$ under either of the two conditions:
    \begin{enumerate}[label=(\roman*)]
        \item $\bar{w}(\cdot) = w(\cdot)$, i.e., the weights are consistent. 
        \item $\bar{l}(\cdot) = l(\cdot)$, i.e., the risk model is consistent.
    \end{enumerate}
\end{theorem}

Apart from the standard convergence conditions, the conditions in Theorem~\ref{thm:asymptotic_sdr} on $\bar{F}(t)$ is standard in the literature~\citep{storey2004strong,jin2023selection,jin2023model} which ensures the selection cutoff in the (e-)BH procedure stabilizes around a constant value. Following~\cite{jin2025cross}, the convergence condition on $\{\hat{w}_i\}$ holds if one fulfills Assumption~\ref{assump:balance_sdr} by running a covariate-balancing program with balancing feature $(\hat{w}(x), \hat{l}(x)\ind\{s(x)\leq \hat{t}\})$ using a preliminary estimated weight function $\hat{w}(\cdot)$ that converges in $L_2$-norm to some fixed function.

\section{Technical proof}

\subsection{Proof of Theorem~\ref{thm:single}}
\label{app:subsec_thm_single}

\begin{proof}[Proof of Theorem~\ref{thm:single}]
% Suppose $E_{\gamma, n+1}\geq 0$ obeys $\EE[Y_{n+1}\cdot E_{\gamma, n+1}]\leq 1$, then by definition, since $Y_{n+1}$ is nonnegative, we have 
% \$
% \EE[Y_{n+1}\ind\{\psi(X_{n+1}) = 1\}]  
% =  \EE[Y_{n+1}\ind\{E_{\gamma, n+1}\geq 1/\alpha\}] \leq \alpha \cdot \EE[Y_{n+1}\cdot E_{\gamma, n+1}] \leq \alpha.
% \$
% We now prove $\EE[Y_{n+1}\cdot E_{\gamma, n+1}]\leq 1$. 
By definition, since $L_{n+1}\in [0,1]$, 
\$
\EE[L_{n+1} E_{\gamma, n+1}] &= 
\EE\Bigg[L_{n+1} \cdot \inf_{\ell \in[0,1]} ~\Bigg\{ \frac{(n+1)\cdot \ind\{s(X_{n+1})\leq t_\gamma(\ell)\}}{\sum_{i=1}^n L_i \ind\{s(X_{i})\leq t_{\gamma}(\ell)\} + \ell\ind\{s(X_{n+1})\leq t_\gamma(\ell)\} } \Bigg\} \Bigg]
\\
&\leq \EE\Bigg[ \frac{ L_{n+1} \cdot (n+1) \ind\{s(X_{n+1})\leq T_{\gamma,n+1}\}}{\sum_{i=1}^n L_i \ind\{s(X_{i})\leq T_{\gamma,n+1}\} + L_{n+1}\ind\{s(X_{n+1})\leq T_{\gamma,n+1}\} }  \Bigg],
\$
where $T_{\gamma,n+1}:= t_{\gamma}(L_{n+1}) = \max \big\{t\colon \textrm{F}(t, L_{n+1} )\leq \gamma \big\}$, and 
\$ 
\textrm{F}( t, L_{n+1} ) =  \frac{\sum_{i=1}^n L_i \ind\{s(X_{i})\leq t\} + L_{n+1}\ind\{s(X_{n+1})\leq t\} }{n+1}.
\$
Note that $T_{\gamma,n+1}$ is invariant to permutations of $(Z_1,\dots,Z_n,Z_{n+1})$ where $Z_i=(X_i,Y_i)$. 
Therefore, $T_{\gamma,n+1}$ is determined if we condition on the unordered set $[\cZ]=[Z_1,\dots,Z_n,Z_{n+1}]$. In addition, for any fixed values $z_1,\dots,z_{n+1}$, conditional on the event $[\cZ] = [z_1,\dots,z_{n+1}]$, the data sequence follows the distribution 
\$
(Z_1,\dots,Z_{n+1})\biggiven \big\{[\cZ]=[z_1,\dots,z_{n+1}] \big\} \quad \sim \quad \frac{1}{(n+1)!}\sum_{\sigma\in S_{n+1}} \delta_{(z_{\sigma(1)},\dots,z_{\sigma(n+1)})},
\$
where $\delta_x$ is the point mass at $x$, and $S_{n+1}$ is the collection of all permutations of $\{1,\dots,n+1\}$. 
Altogether, these imply that for any fixed values $[z_1,\dots,z_{n+1}]$ where $z_i = (x_i,y_i)$, 
\$
&\EE\Bigg[ \frac{ L_{n+1}\cdot (n+1) \ind\{s(X_{n+1})\leq T_{\gamma,n+1}\}}{\sum_{i=1}^n L_i \ind\{s(X_{i})\leq T_{\gamma,n+1}\} + L_{n+1}\ind\{s(X_{n+1})\leq T_{\gamma,n+1}\} } \Bigggiven [\cZ] =[z_1,\dots,z_{n+1}]\Bigg] \\ 
&= \frac{1}{(n+1)!} \sum_{\sigma \in S_{n+1}} \frac{(n+1) \ell_{\sigma(n+1)} \ind\{ s(x_{\sigma(n+1)}) \leq T_{\gamma,n+1} \} }{\sum_{i=1}^{n+1} \ell_i \ind\{s(x_i)\leq T_{\gamma,n+1}\}} \\ 
&=\frac{1}{(n+1)!} \sum_{\sigma \in S_{n+1}} \sum_{j=1}^{n+1} \frac{(n+1) \ind\{\sigma(n+1)=j\} \cdot \ell_{j} \ind\{ s(x_j) \leq T_{\gamma,n+1} \} }{\sum_{i=1}^{n+1} \ell_i \ind\{s(x_i)\leq T_{\gamma,n+1}\}} \\ 
&= \frac{1}{(n+1)!}  \sum_{j=1}^{n+1} \frac{(n+1)n! \cdot \ell_{j} \ind\{ s(x_j) \leq T_{\gamma,n+1} \} }{\sum_{i=1}^{n+1} \ell_i \ind\{s(x_i)\leq T_{\gamma,n+1}\}} = 1,
\$
where $\ell_i := L(f, x_i, y_i)$ and $T_{\gamma,n+1}$ is a function of $[z_1,\dots,z_{n+1}]$. 
Then by the tower property, we know that $\EE[L_{n+1}\cdot E_{\gamma,n+1}]\leq 1$,  which concludes the proof.
\end{proof}

\subsection{Proof of Proposition~\ref{prop:single_equiv}}
\label{app:subsec_single_equiv}

\begin{proof}[Proof of Proposition~\ref{prop:single_equiv}]

Fix any $\gamma \in(0,1)$. We first observe that 
\$
    E_{\gamma, n+1} \geq 1/\alpha \iff s(X_{n+1}) \leq t_{\gamma}(\ell), \text{ and } \textrm{F}(t_{\gamma}(\ell),\ell) \leq \alpha \text{ for any } \ell \in [0,1]. \tag{$*$}
\$
% \yingcomment{But we are interested in the event $E_{\gamma,n+1}\geq 1/\alpha$ for risk control at $\alpha$?} 
Indeed, for any $\ell$, if $s(X_{n+1}) \leq t_\gamma(\ell)$ and $\textrm{F}(t_{\gamma}(\ell),\ell) \leq \alpha$, then by definition, 
\$
    \ind\{s(X_{n+1}) \leq t_\gamma(\ell)\} / \textrm{F}(t_\gamma(\ell); \ell) = 1 / \textrm{F}(t_\gamma(\ell); \ell) \geq 1/\alpha.
\$
Expanding the left hand side, this is equivalent to
\$
    \ind\{s(X_{n+1}) \leq t_\gamma(\ell)\} / \textrm{F}(t_\gamma(\ell); \ell) =  \frac{(n+1)\cdot \ind\{s(X_{n+1})\leq t_\gamma(\ell)\}}{\sum_{i=1}^n L_i \ind\{s(X_{i})\leq t_{\gamma}(\ell)\} + \ell\ind\{s(X_{n+1})\leq t_\gamma(\ell)\} } \geq 1/\alpha.
\$ 
and if above inequality holds for any $\ell$, clearly we have
\$
    E_{\gamma,n+1} = \inf_{\ell \in[0,1]} ~\Bigg\{ \frac{(n+1)\cdot \ind\{s(X_{n+1})\leq t_\gamma(\ell)\}}{\sum_{i=1}^n L_i \ind\{s(X_{i})\leq t_{\gamma}(\ell)\} + \ell\ind\{s(X_{n+1})\leq t_\gamma(\ell)\} } \Bigg\} \geq 1/\alpha.
\$
We note that in above derivation, it is implicit that $t_\gamma(\ell) \neq -\infty$ as otherwise $s(X_{n+1}) \leq t_\gamma(\ell)$ cannot possibly be true. Therefore, the e-value follows the usual definition and not is forced to be zero.
We show the other direction by taking the contrapositive. If for some $\ell \in (0,1)$, $s(X_{n+1}) > t_{\gamma}(\ell)$, the infimum (and thus the e-value) is clearly zero.
On the other hand, if $\textrm{F}(t_{\gamma}(\ell),\ell) > \alpha$ for some $\ell$, then
\$
    \ind\{s(X_{n+1}) \leq t_\gamma(\ell)\} / \textrm{F}(t_\gamma(\ell); \ell) \leq 1 / \textrm{F}(t_\gamma(\ell); \ell) \leq 1 / \alpha,
\$
and we establish the equivalence.

We continue the proof by examining the two events: $s(X_{n+1}) \leq t_{\gamma}(\ell)$ and $\textrm{F}(t_{\gamma}(\ell),\ell) \leq \alpha$. Fix any $\ell \in [0,1]$, for the first event we observe that
\$
    s(X_{n+1}) \leq t_{\gamma}(\ell) \iff \textrm{F}(s(X_{n+1}), \ell) \leq \gamma.
\$
This fact is due to the definition of $t_\gamma(\ell) := \max\{t \in\cM: \textrm{F}(t;\ell) \leq \gamma\}$. Given that $\textrm{F}(s(X_{n+1}), \ell) \leq \gamma$, we have $t_\gamma(\ell) \neq -\infty$, and the left hand side automatically holds by definition. The other direction follows from the (non-decreasing) monotonicity of $\textrm{F}$ in the first argument together with the fact $s(X_{n+1}) \in \cM$, as $\textrm{F}(s(X_{n+1}), \ell) \leq \textrm{F}(t_\gamma(\ell),\ell) \leq \gamma$.

Above equivalence clearly continues to hold if all $\ell$ are considered at the same time, i.e.,
\$
    \forall~ \ell \in [0,1], s(X_{n+1}) \leq t_{\gamma}(\ell) \iff \forall~ \ell \in [0,1],\textrm{F}(s(X_{n+1}), \ell) \leq \gamma.
\$
By monotonicity of $\textrm{F}$ in the second argument, the right hand side is equivalent to $\textrm{F}(s(X_{n+1}), 1) \leq \gamma$. This condition is in turn equivalent to
\@
\frac{1+\sum_{i=1}^n L_i \ind\{s(X_{i})\leq s(X_{n+1})\}}{n+1} \leq \gamma. 
\@
% why no equation tag here??

For the second event, we first observe that this event automatically holds if $\gamma \leq \alpha$, as $\textrm{F}(t_\gamma(\ell), \ell) \leq \gamma \leq \alpha$ by definition.
Otherwise, for this condition to hold, we must ensure that there is no $t\in \cM$ with $\textrm{F}(t;\ell) \in (\alpha, \gamma]$. For the desired event $E_{\gamma, n+1} \geq 1/\alpha$ to hold, since we can already assume the first condition here, the second condition reduces to
\@ 
\textrm{F}(t;\ell) = \frac{\ell + \sum_{i=1}^n L_i\ind\{s(X_i) \leq t\} }{n+1} \notin (\alpha,\gamma], ~\forall t\in \cM, \ell \in [0,1].
\@

Finally, we obtain the claimed equivalence after condition 1, 2, and $(*)$.
\end{proof}

\subsection{Proof of Theorem~\ref{thm:opt_mdr}}
\label{app:subsec_opt_mdr}

\begin{proof}[Proof of Theorem~\ref{thm:opt_mdr}]
The given conditions imply that $\mathrm{F}^*(t)$ is continuous in $t\in [0,1]$ and non-constant in a small neighborhood around $t^*$. 
By the strong law of large numbers, since $L_i\in [0,1]$, we know that 
\$ \label{eq:conv_F_ty}
\sup_{\ell\in [0,1]}\sup_{t\in [0,1]} \big|  \mathrm{F}(t;\ell) - \mathrm{F}^*(t) \big| ~\asto~ 0, \tag{$*$}
\$
where recall that 
$ 
\mathrm{F}^*(t) =\EE[\cL(f,X,Y)\ind\{s(X)\leq t\}]
$
where $s$ and $f$ are viewed as fixed. Recall $t^* := \sup\{t\in [0,1]\colon \mathrm{F}^*(t)\leq \gamma\}$. Since $\mathrm{F}^*(t)$ is continuous and non-constant near $t^*$,~\eqref{eq:conv_F_ty} implies 
\$
\sup_{\ell \in [0,1]} |t_\gamma(\ell) - t^*| ~\asto~ 0.
\$  
Fix any $\delta_1\in (0,1)$, by the continuity of $\mF^*(t)$ around $t^*$, there exists some $\delta_2>0$ such that $\sup_{t\in [t^*-\delta_2, t^*+\delta_2]}\mF^*(t) < \alpha+\delta_1$. 
Since $\mF^*(t)$ is continuous around $t=t^*$, taking $\delta_1\to 0$, we can also take a corresponding sequence of $\delta_2\to 0$. 
We define the event 
\$
\cE := \bigg\{\sup_{\ell \in [0,1]}\sup_{t\in [0,1]} \big|  \mathrm{F}(t;\ell) - \mathrm{F}^*(t) \big| >\delta_1\bigg\}\cup\bigg\{ \sup_{\ell\in [0,1]} |t_\gamma(\ell) - t^*| > \delta_2 \bigg\}.
\$
For simplicity, we write $R_{n+1}=r(X_{n+1},Y_{n+1})$, and define the random variable $\ind_{\cE}$ being $1$ if $\cE$ occurs and $0$ otherwise. 
The a.s.~convergence above implies $\ind_{\cE}\asto 0$. 
The power is then
\$
&\EE[R_{n+1}\hat\psi_{n+1} ] = \EE\big[R_{n+1}\ind\{E_{\gamma,n+1}\geq 1/\alpha\}\big] \\ 
&= \EE\Big[R_{n+1} \ind \big\{ s(X_{n+1})\leq t_\gamma(\ell) ~ \text{and} ~ \mathrm{F}(t_\gamma(\ell),\ell) \leq \alpha,~\forall \ell\in [0,1]\big\} \Big] \\ 
&\leq \EE[R_{n+1}\ind_{ \cE }] + \EE\Big[R_{n+1} \ind_{\cE^c} \ind \big\{ s(X_{n+1})\leq t_\gamma(\ell) ~ \text{and} ~ \mathrm{F}(t_\gamma(\ell),\ell) \leq \alpha,~\forall \ell\in [0,1] \big\}  \Big]\\ 
&\leq \EE[R_{n+1}\ind_{ \cE }] + \EE\Big[R_{n+1}\ind_{\cE^c}\ind \big\{ s(X_{n+1})\leq t^*+\delta_2 ~ \text{and} ~ \sup_{t\in [t^*-\delta_2,t^*+\delta_2]}\mathrm{F}^*(t) \leq \alpha+\delta_1  \big\}  \Big]\\ 
&\leq \EE[R_{n+1}\ind_{ \cE }] + \EE\Big[ R_{n+1} \ind\big\{ s(X_{n+1})\leq t^*+\delta_2 ~ \text{and} ~ \sup_{t \in [t^*-\delta_2,t^*+\delta_2]}\mathrm{F}^*(t) \leq \alpha +\delta_1 \big\}  \Big] %\\ 
% &\leq \PP(\cE) + \PP\Big(  s(X_{n+1})\leq t^*+\delta_2  \Big),
\$
where the second equality uses the definition of $E_{\gamma, n+1}$, and 
the fourth line uses the definition of $\cE$. 
Here since $R_{n+1}$ is bounded and $\ind_{\cE}\asto 0$, we have $\EE[R_{n+1}\ind_{ \cE }]\to 0$ due to the dominated convergence theorem. 
Note that by continuity we have $\mF^*(t)=\gamma$, hence when $\gamma> \alpha$, we can take $\delta_1,\delta_2>0$ small enough such that $\sup_{t \in [t^*-\delta_2,t^*+\delta_2]}\mathrm{F}^*(t) > \alpha +\delta_1$. We thus have $\EE[R_{n+1}\hat\psi_{n+1} ]\to 0$. 
On the other hand, for $\gamma\leq \alpha$, taking $\delta_1\to 0$ and $\delta_2\to 0$ we have 
\$
\limsup_{n\to\infty}\EE[R_{n+1}\hat\psi_{n+1} ] 
\leq \EE[R_{n+1}\ind_{ \cE }] + \EE\big[ R_{n+1} \ind \{ s(X_{n+1})\leq t^*  \}  \big]
\$ 
Similarly, 
\$
&\EE[R_{n+1}\hat\psi_{n+1} ] = \EE\big[R_{n+1}\ind\{E_{\gamma,n+1}\geq 1/\alpha\}\big]  \\  
&\geq \EE\Big[R_{n+1} \ind_{\cE^c} \ind \big\{ s(X_{n+1})\leq t_\gamma(\ell) ~ \text{and} ~ \mathrm{F}(t_\gamma(\ell),\ell) \leq \alpha,~\forall \ell\in [0,1] \big\}  \Big] \\ 
&\geq  \EE\Big[R_{n+1} \ind_{\cE^c} \ind \big\{ s(X_{n+1})\leq t^*- \delta_2 ~ \text{and} ~ \sup_{t\in [t^*-\delta_2,t^*+\delta_2]}\mathrm{F}^*(t) \leq \alpha +\delta_1  \big\} \Big] \\
&\geq  \EE\Big[R_{n+1} \ind \big\{ s(X_{n+1})\leq t^*- \delta_2 ~ \text{and} ~ \sup_{t\in [t^*-\delta_2,t^*+\delta_2]}\mathrm{F}^*(t) \leq \alpha +\delta_1  \big\} \Big]  - \EE[R_{n+1}\ind_{\cE}].
\$ 
For $\gamma\leq \alpha$, taking $\delta_1,\delta_2\to 0$ we have 
\$
\liminf_{n\to\infty} \EE[R_{n+1}\hat\psi_{n+1} ] \geq  \EE\big[R_{n+1} \ind \{ s(X_{n+1})\leq t^*   \} \big].
\$
Combining the two bounds yields the asymptotic power 
\$
\lim_{n\to\infty}\EE[R_{n+1}\hat\psi_{n+1}]= \EE\big[R_{n+1} \ind \{ s(X_{n+1})\leq t^*   \} \big].
\$
Note that $t^*$ increases with $\gamma$, hence the asymptotic power is optimized at $\gamma=\alpha$ when $s(\cdot)$ is fixed.

We now fix $\gamma=\alpha$, and study the maximization of  $\EE\big[R_{n+1} \ind \{ s(X_{n+1})\leq t^*   \} \big]$ as a function of $s(\cdot)$. Due to the monotonicity of $\mF^*(t)$ in terms of $t$, we know that this is equivalent to 
\$
\maximize_{s\colon \cX\to \RR, t\in \RR} \quad  &\EE\big[ r(X,Y)\ind\{s(X)\leq t\}\big]\\ 
\text{subject to} \quad & \EE\big[ \cL(f,X,Y)\ind\{s(X)\leq t\}\big] \leq \alpha.
\$
Let $l(x):=\EE[\cL(f,X,Y)\given X=x]$, $r(x):=\EE[r(X,Y)\given X=x]$, and rewrite $\ind\{s(X)\leq t\} = b(X)$ equivalently via a binary function $b\colon \cX\to \{0,1\}$. The above program is further equivalent to 
\$
\maximize_{b\colon \cX\to \{0,1\}} \quad  &\EE\big[ r(X) b(X)\big]\\ 
\text{subject to} \quad & \EE\big[ l(X)b(X)\big] \leq \alpha.
\$
In the following, we prove the optimal solution similar to the Neyman-Pearson lemma~\citep{lehmann1986testing}. 
We define $\rho(x)=l(x)/r(x)$ and $b^*(x) = \ind\{\rho(x)\leq c_0\}$ where $c_0 = \sup\{c\colon \EE[l(X)\ind\{\rho(X)\leq c \}]\leq \alpha\}>0$. Since the distribution of $\rho(X)$ is non-atomic, we have $\EE[l(X)\cdot b^*(X)]=\alpha$.  
Let $b(\cdot)\colon \cX\to \{0,1\}$ be any binary function obeying $\EE[l(X)b(X)]\leq \alpha$. 
% Consider the quantity 
% \$
% \EE\big[  (l(X)-c_0) (b^*(X)-b(X))  \big].
% \$
When $\rho(X)-c_0<0$, it holds that $b^*(X)-b(X) = 1-b(X)\geq 0$. When $\rho(X)-c_0>0$, it holds that $b^*(X)-b(X)=-b(X)\leq 0$. Therefore, we always have $(\rho(X)-c_0) (b^*(X)-b(X))\leq 0$, which leads to $(l(X)-c_0r(X))(b^*(X)-b(X))\leq 0$ since $r(X)$ is nonnegative. As a result, 
\$
\EE\big[  (l(X)-c_0\cdot r(X)) (b^*(X)-b(X))  \big] \leq 0.
\$
Therefore, 
\$
c_0 \cdot \EE\big[r(X)b^*(X)-r(X)b(X)\big] \geq \EE[l(X)(b^*(X)-b(X))]=\alpha - \EE[l(X)b(X)]\geq 0,
\$ 
where the last inequality is due to the constraint for $b(X)$. 
This yields $\EE[r(X)b^*(X)]\geq \EE[r(X)b(X)]$ since $c_0>0$, which proves the optimality of $b^*(X)$. 
Therefore, the original problem is optimized at any $s(X)$ such that 
$
\ind\{s(X)\leq t\} = \ind\{l(x)/r(x)\leq c_0\},
$
for which a sufficient condition is that $s(x)$ is strictly increasing in $l(x)/r(x)$. 
\end{proof}

% \subsection{Proof of Theorem~\ref{thm:sel_general}}
% \label{app:subsec_thm_sel_general}

% \begin{proof}[Proof of Theorem~\ref{thm:sel_general}]
%     By the definition of the eBH procedure, it holds that $E_{n+j}\geq |\cR|/(\alpha k)$ for every $j\in \cR$ (called the self-consistency condition in \cite{wang2022false}). Therefore,
% \$
% \EE\Bigg[ \frac{\sum_{j=1}^m Y_{n+j} \cdot \ind\{j\in \cR\}}{1\vee|\cR|}   \Bigg] 
% &= \sum_{j=1}^m \sum_{k=1}^m \EE\Big[ \ind\{|\cR|=k\}Y_{n+j} \frac{\ind\{E_{n+j} \geq m/(\alpha k)\}}{k}  \Big]\\
% &\leq \sum_{j=1}^m \sum_{k=1}^m \EE\Big[ \ind\{|\cR|=k\}Y_{n+j} \frac{E_{n+j} \alpha k/m}{k}  \Big] \\ 
% &= \frac{\alpha}{m}\sum_{j=1}^m \sum_{k=1}^m \EE\big[ \ind\{|\cR|=k\}Y_{n+j}  E_{n+j}   \big] \\ 
% &= \frac{\alpha}{m}\sum_{j=1}^m   \EE\big[  Y_{n+j} E_{n+j} \big] \leq \alpha.
% \$
% The inequality uses  $\ind\{E_{n+j}\geq t\}\leq E_{n+j}/t$ deterministically 
% for any $t\geq 0$, and that $Y_{n+j}$ is non-negative. 
% \end{proof} 

\subsection{Proof of Theorem~\ref{thm:evalue_sel}}
\label{app:subsec_thm_sel_eval}
\begin{proof}[Proof of Theorem~\ref{thm:evalue_sel}]
    Similar to the proof of Theorem~\ref{thm:single}, we first have 
    \$
        \EE[L_{n+j}E_{\gamma, n+j}] &= \EE \Bigg[ L_{n+j}\cdot \inf_{\ell \in [0,1]} \bigg\{ \frac{(n+1)\cdot \ind\{s(X_{n+j}) \leq t_{\gamma,n+j}(\ell)\} }{\ell \ind\{s(X_{n+j}) \leq t_{\gamma,n+j}(\ell)\}+  \sum_{i=1}^n  L_i\ind\{s(X_{i}) \leq t_{\gamma,n+j}(\ell)\}} \bigg\}\Bigg] \\
        &\leq \EE \Bigg[\frac{L_{n+j} \cdot(n+1)\ind\{s(X_{n+j}) \leq T_{\gamma, n+j}\} }{L_{n+j} \ind\{s(X_{n+j}) \leq T_{\gamma, n+j}\}+  \sum_{i=1}^n  L_i\ind\{s(X_{i}) \leq T_{\gamma, n+j}\}} \Bigg]
    \$
    where $T_{\gamma, n+j} := t_{\gamma, n+j}(L_{n+j}) = \max \{t: \fr(t, L_{n+j}) \leq \gamma\}$, and
    \$
        \fr(t, L_{n+j}) = \frac{L_{n+j} \ind\{s(X_{n+j})\leq t\} + \sum_{i=1}^n L_i\ind\{s(X_i)\leq t\}}{1+\sum_{\ell\neq j} \ind\{s(X_{n+\ell}) \leq t\}}\cdot \frac{m}{n+1}.
    \$

    We note that by definition, $T_{\gamma, n+j}$ is invariant to permutations of $(Z_1, \dots, Z_n, Z_{n+j})$. In other words, $T_{\gamma, n+j}$ is deterministic if we condition on the unordered set $[\cZ_j] = [Z_1, \dots, Z_n, Z_{n+j}]$ and the (ordered) set of remaining data $\bar\cZ_j = \{Z_{n+\ell}\}_{\ell \neq j}$.
    Consequently, the value of $L_{n+j} \ind\{s(X_{n+j})\leq T_{\gamma,n+j}\} + \sum_{i=1}^n L_i\ind\{s(X_i)\leq T_{\gamma,n+j}\}$ is also determined. In addition, conditional on the event $[\cZ] = [z_1, \dots, z_n, z_{n+j}]$ for any fixed values of $z_1, \dots, z_n, z_{n+j}$, by exchangeability we have that
    \$
        (Z_1, \dots, Z_n, Z_{n+j}) \biggiven \big\{[\cZ]=[z_1,\dots,z_n, z_{n+j}] \big\} \quad \sim \quad \frac{1}{(n+1)!}\sum_{\sigma\in S_{n+j}} \delta_{(z_{\sigma(1)},\dots,z_{\sigma(n)},z_{\sigma(n+j)})}, 
    \$
    where $\delta_x$ is the point mass at $x$, and $S_{n+j}$ is the collection of all permutations on the set $\{1, \dots, n, n+j\}$. We write $[\boldsymbol{z}] := [z_1, \dots, z_n, z_{n+j}]$ and $\bar{\boldsymbol{z}} := \{z_{n+\ell}\}_{\ell \neq j}$ for simplicity. As such, 
    \$
    \EE \Bigg[\frac{L_{n+j} \cdot(n+1)\ind\{s(X_{n+j}) \leq T_{\gamma, n+j}\} }{L_{n+j} \ind\{s(X_{n+j}) \leq T_{\gamma, n+j}\}+  \sum_{i=1}^n  L_i\ind\{s(X_{i}) \leq T_{\gamma, n+j}\}} \biggiven [\cZ_j] = [\boldsymbol{z}],\bar\cZ_j = \bar{\boldsymbol{z}} \Bigg] \\
    = \frac{1}{n+1} \sum_{k \in \{1, \dots, n, n+j\} } \frac{\ell_k \cdot (n+1)\ind\{s(x_k) \leq T_{\gamma, n+j}\} }{\sum_{i=1}^n \ell_i \ind\{s(x_i) \leq T_{\gamma, n+j}\} + \ell_{n+j} \ind\{s(x_{n+j}) \leq T_{\gamma, n+j}\} } = 1,
    \$
    and we conclude the proof using the tower property.
\end{proof}

\subsection{Proof of Proposition~\ref{prop:mult_equiv}}
\label{app:subsec_mult_equiv}

\begin{proof}[Proof of Proposition~\ref{prop:mult_equiv}]
In this proof, we first show that the e-value computed using Algorithm~\ref{alg:sel} is identical to the e-value purposed in~\eqref{eq:def_eval_sel}. Then, we provide a detailed pseudocode that implements Algorithm~\ref{alg:sel} and runs under the claimed time complexity.

To simplify the computation, we begin by ruling out some cases where $E_{\gamma,n+j}=0$. Fix any $\gamma > 0$, we first observe that
\@\label{eq:comp_equiv}
E_{\gamma, n+j} \geq 1/\gamma \iff s(X_{n+j}) \leq t_{\gamma, n+j}(\ell) \text{ for any } \ell \in [0,1],  \tag{$*$}
\@
and the RHS is further equivalent to $E_{\gamma,n+j} > 0$ by definition.
First, both sides imply $t_{\gamma, n+j}(\ell) \neq -\infty$ and thus we can assume this condition.
Then, the LHS to RHS direction is easy by taking the contrapositive: if $s(X_{n+j}) > t_{\gamma,n+j}(\ell)$ for some $\ell\in[0,1]$,  then clearly the infimum is zero. For the other direction, if RHS is true, then letting $E_{\gamma,n+j}(\ell)$ to be the quantity being taken infimum over in~\eqref{eq:def_eval_sel}, we have 
\$
E_{\gamma, n+j}(\ell) &=  \frac{n+1}{\ell + \sum_{i=1}^n  L_i\ind\{s(X_{i}) \leq t_{\gamma,n+j}(\ell)\}} \\
&\geq \frac{n+1}{\ell + \sum_{i=1}^n  L_i\ind\{s(X_{i}) \leq t_{\gamma,n+j}(\ell)\}} \cdot \frac{1 +\sum_{\ell \neq j} \ind\{s(X_{n+\ell}) \leq t_{\gamma, n+j}(\ell) \} }{m}
\\ 
&=
1/\fr_{n+j}(t_{\gamma, n+j}(\ell)) \geq 1/\gamma,
\$
where $E_{\gamma, n+j}(\ell)$ is defined to be the value inside the infimum in~\eqref{eq:def_eval_sel}.
We also note the monotonicity of $t_{\gamma, n+j}(\cdot)$: since $\fr_{n+j}$ is non-decreasing in its second argument, we know $t_{\gamma, n+j}(\ell)$ is non-increasing in $\ell\in [0,1]$. Therefore if $s(X_{n+j}) \leq t_{\gamma, n+j}(1)$, the RHS in~\eqref{eq:comp_equiv}  would be true for any $\ell\in[0,1]$, and thus $E_{\gamma, n+j} \geq 1/\gamma$. In other words, we further have 
\$
    E_{\gamma, n+j} \geq 1/\gamma \iff s(X_{n+j}) \leq t_{\gamma, n+j}(1).
\$
Above equivalence  establishes that $E_{\gamma,n+j}=0$ if the RHS in~\eqref{eq:comp_equiv} does not hold, justifying Line 6 and 7 in Algorithm~\ref{alg:sel}. While the marginal risk control case (Proposition~\ref{prop:single_equiv}) essentially relies on a similar equivalence as above, we note that in the selective case, the equivalence itself is insufficient for computing the final outcome of eBH. Specifically, eBH requires evaluating of $\ind\{E_{\gamma, n+j} \geq m/(\alpha\tau)\}$ for different values of $\tau$, where $m/(\alpha\tau)$ may not equal $1/\gamma$ in general. %The equivalence above cannot be applied to different cutoffs $\gamma$, which is fixed beforehand.
 
We now proceed to address cases where the RHS in~\eqref{eq:comp_equiv} holds, i.e., assuming that $s(X_{n+j}) \leq t_{\gamma, n+j}(\ell)$ for any $\ell\in[0,1]$. 
In this case, we have 
\$
E_{\gamma, n+j}(\ell) &=  \frac{n+1}{\ell +  \sum_{i=1}^n  L_i\ind\{s(X_{i}) \leq t_{\gamma,n+j}(\ell)\}}.
\$
We now define the set of values of $\ell$ such that $t_{\gamma,n+j}(\ell)=t$: 
\$
\cL(t) := \{\ell \in [0,1]: t_{\gamma, n+j}(\ell) = t\}.
\$ 
Since now $s(X_{n+j}) \leq t_{\gamma, n+j}(\ell)$, for any $t$ such that $\cL(t) \neq \emptyset$, it must obey $t\in\cM^+ := \{s(X_{i}): i \in [n+m], s(X_i) \geq s(X_{n+j})\}$. 
Then we can rewrite $E_{\gamma, n+j}$ in terms of potential values of $t_{\gamma, n+j}(\ell)$:
\@\label{eq:sel_e_equiv}
E_{\gamma, n+j} &= \inf_{t \in \cM^+, \cL(t) \neq \emptyset} \inf_{\ell \in \cL(t)}  \frac{n+1}{\ell +  \sum_{i=1}^n  L_i\ind\{s(X_{i}) \leq t\}} \notag 
\\
&= \inf_{t \in \cM^+, \cL(t) \neq \emptyset} \frac{n+1}{\sup \cL(t) +  \sum_{i=1}^n  L_i\ind\{s(X_{i}) \leq t\}} \tag{$\triangle$}
\@
We first consider the simplest case where $\{t\colon \cL(t)\neq \emptyset\}$ is a singleton. By monotonicity, $t_{\gamma, n+j}(1) \leq t_{\gamma, n+j}(\ell) \leq t_{\gamma, n+j}(0)$ for all $\ell\in[0,1]$, hence $\{t\colon \cL(t)\neq \emptyset\}\subseteq[t_{\gamma,n+j}(1), t_{\gamma,n+j}(0)]$. As long as $t_{\gamma, n+j}(0)=t_{\gamma, n+j}(1)$, we would have $\{t\colon \cL(t)\neq \emptyset\} = \{t_{\gamma,n+j}(0)\}$, in which case 
\$
E_{\gamma, n+j} = \frac{n+1}{1 + \sum_{i=1}^n  L_i\ind\{s(X_{i}) \leq t_{\gamma, n+j}(1)\}}.
\$
This corresponds to the case addressed in Lines 8 and 9 of Algorithm~\ref{alg:sel}.

By the above alternative expression of the e-value, we can easily compute its value if we know how to efficiently compute the sets $\cL(t)$.
We now move to the general case by considering any $t \in [t_{\gamma, n+j}(1), t_{\gamma, n+j}(0)] \cap \cM^+$. We realize that
\$
\cL(t) &= \{\ell \in [0,1]: t_{\gamma, n+j}(\ell) = t\} \\
&= \big\{\ell \in [0,1]: \max\{\tau \in \cM: \fr_{n+j}(\tau;\ell) \leq \gamma \} = t \big\} \\
&= \big\{\ell \in [0,1]: \fr_{n+j}(t;\ell) \leq \gamma ~\text{and}~ \fr_{n+j}(t';\ell) > \gamma ~ \text{for all}~ t' > t, t' \in \cM  \big\} \\
&= \{\ell \in [0,1]: \fr_{n+j}(t;\ell) \leq \gamma \} \cap \bigcap_{t' > t, t' \in \cM} \{\ell \in [0,1]: \fr_{n+j}(t'; \ell) > \gamma\}.
\$
The sets in the above intersection must be intervals, due to the monotonicity of $\fr_{n+j}(\cdot;\cdot)$ in the second argument. Consequently, it suffices to compute the endpoints of these intervals.
For example, we can first check $\fr_{n+j}(t;0)$, the smallest value of $\fr_{n+j}(t;\ell)$ over $\ell\in[0,1]$.  If this smallest value is larger than $\gamma$, then clearly $\cL(t) = \emptyset$. Otherwise, since $\fr_{n+j}(t;\ell)$ is linear and increasing in $\ell$, we can compute the maximum offset, say $\bar{\ell}(t)$, such that $\fr_{n+j}(t;\bar{\ell}(t)) = \gamma$. Then the first set in the intersection would be $[0, \bar{\ell}(t)]$.
Similarly, to compute the second collection of sets, for any $t' > t, t' \in \cM$, we can compute the offset $\bar{\ell}(t')$ with $\fr_{n+j}(t'; \bar{\ell}(t')) = \gamma$, and the set would be $[\bar{\ell}(t'), 1]$ if $\bar{\ell}(t') \leq 1$.
% We now carry out this computation explicitly. For every set in the intersection, we find the value $\bar{\ell}(t)$ with $\fr_{n+j}(t;\bar{\ell}(t)) = \gamma$. 
Solving the equation $\fr_{n+j}(t;\bar{\ell}(t)) = \gamma$ gives the explicit formula
\$
\bar{\ell}(t) = \frac{\gamma(n+1)}{m} \Big(1 + \sum_{\ell \neq j}\ind\{s(X_{n+\ell}) \leq t\} \Big) - \sum_{i=1}^n L_i \ind\{s(X_{i}) \leq t\}
\$
since for $t \in \cM^+$ it holds that $\ind\{s(X_{n+j}) \leq t\} = 1$.
By the arguments above, we know that % $\cL(t)$ would be
\$
\cL(t) &= [0, \bar{\ell}(t)] \cap \bigcap_{t' > t, t \in \cM, \bar{\ell}(t') > 0} [\bar{\ell}(t'), 1] = \Big[\max_{t' > t,t' \in \cM, \fr_{n+j}(t'; 0) \leq \gamma} \bar{\ell}(t'), \bar{\ell}(t) \Big].
\$
% {\color{red} Why in the original derivation, it's $\max_{t<t'<t_{\gamma, n+j}(0), \fr(t'; 0) \geq \gamma}$? Why don't we need to consider $t' \geq t_{\gamma, n+j}(0)$?} \yingcomment{because $t_{\gamma,n+j}(\ell)\leq t_{\gamma, n+j}(0)$, so we don't need to consider candidate values above $t_{\gamma, n+j}(0)$?}
% The above maximization can further be simplified:  if $\fr(t'; 0) > \gamma$, the second set would be $[0,1]$, and we don't need to account for that $t'$ in the maximization.

To compute~\eqref{eq:sel_e_equiv}, it thus suffices to consider $t' \in \cM^*$, where
\$
\cM^* &= \{t \in \cM^+, \cL(t) \neq \emptyset\} \\
&= \cM^+ \cap [t_{\gamma, n+j}(1), t_{\gamma, n+j}(0)] \cap \Big\{t: \fr_{n+j}(t;0) \leq \gamma, \text{ and } \max_{t' > t,t' \in \cM, \fr_{n+j}(t'; 0) \leq \gamma} \bar{\ell}(t') \leq \bar{\ell}(t)  \Big\},
\$
and we have the simplified computation 
\$
E_{\gamma, n+j}
% &= \inf_{t \in \cM^+, \cL(t) \neq \emptyset} \frac{n+1}{\sup \cL(t) +  \sum_{i=1}^n  L_i\ind\{s(X_{i}) \leq t\}} \\
&= \inf_{t \in \cM^*} \frac{n+1}{\bar{\ell}(t) + \sum_{i=1}^n L_i \ind\{s(X_i) \leq t\} }.
\$
% The computation of $\cM^*$ is particularly convenient, since we only need to evaluate $\fr_{n+j}(t;0)$ for at most $|\cM|$ times, and computing $\bar{\ell}(t)$ for at most $|\cM|$ many values of $t$.  

In the above, we established the correctness of Algorithm~\ref{alg:sel}. While a naive implementation has cubic time complexity, we show below that an efficient implementation with time complexity at most $\mathcal{O}\big((n+m)m + (n+m)\log(n+m)\big)$ can be achieved by precomputing the prefix sums, $\mathrm{FR}_{n+j}$, and $t_{\gamma, n+j}$ values. We note that below pseudocode (Algorithm~\ref{alg:sel_pseudocode}) is 1-based.

In the pseudocode, array $A$ (Line 4) can be computed in linear time via the recurrence:
\$
A[i] = 
\begin{cases}
    A[i-1] + M[i][2], & \text{if } M[i] \text{ correspond to a calibration score, i.e. } M[i][2] \text{ is not null}, \\
    A[i-1],  & \text{otherwise}.
\end{cases}
\$
for $i > 1$. Similarly, arrays $B$ and $D$ admit linear-time computation due to the recurrence relations
\$
B[i] = 
\begin{cases}
    B[i-1] + 1, & \text{if } M[i][2] \text{ is null}, \\
    B[i-1],  & \text{otherwise}.
\end{cases},
~\text{and}~
D[i] = 
\begin{cases}
    \max(C[i+1], D[i+1]), & \text{if } FR_0[i+1] \leq \gamma, \\
    D[i+1],  & \text{otherwise}.
\end{cases}.
\$
We note that if there are ties among the scores, some values of $A[i]$ and $B[i]$ may be underestimated using above recurrence. To address this, we can either perform a backward pass to check for ties or use a sliding window to track indices with equal values.
The computational bottlenecks are therefore the sorting operation in Line 3, with complexity $\mathcal{O}((n+m)\log(n+m))$, and the $\mathcal{O}(m)$ iterations of Lines 6-21, each requiring $\mathcal{O}(n+m)$ time. Therefore, the overall time complexity of Algorithm~\ref{alg:sel_pseudocode} is at most $\mathcal{O}\big((n+m)m + (n+m)\log(n+m)\big)$.

\begin{algorithm}
    \small
    \captionsetup{font=small}
    \caption{Pseudocode for Algorithm~\ref{alg:sel}}
    \label{alg:sel_pseudocode}
\begin{algorithmic}[1]
    \REQUIRE{Labeled data $\{(X_i, Y_i)\}_{i=1}^n$, test data $\{X_{n+j}\}_{j=1}^m$, pretrained score $s$.} \\[0.5ex]
    \STATE Compute the true calibration risks $\{L_i\}_{i=1}^n$ and scores $\cM:=\{S_i\}_{i=1}^{n+m}$.
    \STATE Let $M_{\text{calib}}$ to be the array of pairs so that $M_{\text{calib}}[i][1]=S_i$ and  $M_{\text{calib}}[i][2]= L_i $ for $i = 1, \dots, n$, and analogously, let $M_{\text{test}}$ to be the array of pairs with elements $(S_{n+j}, \text{null})$ for $j = 1, \dots, m$.
    \STATE Concatenate $M_{\text{calib}}$ and $M_{\text{test}}$, and let $M$ to be the resulting array sorted according to the first entry.
    \STATE Compute the prefix sum arrays $A[i] = \sum_{k=1}^n L_k \ind\{S_k \leq M[i][1]\}$ and $B[i] = 1 + \sum_{k=1}^m \ind\{S_{n+k} \leq M[i][1]\}$, where $i = 1, \dots, n+m$.
    \STATE Initialize empty scalar arrays $FR_0, FR_1$ and $C$ of size $n+m$.
    \FOR{$j=1, \dots, m$}  
        \FOR{$i=1, \dots, n+m$}
            \STATE Compute $FR_0[i] = A[i] / (B[i] - \ind\{S_{n+j} \leq M[i][1]\}) \cdot m / (n+1)$. 
            % \hfill{\textcolor{Gray}{\texttt{// $\textrm{FR}_{n+j}(M[i][2], 0)$}}} 
            \STATE Compute $FR_1[i] = (A[i] + \ind\{S_{n+j} \leq M[i][1]\}) / (B[i] - \ind\{S_{n+j} \leq M[i][1]\}) \cdot m / (n+1)$.
            \STATE Let $C[i] = (n+1)/m \cdot \gamma \cdot (B[i] - \ind\{S_{n+j} \leq M[i][1]\}) - A[i]$.
        \ENDFOR
        \STATE Let $i_0$ be largest element in $\{1, \dots, n+m\}$ with $FR_0[i_0] \leq \gamma$, and let $t_0 = M[i_0][1]$.
        \STATE Let $i_1$ be largest element in $\{1, \dots, n+m\}$ with $FR_1[i_1] \leq \gamma$, and let $t_1 = M[i_1][1]$.
        \STATE Execute Line 6-9 of Algorithm~\ref{alg:sel}, with $t_0, t_1$ in the place of $t_{\gamma, n+j}(0), t_{\gamma, n+j}(1)$ respectively.
        \STATE Compute the array $D$ where $D[i] = \max_{M[j][1] > M[i][1], FR_0[j] \leq \gamma} C[j]$.
        \STATE Initialize empty set $\cM^*$.
        \FOR{$i=1, \dots, n+m$}
            \STATE Append $i$ to $\cM^*$ if $t_0 \geq M[i][1] \geq \max(S_{n+j}, t_1), FR_0[i] \leq \gamma$ and $C[i] \geq D[i]$.
        \ENDFOR
        \STATE Compute the e-value $E_{\gamma, n+j}$ as the minimum of $(n+1) / (A[i] + C[i])$ over $i \in \cM^*$.
    \ENDFOR  
    \ENSURE{The computed e-values $\{E_{\gamma, n+j}\}_{j=1}^m$.}
\end{algorithmic}
\end{algorithm}

% Formally, to see the computation complexity, we note that Line 1-2 needs $\cO(n+m)$ time. Line 3 can be completed within  $\cO((n+m)^2)$ time, which can be reduced if one first sorts the data by $s(X_i)$. In Lines 4-14, for each $j$, computing $t_{\gamma,n+j}(\ell)$ needs at most $\cO((n+m)^2)$ time (which can be reduced by first sorting); in Line 11-12, computing $\cM^*$ needs at most $\cO((n+m)^2)$ time;  Line 13 needs at most $\cO((n+m)^2)$ time. Therefore,  Lines 4-14 needs at most $\cO(m(n+m)^2)$ time. 
\end{proof}

\subsection{Proof of Proposition~\ref{thm:sdr_conn_CS}}
\label{app:subsec_thm_sdr_conn_CS}

\begin{proof}[Proof of Proposition~\ref{thm:sdr_conn_CS}]
    Throughout the proof, we denote $\cS^{\text{p}}$ (resp. $\cS^\text{e}$) as the function outputting the rejection set of BH (resp.~eBH) with p-values (resp.~e-values). Recall that we consider the conformal p-values 
    \@\label{eq:def_conformal_pval_j}
p_j = \frac{1+\sum_{i=1}^n\ind\{V(X_i, Y_i) \leq V(X_{n+j}, c)\}}{n+1},
\@
where $V(x, y) = \infty \ind\{y > c\} + s(x)$ is the clipped nonconformity score per~\cite{jin2023selection}.
    We also write $V_i = V(X_i, Y_i)$ for $i \in [n]$, and $V_{n+j} = V(X_{n+j}, Y_{n+j})$, $\widehat{V}_{n+j} = V(X_{n+j}, c)$ for $j \in [m]$. We let $\cS^{\cs}$ be the conformal selection set applied to the conformal p-values~\eqref{eq:def_conformal_pval_j}  at nominal level $\alpha$. In constructing SCoRE e-values, we set $\gamma = \alpha$.
    Also recall that we defined 
    \$
    e_j = \frac{\ind\{p_j \leq \alpha|\cS^\cs|/m\}}{\alpha|\cS^{\cs}|/m}.
    \$
    
    The following lemma is central to our proof, and its proof is at the end of this subsection.
    \begin{lemma} \label{lemma:cs_sel_vs_score_sel}
        For any $j = 1, \dots, m$, the following holds.
        \begin{enumerate}
            \item[(i)] For any $j \in \cS^{\cs}$, we have $\cS^{\cs} = \{\ell \in [m]: s(X_{n+\ell}) \leq t_{\gamma, n+j}(1)\}$.
            \item[(ii)] $j \in \cS^{\cs}$ if and only if $s(X_{n+j}) \leq t_{\gamma, n+j}(1)$.
        \end{enumerate}
    \end{lemma}

    With Lemma~\ref{lemma:cs_sel_vs_score_sel} in hand, for any $j$, we have 
    \$
        E_{\gamma, n+j}(1) &= \frac{(n+1) \cdot \ind\{s(X_{n+j}) \leq t_{\gamma, n+j}(1)\} }{\ind\{s(X_{n+j}) \leq t_{\gamma, n+j}(1)\} + \sum_{i=1}^n L_i \ind\{s(X_i) \leq t_{\gamma, n+j}(1)\} } \\
        &= \frac{m}{\fr_{n+j}(t_{\gamma, n+j}(1), 1)} \cdot\frac{\ind\{s(X_{n+j}) \leq t_{\gamma, n+j}(1)\}}{1 + \sum_{\ell \neq j} \ind\{s(X_{n+\ell}) \leq t_{\gamma, n+j}(1)\} } \\
        &\geq \frac{m}{\alpha} \cdot\frac{\ind\{s(X_{n+j}) \leq t_{\gamma, n+j}(1)\}}{1 + \sum_{\ell \neq j} \ind\{s(X_{n+\ell}) \leq t_{\gamma, n+j}(1)\} } \\
        &\geq \frac{m}{\alpha} \cdot \frac{\ind\{p_j \leq \alpha| \cS^{\text{cs}} | / m\}}{|\cS^{\text{cs}}|} = e_j, 
    \$
    where the first inequality is due to the definition of $\fr_{n+j}(\cdot,\cdot)$, and the second inequality follows from Lemma~\ref{lemma:cs_sel_vs_score_sel}. Finally, if $e_j = 0$, then $j \notin \cS^{\text{CS}}$. By Lemma~\ref{lemma:cs_sel_vs_score_sel}, this implies $s(X_{n+j}) > t_{\gamma, n+j}(1)$, and hence $E_{\gamma,n+j} = 0$. We thus conclude the proof of Proposition~\ref{thm:sdr_conn_CS}.
\end{proof}

    \begin{proof}[Proof of Lemma~\ref{lemma:cs_sel_vs_score_sel}]
        Fix any $j \in \cS^{\text{CS}}$ throughout. We prove the two facts separately. 
    \vspace{-1em}
        \paragraph{Proof of (i).} 
        Similar to the proof of Theorem 2.6 in~\cite{jin2023selection}, we note that 
        \$
            \cS^{\text{CS}} = \cS^{\text{p}}(p_1, \dots, p_m) = \cS^{\text{p}}(p_1^{(j)}, \dots, p_m^{(j)})
        \$
        for any $j \in \cS^{\text{CS}}$, where the ``proxy'' p-values $\{p_{\ell}^{(j)}\}^m_{\ell=1}$ are given by 
        \$
            p_\ell^{(j)} = \frac{1}{n+1} \Big[  \ind\{\widehat{V}_{n+j} \leq \widehat{V}_{n+\ell}\} + \sum_{i=1}^n \ind\{V_i \leq \widehat{V}_{n+\ell}\} \Big].
        \$
        % is the $j$-th set of proxy p-values. %Note that $p_\ell^{(j)}$ here is defined differently as in~\cite{jin2023selection}, where $V_{n+j}$ is used in place of $\widehat{V}_{n+j}$. 
        To see this equivalence, comparing the pairs of p-values $p_\ell$ and $p_\ell^{(j)}$, we observe that when $p_j \leq p_\ell$, we have $\widehat{V}_{n+j} \leq \widehat{V}_{n+\ell}$, hence $p_\ell^{(j)} = p_\ell$. On the other hand, if $p_j > p_{\ell}$, we have $\widehat{V}_{n+j} > \widehat{V}_{n+\ell}$ and thus $p_\ell^{(j)} \leq p_\ell \leq p_j$.
        In both cases, the ordering of each $p_\ell$ relative to $p_j = p_j^{(j)}$ is preserved when replacing $p_\ell$ with $p_\ell^{(j)}$. By the step‑up property of the BH procedure, this implies that the BH selection set remains unchanged when replacing $(p_1,\dots,p_m)$ with $(p_1^{(j)},\dots, p_m^{(j)})$. 
        
        We now rank the proxy conformal p-values to obtain $p_{(1)}^{(j)} \leq \dots \leq p_{(m)}^{(j)}$. In addition, since each p-value increases with the predicted risk $s(X_i)$ by our definition, with a slight abuse of notation we also denote $s(X_{n+(1)}) \leq \dots \leq s(X_{n+(m)})$, where $s(X_{n+(1)}) $ corresponds to $p_{(1)}^{(j)}$, and so on.
        
        By the property of BH procedure, we know $p_{(k^*)}^{(j)} \leq \alpha k^*/m$ where $k^* = |\cS^{\cs}|$ since $j\in \cS^{\cs}$. Define $\ell^* = |\{\ell \in [m]: s(X_{n+\ell}) \leq t_{\gamma, n+j}(1)\}|$, so that there are $\ell^*$-many predicted test scores below $t_{\gamma,n+j}(1)$. We then have
        \$
        \fr_{n+j}&(s(X_{n+(k^*)}), 1) = \frac{\ind\{s(X_{n+j})\leq s(X_{n+(k^*)})\} + \sum_{i=1}^n L_i\ind\{s(X_i)\leq s(X_{n+(k^*)})\}}{1+\sum_{\ell\neq j} \ind\{s(X_{n+\ell}) \leq s(X_{n+(k^*)})\}}\cdot \frac{m}{n+1} \\
        &= \frac{ \ind\{s(X_{n+j})\leq s(X_{n+(k^*)})\} + \sum_{i=1}^n \ind\{L_i \neq 0\}\ind\{s(X_i)\leq s(X_{n+(k^*)})\}}{1+\sum_{\ell\neq j} \ind\{s(X_{n+\ell}) \leq s(X_{n+(k^*)})\}}\cdot \frac{m}{n+1} \\
        &= \frac{\ind\{\widehat{V}_{n+j} \leq \widehat{V}_{n+(k^*)} \} + \sum_{i=1}^n \ind\{V_i \leq \widehat{V}_{n+(k^*)}\}}{1+\sum_{\ell\neq j} \ind\{s(X_{n+\ell}) \leq s(X_{n+(k^*)})\}}\cdot \frac{m}{n+1} \\
        &= p_{(k^*)}^{(j)} \cdot \frac{m}{k^*} \leq \frac{\alpha k^*}{m} \frac{m}{k^*} = \alpha,
        \$
        by the construction of the risk function and scores.
        Consequently, we know $s(X_{n+(k^*)}) \leq t_{\gamma, n+j}(1)$ by the definition of $t_{\gamma, n+j}(\cdot)$. This implies   $s(X_{n+\ell}) \leq s(X_{n+(k^*)})$, hence $s(X_{n+\ell}) \leq t_{\gamma, n+j}(1)$, for any $\ell \in \cS^{\cs}$. We thus establish  the direction $\cS^{\cs} \subseteq \{\ell \in [m]:s(X_{n+\ell}) \leq t_{\gamma, n+j}(1) \}$.

        For the converse, we see that
        \$
            p_{(\ell^*)}^{(j)} &= \frac{\ind\{\widehat{V}_{n+j} \leq \widehat{V}_{n+(\ell^*)} \} + \sum_{i=1}^n \ind\{V_i \leq \widehat{V}_{n+(\ell^*)}\}}{n+1} \\
            &= \frac{\ind\{\widehat{V}_{n+j} \leq \widehat{V}_{n+(\ell^*)} \} + \sum_{i=1}^n \ind\{V_i \leq \widehat{V}_{n+(\ell^*)}\}}{1+\sum_{\ell\neq j} \ind\{s(X_{n+\ell}) \leq s(X_{n+(\ell^*)})\}} \cdot \frac{m}{n+1} \cdot \frac{\ell^*}{m} \\
            &= \frac{\ind\{s(X_{n+j})\leq s(X_{n+(\ell^*)})\} + \sum_{i=1}^n L_i\ind\{s(X_i)\leq s(X_{n+(\ell^*)})\}}{1+\sum_{\ell\neq j} \ind\{s(X_{n+\ell}) \leq s(X_{n+(\ell^*)})\}} \cdot \frac{m}{n+1} \cdot \frac{\ell^*}{m} \\
            &= \fr_{n+j}(s(X_{n+(\ell^*)}), 1) \cdot \frac{\ell^*}{m},
        \$
        where we used the fact that $1+\sum_{\ell\neq j} \ind\{s(X_{n+\ell}) \leq s(X_{n+(\ell^*)})\} = \ell^*$. From here, an important observation is that $s(X_{n+(\ell^*)})$ is the largest test prediction no greater than $t_{\gamma, n+j}(1)$. 
        When $t_{\gamma, n+j}(1)$ corresponds to a test point, we must have   $s(X_{n+(\ell^*)}) = t_{\gamma, n+j}(1)$, and thus $ \fr_{n+j}(s(X_{n+(\ell^*)}), 1) =  \fr_{n+j}(t_{\gamma, n+j}(1), 1)$. Otherwise if $t_{\gamma,n+j}(1)$ corresponds to a calibration point, we notice that the function $\fr_{n+j}(\cdot, 1)$ is monotonically increasing in the range $[s(X_{n+(\ell^*)}), t_{\gamma, n+j}(1)]$, since the nominator is increasing and the denominator is constant across this range. As such, we must have $\fr_{n+j}(s(X_{n+(\ell^*)}), 1) \leq  \fr_{n+j}(t_{\gamma, n+j}(1), 1)$.
        Combining the two cases, we have
        \$
            p_{(\ell^*)}^{(j)} = \fr_{n+j}(s(X_{n+(\ell^*)}), 1) \cdot \frac{\ell^*}{m} \leq \fr_{n+j}(t_{\gamma, n+j}(1), 1) \cdot \frac{\ell^*}{m} \leq \gamma \ell^* /m = \alpha \ell^*/m,
        \$
        so we must have $\ell^* \in \cS^{\cs}$ by the nature of the BH procedure. Therefore, for any $\ell \in [m]$ with $s(X_{n+\ell}) \leq t_{\gamma, n+j}(1)$, we know $s(X_{n+\ell}) \leq s(X_{n+(\ell^*)})$ and $p_\ell^{(j)} \leq p_{(\ell^*)}^{(j)}$, so we must also have $\ell \in \cS^{\cs}$. We thus establish the other direction, which, together with the preceding part, proves (i). 

    \vspace{-0.5em}
        \paragraph{Proof of (ii).} 
        If $j \in \cS^{\cs}$, we know $j \in \{\ell \in [m]: s(X_{n+\ell}) \leq t_{\gamma, n+j}(1)\}$ by (i), and thus $s(X_{n+j}) \leq t_{\gamma, n+j}(1)$. For the other direction, suppose $s(X_{n+j}) \leq t_{\gamma, n+j}(1)$. Then by definition of $\ell^*$ (recall its definition in our proof of (i)), we have $s(X_{n+j}) \leq s(X_{n+(\ell^*)}) \leq t_{\gamma, n+j}(1)$, which implies $p_{(\ell^*)}^{(j)} = p_{(\ell^*)}$. By the same arguments as in the proof of (i), we know  $\fr_{n+j}(s(X_{n+(\ell^*)}), 1) \leq  \fr_{n+j}(t_{\gamma, n+j}(1), 1) \leq \alpha$, hence $p_{(\ell^*)}^{(j)} = p_{(\ell^*)} \leq \alpha \ell^*/m$, and thus $\ell^* \in \cS^{\cs}$. As a result, we must have $j \in \cS^{\cs}$ since $p_j \leq p_{(\ell^*)}$. This concludes the proof of the lemma.
    \end{proof}

\subsection{Proof of Corollary~\ref{cor:compare_cs}}
\label{app:cor_compare_cs}

\begin{proof}[Proof of Corollary~\ref{cor:compare_cs}]
    For (i), observe that by Theorem~\ref{thm:evalue_sel} we have $\EE[L_{n+j}E_{\gamma,n+j}(L_{n+j})] \leq 1$. Since $L_{n+j} \in \{0, 1\}$, it follows that $L_{n+j}E_{\gamma,n+j}(L_{n+j}) = L_{n+j}E_{\gamma, n+j}(1)$, and hence $\EE[L_{n+j}E_{\gamma,n+j}'] \leq 1$, which yields  (i) by Theorem~\ref{thm:general_sel_risk}. For (ii), denote the selection set of conformal selection and SCoRE (at the same nominal level $\alpha$ and setting $\gamma=\alpha$) by $\cS^{\text{CS}}$ and $\cS^{\text{SCoRE}}$, respectively. By  Theorem~\ref{thm:sdr_conn_CS} and property of the eBH procedure we immediately have $\cS^{\text{SCoRE}} \supseteq \cS^{\text{CS}}$. Conversely, suppose that $j \in \cS^{\text{SCoRE}}$. Then clearly, $E_{\alpha, n+j} \neq 0$, which by Theorem~\ref{thm:sdr_conn_CS} implies $e_j \neq 0$. By definition of $e_j$, this in turn gives $p_j \leq \alpha|\cS^{\text{CS}}| / m$, so $j \in \cS^{\text{CS}}$. Hence $\cS^{\text{SCoRE}} \subseteq \cS^{\text{CS}}$, which concludes the proof of (ii).
\end{proof}

\subsection{Proof of Theorem~\ref{thm:eval_boost}}
\label{app:subsec_boost}

\begin{proof}[Proof of Theorem~\ref{thm:eval_boost}]
    In this proof, we define the notation $\EE_{\bE, \bL}$ and $\PP_{\bE, \bL}$ to denote expectation or probability conditional on the base e-values $\bE := (E_{n+1}, \dots, E_{n+m})$ and the risks $\bL := (L_{n+1}, \dots, L_{n+m})$. Our proof strategy is to analyze the impact of boosting coefficients $\xi_{n+j}$ on the selection set, after fixing a certain set of e-values and risks.

    First, by the law of total expectation, we have
    \begin{align*}
        \textnormal{SDR} = \EE[\textnormal{SDR}(\bE)] = \sum_{j=1}^m \EE[\textnormal{SDR}(\bE, j)],
    \end{align*}
    where we define
    \begin{align*}
        \textnormal{SDR}(\bE, \bL) := \EE_{\bE, \bL} \Bigg[ \frac{\sum_{j=1}^m L_{n+j} \ind\{j \in \cR\}}{|\cR|} \Bigg], \quad \text{and} \quad \textnormal{SDR}(\bE, \bL, j) := \EE_{\bE, \bL} \Bigg[ \frac{L_{n+j} \ind\{j \in \cR\}}{|\cR|} \Bigg].
    \end{align*}
    We note that the randomness of $\textnormal{SDR}(\bE, \bL)$ and $\textnormal{SDR}(\bE, \bL, j)$ only lies in the boosting coefficients $\xi_j$. By properties of the eBH procedure, we can then write
    \begin{align*}
        \textnormal{SDR}(\bE, \bL, j) &= \EE_{\bE, \bL} \Bigg[ \frac{L_{n+j} \ind\{E_{n+j}/\xi_{n+j} \geq m/\alpha |\cR|\}}{|\cR|} \Bigg] \\
        &= \sum_{k=1}^m \EE_{\bE, \bL}  \Bigg[ \frac{L_{n+j} \ind\{E_{n+j}/\xi_{n+j} \geq m/\alpha k\}}{k} \ind\{|\cR| = k\} \Bigg].
    \end{align*}
    From here, we consider the case of heterogeneous and homogeneous boosting separately.
    
    \textbf{Heterogeneous Boosting}. We employ a similar leave-one-out argument as in~\cite{jin2023model, bai2024optimized}. Define $\cR_{j \to \infty}$ as the rejection index set of eBH (at level $\alpha$) applied to the set $\{E_{n+1}/\xi_{n+1}, \dots, E_{n+j-1}/\xi_{n+j-1}, \infty, E_{n+j+1}/\xi_{n+j+1}, \dots, E_{n+m}/\xi_{n+m}\}$. Then, clearly $\cR \subseteq \cR_{j \to \infty}$, and when the $j$-th sample is already rejected, i.e. $E_{n+j}/\xi_{n+j} \geq m/\alpha k$, we know $\cR = \cR_{j \to \infty}$. This is due to the step-up nature of eBH. Consequently,
    \begin{align*}
        &\EE_{\bE, \bL}  \Bigg[ \frac{L_{n+j} \ind\{E_{n+j}/\xi_{n+j} \geq m/\alpha k\}}{k} \ind\{|\cR| = k\} \Bigg]\\
        &\leq \EE_{\bE, \bL}  \Bigg[ \frac{L_{n+j} \ind\{E_{n+j}/\xi_{n+j} \geq m/ \alpha k \} }{k} \ind\{|\cR_{j\to \infty}| = k\} \Bigg] \\
        &= \frac{L_{n+j}}{k}\EE_{\bE, \bL} \Big[\ind\{|\cR_{j\to \infty}| = k\} \Big] \cdot \PP(\xi_{n+j} \leq E_{n+j} \alpha k/m) \\
        &= \EE_{\bE, \bL} \Big[\ind\{|\cR_{j\to \infty}| = k\} \Big] \cdot L_{n+j} E_{n+j} \alpha/m 
    \end{align*}
    where the first equality is due to independence between $\cR_{j \to \infty}$ and $\xi_j$, and the second equality is due to the uniformity of $\xi_j$. Then, we know
    \begin{align*}
        \textnormal{SDR}(\bE, \bL) = \sum_{j=1}^m \sum_{k=1}^m \EE_{\bE, \bL} [\ind\{|\cR_{j\to \infty}| = k\}] \cdot L_{n+j}E_{n+j} \alpha/m = \sum_{j=1}^m L_{n+j}E_{n+j} \alpha/m.
    \end{align*}
    Finally, taking expectation over $(\bE, \bL)$ yields $\textnormal{SDR} \leq \sum_{j=1}^m \frac{\alpha}{m}\EE[L_{n+j}E_{n+j}] \leq \sum_{j=1}^m \alpha/m =\alpha$.

    \textbf{Homogeneous Boosting}. 
    To prove this case, we first further decompose the SDR. We have
    \begin{align*}
       & \textnormal{SDR}(\bE, \bL, j)
        = \sum_{k=1}^m \EE_{\bE, \bL} \Bigg[ \frac{L_{n+j}\ind\{E_{n+j} /\xi \geq m/\alpha k\} }{k} (\ind\{|\cR| \leq k\} - \ind\{|\cR| \leq k-1\}) \Bigg] \\
        &= 
        \EE_{\bE, \bL} \Big[  \frac{L_{n+j}\ind\{E_{n+j} /\xi \geq \alpha\} }{m} \Big] + \sum_{k=1}^m \EE_{\bE, \bL} \Big[ \frac{L_{n+j}\ind\{E_{n+j} /\xi \geq m/\alpha k\}}{k} \ind\{|\cR| \leq k\} \Big] \\
        & \qquad \qquad \qquad \qquad \qquad-
        \sum_{k=0}^{m-1} \EE_{\bE, \bL} \Big[ \frac{L_{n+j}\ind\{E_{n+j} /\xi \geq m/\alpha (k+1)\} }{k+1} \ind\{|\cR| \leq k\} \Big] \\
        &= \frac{L_{n+j}}{m}\PP_{\bE, \bL}(\xi \leq \alpha E_{n+j}) + \sum_{k=1}^m \frac{L_{n+j}}{k} \PP_{\bE, \bL}(\xi/E_{n+j} \leq \alpha k/m, |\cR| \leq k) \\
        &\qquad \qquad \qquad \qquad \qquad- \sum_{k=0}^{m-1} \frac{L_{n+j}}{k+1} \PP_{\bE, \bL}(\xi/E_{n+j} \leq \alpha (k+1)/m, |\cR| \leq k) \\
        &= \frac{L_{n+j}}{m}\PP_{\bE, \bL}(\xi \leq \alpha E_{n+j}) + \sum_{k=1}^m \frac{L_{n+j}}{k} \PP_{\bE, \bL}( |\cR| \leq k \mid \xi/E_{n+j} \leq \alpha k/m) \PP_{\bE, \bL}(\xi/E_{n+j} \leq \alpha k/m) \\
        &\qquad \qquad- \sum_{k=0}^{m-1} \frac{L_{n+j}}{k+1} \PP_{\bE, \bL}(|\cR| \leq k \mid \xi/E_{n+j} \leq \alpha (k+1)/m) \PP_{\bE, \bL}(\xi/E_{n+j} \leq \alpha (k+1)/m).
    \end{align*}
    To proceed, we relies on the PRDS property of the boosted e-values $(\xi /E_{n+1}, \dots, \xi /E_{n+m})$, if a common boosting factor is used. 

    \begin{lemma}
        Let $a_1, \dots, a_m \in \mathbb{R} \cup\{+\infty\}$ be non-negative, fixed constants, and let $\xi \sim \textnormal{Unif}(0,1)$. Then, the random variables $(a_1\xi, \dots, a_m\xi)$ are PRDS on the index set $\{j: a_j \neq \infty\}$.
    \end{lemma}
    
    The proof of above lemma can be easily adapted from that of Lemma C.1 in~\cite{jin2023model} by additionally considering the case $a_j = \infty$. Setting $a_j = 1/E_{n+j}$ in this lemma, we know that $(\xi /E_{n+1}, \dots, \xi /E_{n+m})$ is PRDS on $\{j: E_{n+j} \neq 0\}$, conditional on $(\bE, \bL)$. As a result,
    \begin{align*}
        \PP_{\bE, \bL}( |\cR| \leq k \mid \xi/E_{n+j} \leq \alpha k/m) \leq \PP_{\bE, \bL}( |\cR| \leq k \mid \xi/E_{n+j} \leq \alpha (k+1)/m)
    \end{align*}
    for $j$ with $e_j < \infty$, since $|\cR| \leq k$ is an increasing set in the boosted e-values. Due to the independence of $\xi$ from every other variable, we obtain
    \begin{align*}
        &\textnormal{SDR}(\bE, \bL, j) \leq\frac{L_{n+j}}{m}\PP_{\bE, \bL}(\xi \leq \alpha E_{n+j})+ 
        \sum_{k=1}^{m-1} L_{n+j} \Big\{ \frac{1}{k} \PP_{\bE, \bL}(\xi/E_{n+j} \leq \alpha k/m) - \phantom{a}\\
        &\qquad \qquad \frac{1}{k+1} \PP_{\bE, \bL}(\xi/E_{n+j} \leq \alpha(k+1)/m) \Big\}\PP_{\bE, \bL}(|\cR| \leq k \mid\xi/E_{n+j} \leq \alpha k/m) \\
        &= \frac{ L_{n+j} \min\{\alpha E_{n+j}, 1\}}{m} + L_{n+j}\sum_{k=1}^{m-1} \Big\{\frac{\min\{1, \alpha kE_{n+j}/m\}}{k} - \frac{\min\{1, \alpha (k+1)E_{n+j}/m\}}{k+1} \Big\} \\
        & \qquad \qquad \quad \cdot \PP_{\bE, \bL}(|\cR| \leq k \mid\xi/E_{n+j} \leq \alpha k /m). \\
        &= \frac{ L_{n+j} \min\{\alpha E_{n+j}, 1\}}{m} + L_{n+j}\sum_{k=1}^{m-1} \Big\{\min\{1/k, \alpha E_{n+j}/m\} - \min\{1/(k+1), \alpha E_{n+j}/m\} \Big\} \\
        & \qquad \qquad \quad \cdot \PP_{\bE, \bL}(|\cR| \leq k \mid\xi/E_{n+j} \leq \alpha k /m). \tag{$*$}
    \end{align*}
    
    Now, if $\alpha E_{n+j} \leq 1,$ then both minimum term in the summation evaluated to $\alpha E_{n+j}/m$ for any $k$, and we have $\textnormal{SDR}(\bE, \bL, j) \leq L_{n+j} \min\{\alpha E_{n+j}, 1\}/m =  \alpha L_{n+j}E_{n+j}/m$. Otherwise, we let $k^* \in \mathbb{N}$ to be the unique integer with $\frac{1}{k^*+1} \leq \alpha E_{n+j}/m \leq \frac{1}{k^*}$, and we have
    \begin{align*}
        \textnormal{SDR}(\bE, \bL, j) &= \frac{ L_{n+j}}{m} + L_{n+j} \Bigg\{\frac{\alpha E_{n+j}}{m} - \frac{1}{k^*+1} + \sum_{k=k^*+1}^{m-1} (\frac{1}{k} - \frac{1}{k+1}) \Bigg\}\PP_{\bE, \bL}(|\cR| \leq k \mid\xi/E_{n+j} \leq \alpha k /m) \\
        &\leq \frac{L_{n+j}}{m} + L_{n+j} \Big( \frac{\alpha E_{n+j}}{m} - \frac{1}{m} \Big) \\
        &= \frac{\alpha L_{n+j}E_{n+j}}{m}.
    \end{align*}
    Putting above bounds together, we know
    \begin{align*}
        \textnormal{SDR}(\bE, \bL) \leq \sum_{j=1}^m L_{n+j}E_{n+j}\alpha/m,
    \end{align*}
    and we conclude the proof by taking the expectation over $(\bE, \bL)$.
\end{proof}

\subsection{Proof of Theorem~\ref{thm:asymptotic_sdr}}
\label{app:subsec_proof_asymptotic_sdr}

\begin{proof}[Proof of Theorem~\ref{thm:asymptotic_sdr}]
    Throughout, we view $s(\cdot)$ as fixed. 
    Note that by definition, we have 
    \$
    \frac{ \sum_{i=1}^n L_i\ind\{s(X_i)\leq t\}}{1+\sum_{\ell=1}^m \ind\{s(X_{n+\ell}) \leq t\}}\cdot \frac{m}{n+1} \leq \fr_{n+j}(t;\ell) \leq  \frac{1 + \sum_{i=1}^n L_i\ind\{s(X_i)\leq t\}}{ 1\vee \sum_{\ell=1}^m \ind\{s(X_{n+\ell}) \leq t\}}\cdot \frac{m}{n+1}.
    \$
    Then, as $m,n\to \infty$, the uniform law of large numbers applied to $\frac{1}{n}\sum_{i=1}^n L_i\ind\{s(X_i)\leq t\}$ and $\frac{1}{m}\sum_{\ell=1}^m \ind\{s(X_{n+\ell}) \leq t\}$ implies that 
    \$
    \sup_{1\leq j\leq m} \sup_{t\in\cM,\ell\in [0,1]} \big|\fr_{n+j}(t;\ell) - \fr(t) \big| ~\asto~ 0.
    \$
    Since the distribution of $s(X)$ has no point mass, the function $\fr(t)$ is continuous. Since $\fr(t)<\alpha$ for $t\in (t^*-\delta,t^*)$ for any sufficiently small $\delta>0$, we know that 
    \@\label{eq:conv_t_ell}
    \sup_{1\leq j\leq m} \sup_{ \ell\in [0,1]} \big| t_{\gamma,n+j} (\ell) - t_\gamma^* \big|  ~\asto ~ 0. 
    \@
    Also, due to the continuity of the distribution of $s(X)$, we have 
    \@\label{eq:fr_t_gamma_star}
    \fr(t_\gamma^*) =\frac{\EE[L\ind\{s(X)\leq t_\gamma^*\}]}{\PP(s(X)\leq t_\gamma^*)} = \gamma. 
    \@
    For simplicity, we write $s_i=s(X_i)$ for $i\in [n+m]$. We define 
    \$
    \hat{F}(\eta) = \frac{1}{m}\sum_{j=1}^m \ind\{E_{\gamma,n+j}\geq \eta\},\quad \eta>0.
    \$
    By Proposition~\ref{prop:mult_equiv}, we know that $t_{\gamma,n+j}(\ell)$ is decreasing in $\ell \in [0,1]$, and therefore
    \$
    \hat{F}(\eta) = \frac{1}{m}\sum_{j=1}^m \ind\{s_{n+j} \leq t_{\gamma,n+j}(0)\} \ind \Big\{ \frac{1+\sum_{i=1}^n L_i\ind\{s_i\leq t_{\gamma,n+j}(0)\}}{n+1} \leq 1/\eta \Big\}.
    \$
    By~\eqref{eq:conv_t_ell} and the uniform law of large numbers applied to $\frac{1}{n}\sum_{i=1}^n L_i\ind\{s_i \leq t\}$, as well as the continuity of the distribution of $s(X)$, we have 
    \$
    \sup_{1\leq j\leq m} \bigg| \frac{1+\sum_{i=1}^n L_i\ind\{s_i\leq t_{\gamma,n+j}(0)\}}{n+1} - \EE[L\ind\{s(X)\leq t_\gamma^*\}] \bigg| ~\asto ~ 0.
    \$
    As such, for any $\eta<1/\EE[L\ind\{s(X)\leq t_\gamma^*\}]$, by~\eqref{eq:conv_t_ell} we have 
    \$
    \min_{1\leq j\leq m} \ind \bigg\{ \frac{1+\sum_{i=1}^n L_i\ind\{s_i\leq t_{\gamma,n+j}(0)\}}{n+1} \leq 1/\eta \bigg\} ~\asto ~ 1,
    \$
    whereas for any $\eta> 1/\EE[L\ind\{s(X)\leq t_\gamma^*\}]$, by~\eqref{eq:conv_t_ell} we have 
    \$
    \max_{1\leq j\leq m} \ind \bigg\{ \frac{1+\sum_{i=1}^n L_i\ind\{s_i\leq t_{\gamma,n+j}(0)\}}{n+1} \leq 1/\eta \bigg\} ~\asto ~ 0.
    \$ 
    Therefore, the uniform law of large numbers and the continuity of the distribution of $s_{n+j}$'s imply  
    \@\label{eq:conv_F_eta}
    \sup_{\eta>0,~\eta\neq 1/ \EE[L\ind\{s(X)\leq t_\gamma^*\}]} \big|\hat{F}(\eta) - F^*(\eta)\big|~\asto ~ 0, 
    \@
    where  
    \$
    F^*(\eta) = \PP\big( s(X)\leq t_\gamma^*  \big) \ind\Big\{\eta \leq 1/\EE[L\ind\{s(X)\leq t_\gamma^*\}]  \Big\}.
    \$
    Note that the e-BH procedure can be rewritten as 
    \$
    \hat\psi_{n+j} = \ind\{E_{\gamma,n+j}\geq \hat\eta\}, \quad \hat\eta = \inf\bigg\{ \eta \colon \frac{m}{\eta\sum_{j=1}^m \ind\{E_{\gamma,n+j}\geq \eta\}} \leq \alpha  \bigg\}.
    \$
    Put differently, we have $\hat\eta = \inf\{\eta\colon \eta \hat{F}(\eta) \geq 1/\alpha\}$. Due to~\eqref{eq:conv_F_eta}, we know that $\eta\hat{F}(\eta) ~\asto~ 0$ uniformly over $\eta>1/\EE[L\ind\{s(X)\leq t_\gamma^*\}]$, whereas 
    \@\label{eq:conv_eta_F_eta}
    \sup_{\eta< 1/\EE[L\ind\{s(X)\leq t_\gamma^*\}]} \big| \eta\hat{F}(\eta) - \eta F^*(\eta) \big| ~\asto ~ 0.
    \@
    Writing $\eta^* = 1/\EE[L\ind\{s(X)\leq t_\gamma^*\}]$, we have 
    \$
    \eta^* F^*(\eta^*) = \frac{\PP(s(X)\leq t_\gamma^*)}{\EE[L\ind\{s(X)\leq t_\gamma^*\}]}  = \fr(t_\gamma^*)^{-1} = 1/\gamma,
    \$
    and $\eta F(\eta)=0$ for all $\eta>\eta^*$. On the other hand, for any $\eta<\eta^*$, it holds that 
    \$
    \eta F^*(\eta) = \eta \PP(s(X)\leq t_\gamma^*) < \frac{\PP(s(X)\leq t_\gamma^*)}{\EE[L\ind\{s(X)\leq t_\gamma^*\}]} = 1/\gamma,\quad \text{for all } \eta <\eta^*,
    \$
    where the inequality uses $\eta<\eta^*$ and~\eqref{eq:fr_t_gamma_star}.
    Therefore, for any $\gamma>\alpha$, we have $\PP(\hat\eta=+\infty)\to 1$ and thus $\PP(\hat\psi_{n+j}=1,~\forall j)\to 0$, i.e., the procedure is powerless. On the other hand, for any $\gamma\leq \alpha$, due to the uniform convergence~\eqref{eq:conv_eta_F_eta} and the linearity of the limiting function $\eta F^*(\eta)$, we see that 
    \$
    \hat\eta~\asto ~ \eta_\gamma^*,\quad \text{where}\quad \eta_\gamma^* = \frac{1}{\alpha \cdot \PP(s(X)\leq t_\gamma^*)} = \inf\{\eta\colon \eta F^*(\eta)\geq 1/\alpha\}<\eta^*.
    \$
    Recalling the definition of power and writing $R_{n+j}=r(X_{n+j},Y_{n+j})$ for simplicity, for any $\eta<\alpha$, 
    \@\label{eq:power_limit}
    &\frac{1}{m}\sum_{j=1}^m r(X_{n+j},Y_{n+j}) \hat\psi_{n+j} \notag \\ 
    &= \frac{1}{m}\sum_{j=1}^m R_{n+j} \ind\{E_{\gamma,n+j}\geq \hat\eta\} \notag \\ 
    &= \frac{1}{m}\sum_{j=1}^m R_{n+j} \ind\{s_{n+j} \leq t_{\gamma,n+j}(0)\} \ind \Big\{ \frac{1+\sum_{i=1}^n L_i\ind\{s_i\leq t_{\gamma,n+j}(0)\}}{n+1} \leq 1/\hat\eta \Big\} \notag \\ 
    &\asto ~ \EE\big[  r(X,Y) \ind\{s(X)\leq t_\gamma^* \} \big] \cdot \ind\big\{ \EE[L\ind\{s(X)\leq t_\gamma^*\}] \leq 1/\eta_\gamma^*\big\},
    \@
    where the last a.s.~convergence uses the uniform law of large numbers for $\{s_{n+j}\}_{j=1}^m$ and $\{s_i\}_{i=1}^n$, and the convergence of $1/\hat\eta~\asto~ 1/\eta_\gamma^* > 1/\eta^* = \EE[L\ind\{s(X)\leq t_\gamma^*\}] $.  
    We thus have 
    \$
    &\frac{1}{m}\sum_{j=1}^m r(X_{n+j},Y_{n+j}) \hat\psi_{n+j}~\asto ~ \EE\big[  r(X,Y) \ind\{s(X)\leq t_\gamma^* \} \big],\quad \text{for all } \gamma< \alpha. 
    \$
    By the dominated convergence theorem, this also implies the converges of the expectation since every $r(X_{n+j},Y_{n+j})$ is bounded. 
    We also remark that the limiting behavior at the critical point $\gamma=\alpha$ is unclear, since in the indicator function $ \ind \big\{ \frac{1+\sum_{i=1}^n L_i\ind\{s_i\leq t_{\gamma,n+j}(0)\}}{n+1} \leq 1/\hat\eta \big\}$, both sides converge to the same limit. We thus focus our discussion on $\gamma\uparrow \alpha$ in the next. 

    Now let $\gamma\uparrow \alpha$. The problem of optimizing the asymptotic power subject to SDR constraint reduces to maximizing $\EE\big[  r(X,Y) \ind\{s(X)\leq t_\gamma^* \} \big]$ where $\gamma\uparrow \alpha$, which is equivalent to 
    \$
    \maximize_{s(\cdot),t\in \RR} \quad &\EE\big[  r(X,Y) \ind\{s(X)\leq t \} \big] \\ 
    \textnormal{subject to} \quad &\frac{\EE[L\ind\{s(X)\leq t\}]}{\PP(s(X)\leq t)} \leq \alpha.
    \$
    Using the equivalent representation with the binary function $b(X)=\ind\{s(X)\leq t\}$ as the decision variable, the above optimization program is further equivalent to 
    \$
    \maximize_{b(\cdot) } \quad &\EE\big[  r(X,Y) b(X)\big] \\ 
    \textnormal{subject to} \quad & \EE[Lb(X)]  \leq \alpha \EE[b(X)].
    \$
    Now letting $r(X)=\EE[r(X,Y)\given X]$ and $l(X)=\EE[L\given X]$, it is further equivalent to 
    \$
    \maximize_{b(\cdot) } \quad &\EE\big[  r(X) b(X)\big] \\ 
    \textnormal{subject to} \quad & \EE[(l(X) -\alpha)b(X)]  \leq 0.
    \$
    It is clear that the optimal $b^*(X)$ must take the value of $1$ whenever $l(X)-\alpha \leq 0$, as it increase the objective without increasing $\EE[(l(X) -\alpha)b(X)]$. As such, the optimal solution and objective of the above program is further equivalent to those of
    \$
    \maximize_{b(\cdot) } \quad &\EE\big[  r(X) b(X)\big] \\ 
    \textnormal{subject to} \quad & \EE[(l(X) -\alpha)_+b(X)]  \leq \EE[(l(X) -\alpha)_-],
    \$
    where we denote $x_+=\max\{x,0\}$ and $x_-=\max\{-x,0\}$ for any $x\in \RR$. 

    We now define $\rho(x)= (l(x)-\alpha)_+/r(x)$, and $b^*(x):= \ind\{\rho(x) \leq c_0\}$ where $c_0= \sup\{c\colon \EE[(l(X)-\alpha)_+ \ind\{\rho(X)\leq c\} ] \leq \EE[(l(X)-\alpha)_-]\}$, and show the optimality of $b^*(X)$ with similar ideas as the Neyman-Pearson Lemma. Due to  the continuity of the distribution of $\rho(X)$, we know $\EE[(l(X)-\alpha)_+b^*(X)]=\EE[(l(X)-\alpha)_-]$. 
    Let $b(\cdot)$ be any binary function that obeys $\EE[(l(X)-\alpha)_+b(X)]\leq \EE[(l(X)-\alpha)_-]$. Since $b(x)\in \{0,1\}$, we have $b^*(x)-b(x)\geq 0$ whenever $\rho(x)\leq c_0$ and $b^*(x)-b(x)\leq 0$ whenever $\rho(x) > c_0$. This implies $ (\rho(X)-c_0)(b(X)-b^*(X)) \geq 0$ almost surely. Multiplying both sides by $r(X)$ and taking expectation, we obtain 
    \$
    \EE\Big[ \big\{(l(X)-\alpha)_+ - c_0r(X)\big\} \cdot\big\{ b(X)-b^*(X)\big\}  \Big] \geq 0.
    \$
    Re-organizing terms, we have 
    \$
    c_0\EE\big[r(X)b(X)\big] &\leq c_0\EE\big[r(X)b^*(X)\big] + \EE\big[  (l(X)-\alpha)_+ \cdot \{ b(X)-b^*(X) \}  \big] \\ 
    &= c_0\EE\big[r(X)b^*(X)\big] + \EE\big[  (l(X)-\alpha)_+   b(X)  \big] - \EE\big[  (l(X)-\alpha)_+   b^*(X)  \big] \\ 
    &\leq c_0\EE\big[r(X)b^*(X)\big],
    \$
    where the last inequality uses the fact that $\EE[(l(X)-\alpha)_+b(X)]\leq \EE[(l(X)-\alpha)_-]$ and $\EE[(l(X)-\alpha)_+b^*(X)] = \EE[(l(X)-\alpha)_-]$. Dividing both sizes by $c_0$, we then have $\EE[r(X)b(X)]\leq \EE[r(X)b^*(X)]$, confirming the optimality of $b^*(\cdot)$. 
    Recalling that $b^*(x)=1$ whenever $l(x)\leq \alpha$, it can be equivalently written as $b^*(x) = \ind\{(l(x)-\alpha)/r(x)\leq c_0\}$, where $c_0= \sup\{c\colon \EE[(l(X)-\alpha)  \ind\{\rho(X)\leq c\} ] \leq 0\}$. 

    So far, we have shown that the asymptotic power (as $\gamma\uparrow \alpha$) is optimized for any function $s(\cdot)$ such that $b^*(X)=\ind\{s(X)\leq t\}$ for the critical value of $t$ that obeys the constraint $\frac{\EE[L\ind\{s(X)\leq t\}]}{\PP(s(X)\leq t)} \leq \alpha$, where $b^*(x)=\ind\{(l(x)-\alpha)/r(x)\leq c_0\}$. Noting the equivalent constraint $ \EE[l(X)\ind\{s(X)\leq t\}]\leq \alpha \PP(s(X)\leq t)$, we see that this is true for any $s(x)$ that is monotone in $(l(x)-\alpha)/r(x)$, thereby completing the proof of the last statement.  
\end{proof}

\subsection{Proof of Theorem~\ref{thm:w_single}}

\begin{proof}[Proof of Theorem~\ref{thm:w_single}] \label{app:thm_w_single}
We use a similar proof strategy as in the proof of Theorem~\ref{thm:single}. Since $L_{n+j} \in [0,1]$, we have
\$
\EE[L_{n+1} E_{\gamma, n+1}] &= 
\EE\Bigg[L_{n+1} \cdot \inf_{\ell \in[0,1]} ~\Bigg\{ \frac{  \ind\{s(X_{n+1})\leq t(\ell)\} \cdot \sum_{i=1}^{n+1}w(X_{i}) }{\sum_{i=1}^n w(X_i)\cdot L_i \ind\{s(X_{i})\leq t(\ell)\} + w(X_{n+1})\cdot \ell\ind\{s(X_{n+1})\leq t(\ell)\} } \Bigg\} \Bigg]
\\
&\leq \EE\Bigg[ \frac{ L_{n+1}  \ind\{s(X_{n+1})\leq T_{\gamma,n+1}\} \cdot \sum_{i=1}^{n+1}w(X_{i})}{\sum_{i=1}^n w(X_i) \cdot L_i \ind\{s(X_{i})\leq T_{\gamma,n+1}\} + w(X_{n+1}) \cdot L_{n+1}\ind\{s(X_{n+1})\leq T_{\gamma,n+1}\} }  \Bigg],
\$
where $T_{\gamma, n+1} := t_{\gamma}(L_{n+1}) = \max\{t: \mathrm{F}(t, L_{n+1}) \leq \gamma\}$.  By definition, $\mathrm{F}(t, L_{n+1})$ is invariant to permutations of $(Z_1, \dots, Z_{n+1})$ for any $t$, hence so must be $T_{\gamma, n+1}$. $T_{\gamma, n+1}$ is therefore deterministic, conditional on $[\cZ]$. In addition, due to the weighted exchangeability~\citep{TibshiraniBCR19}, for any fixed values $z_1, \dots, z_{n+1}$, conditional on the event $[\cZ] = [z_1, \dots, z_{n+1}]$, the data sequence follows the distribution
\$
(Z_1,\dots,Z_{n+1})\biggiven \big\{[\cZ]=[z_1,\dots,z_{n+1}] \big\} \quad &\sim \quad \sum_{\sigma \in S_{n+1}} \frac{\prod_{i=1}^{n+1} w_i(x_{\sigma(i)})}{\sum_{\pi \in S_{n+1}} \prod_{i=1}^{n+1} w_i(x_{\pi(i)})} \delta_{(z_{\sigma(1)},\dots,z_{\sigma(n+1)})} \\
&= \quad \sum_{\sigma \in S_{n+1}} \frac{w_{n+1}(x_{\sigma(n+1)})}{\sum_{\pi \in S_{n+1}} w_{n+1}(x_{\pi(n+1)})} \delta_{(z_{\sigma(1)},\dots,z_{\sigma(n+1)})}
\$
where $w_i \equiv 1$ for $1 \leq i \leq n$, $w_{n+1} = w$ in the definition of weighted exchangeability, $\delta_x$ is the point mass at $x$, and $S_{n+1}$ is the collection of all permutations of $\{1, \dots, n+1\}$. Putting them together, for any fixed values $[z_1, \dots, z_{n+1}]$,
\$
&\EE\Bigg[ \frac{ L_{n+1}  \ind\{s(X_{n+1})\leq T_{\gamma,n+1}\} \cdot \sum_{i=1}^{n+1}w(X_{i})}{\sum_{i=1}^n w(X_i) \cdot L_i \ind\{s(X_{i})\leq T_{\gamma,n+1}\} + w(X_{n+1}) \cdot L_{n+1}\ind\{s(X_{n+1})\leq T_{\gamma,n+1}\} }\Bigggiven [\cZ] =[z_1,\dots,z_{n+1}]\Bigg] \\ 
&= \sum_{\sigma \in S_{n+1}} \frac{w_{n+1}(x_{\sigma(n+1)})}{\sum_{\pi \in S_{n+1}} w_{n+1}(x_{\pi(n+1)})} \frac{\sum_{i=1}^{n+1}w_{n+1}(x_{i}) \cdot \ell_{\sigma(n+1)} \ind\{ s(x_{\sigma(n+1)}) \leq T_{\gamma,n+1} \} }{\sum_{i=1}^{n+1} w_{n+1}(x_i) \cdot \ell_i \ind\{s(x_i)\leq T_{\gamma,n+1}\}} \\ 
&= \sum_{\sigma \in S_{n+1}} \sum_{j=1}^{n+1} \frac{\ind\{\sigma(n+1) = j\} \cdot w_{n+1}(x_{j})}{\sum_{\pi \in S_{n+1}} w_{n+1}(x_{\pi(n+1)})} \cdot \frac{\ell_{j} \sum_{i=1}^{n+1}w_{n+1}(x_{i}) \ind\{ s(x_j) \leq T_{\gamma,n+1} \} }{\sum_{i=1}^{n+1} w_{n+1}(x_i) \cdot \ell_i \ind\{s(x_i)\leq T_{\gamma,n+1}\}} \\ 
&= \sum_{j=1}^{n+1} \frac{n! \cdot w_{n+1}(x_{j})}{\sum_{\pi \in S_{n+1}} w_{n+1}(x_{\pi(n+1)})} \cdot \frac{\ell_{j} \sum_{i=1}^{n+1}w_{n+1}(x_{i}) \ind\{ s(x_j) \leq T_{\gamma,n+1} \} }{\sum_{i=1}^{n+1} w_{n+1}(x_i) \cdot \ell_i \ind\{s(x_i)\leq T_{\gamma,n+1}\}} \\ 
&= \sum_{j=1}^{n+1} \frac{n! \cdot w_{n+1}(x_{j})}{n! \sum_{i=1}^{n+1} w_{n+1}(x_i)} \cdot \frac{\ell_{j} (\sum_{i=1}^{n+1}w_{n+1}(x_{i})) \ind\{ s(x_j) \leq T_{\gamma,n+1} \} }{\sum_{i=1}^{n+1} w_{n+1}(x_i) \cdot \ell_i \ind\{s(x_i)\leq T_{\gamma,n+1}\}} \\
&= \sum_{j=1}^{n+1} \frac{n! \sum_{i=1}^{n+1}w_{n+1}(x_{i})}{n! \sum_{i=1}^{n+1} w_{n+1}(x_i)} \cdot \frac{\ell_{j} \cdot w_{n+1}(x_{j}) \ind\{ s(x_j) \leq T_{\gamma,n+1} \} }{\sum_{i=1}^{n+1} \ell_i \cdot w_{n+1}(x_i) \ind\{s(x_i)\leq T_{\gamma,n+1}\}} = 1.
\$
where $\ell_i := L(f, x_i, y_i)$. We now conclude the proof by the tower property.
\end{proof}

\subsection{Proof of Theorem~\ref{thm:w_evaluesel}}

\begin{proof}[Proof of Theorem~\ref{thm:w_evaluesel}] \label{app:thm_w_evaluesel}
Since $L_{n+j} \in [0,1]$, we first have
\$
&\quad ~ \EE[L_{n+j}E_{\gamma, n+j}] \\
&= \EE \Bigg[ L_{n+j}\cdot \inf_{\ell \in [0,1]} \bigg\{ \frac{ \ind\{s(X_{n+j}) \leq t_{\gamma,n+j}(\ell)\} \cdot  (w(X_{n+j})+\sum_{i=1}^n w(X_i)) }{w(X_{n+j}) \cdot\ell \ind\{s(X_{n+j}) \leq t_{\gamma,n+j}(\ell)\}+  \sum_{i=1}^n  w(X_i) \cdot L_i\ind\{s(X_{i}) \leq t_{\gamma,n+j}(\ell)\}} \bigg\}\Bigg] \\
&\leq \EE \Bigg[\frac{ L_{n+j} \ind\{s(X_{n+j}) \leq T_{\gamma, n+j}\} \cdot (w(X_{n+j})+\sum_{i=1}^n w(X_i))  }{w(X_{n+j}) \cdot L_{n+j} \ind\{s(X_{n+j}) \leq T_{\gamma, n+j}\}+  \sum_{i=1}^n  w(X_i) \cdot L_i\ind\{s(X_{i}) \leq T_{\gamma, n+j}\}} \Bigg]
\$
where $T_{\gamma,n+j}$ is defined as $t_{\gamma,n+j}(L_{n+j})$. By definition, we note that $T_{\gamma,n+j}$ is invariant to permutations of $(Z_1, \dots, Z_n, Z_{n+j})$, and so is the denominator inside the last expectation.

Consider the unordered set $[\cZ_j] = [Z_1, \dots, Z_n, Z_{n+j}]$ and the ordered set of remaining data $\bar\cZ_j = \{Z_{n+\ell}\}_{\ell \neq j}$. Conditional on $[\cZ_j]$ and $\bar\cZ_j$, the remaining randomness lies in which values in $[\cZ_j]$ the (ordered) random variables $(Z_1,\dots,Z_{n+j})$ take. Consider any fixed values $z_1,\dots,z_n,z_{n+1},\dots,z_{n+m}$, and consider the event $[\cZ_j] = [z_1,\dots,z_n, z_{n+j}]$ and $\bar\cZ_j = \bar\bz:=(z_{n+1},\dots,z_{n+j-1},z_{n+j+1},\dots,z_{n+m})$, and write the corresponding fixed values of the risks, denoted as $l_1,\dots,l_n,l_{n+j}$, where $l_i = \cL(f,x_i,y_i)$. The above arguments imply that conditional on $[\cZ_j] = [z_1,\dots,z_n, z_{n+j}]$ and $\bar\cZ_j = \bar\bz$, the random variable $T_{\gamma,n+j}$ equals a deterministic quantity, which we denote as $t_{[\bz],\bar\bz}$. In addition, 
\$
&~ \EE \Bigg[\frac{ L_{n+j} \ind\{s(X_{n+j}) \leq T_{\gamma, n+j}\} \cdot (w(X_{n+j})+\sum_{i=1}^n w(X_i))  }{w(X_{n+j}) \cdot L_{n+j} \ind\{s(X_{n+j}) \leq T_{\gamma, n+j}\}+  \sum_{i=1}^n  w(X_i) \cdot L_i\ind\{s(X_{i}) \leq T_{\gamma, n+j}\}}  \biggiven [\cZ_j] = [\boldsymbol{z}],\bar\cZ_j = \bar{\boldsymbol{z}} \Bigg] \\
&= \EE \Bigg[\frac{ L_{n+j} \ind\{s(X_{n+j}) \leq t_{[\bz],\bar\bz}\} \cdot (w(X_{n+j})+\sum_{i=1}^n w(X_i))  }{w(x_{n+j}) \cdot l_{n+j} \ind\{s(x_{n+j}) \leq t_{[\bz],\bar\bz}\}+  \sum_{i=1}^n  w(x_i) \cdot l_i\ind\{s(x_{i}) \leq t_{[\bz],\bar\bz}\}}  \biggiven [\cZ_j] = [\boldsymbol{z}],\bar\cZ_j = \bar{\boldsymbol{z}} \Bigg]  ,
\$
where the denominator is fixed given the conditioning information. 
Furthermore, by weighted exchangeability of the data $\cZ_j$, conditional on the event $[\cZ_j] = [z_1, \dots, z_n, z_{n+j}]$ where $z_i=(x_i,y_i)$, we have
\$
(Z_1, \dots, Z_n, Z_{n+j}) \biggiven \big\{[\cZ_j]=[z_1,\dots,z_n, z_{n+j}] \big\} \quad &\sim \quad \sum_{\sigma \in S_j} \frac{w(x_{\sigma(n+j)})}{\sum_{\pi \in S_{j}} w(x_{\pi(n+j)})} \delta_{(z_{\sigma(1)},\dots,z_{\sigma(n)},z_{\sigma(n+j)})}.
\$
We thus have, similar to the proof of Theorem~\ref{thm:w_single}, 
\$ 
&~ \EE \Bigg[\frac{ L_{n+j} \ind\{s(X_{n+j}) \leq t_{[\bz],\bar\bz}\} \cdot (w(X_{n+j})+\sum_{i=1}^n w(X_i))  }{w(x_{n+j}) \cdot l_{n+j} \ind\{s(x_{n+j}) \leq t_{[\bz],\bar\bz}\}+  \sum_{i=1}^n  w(x_i) \cdot l_i\ind\{s(x_{i}) \leq t_{[\bz],\bar\bz}\}}  \biggiven [\cZ_j] = [\boldsymbol{z}],\bar\cZ_j = \bar{\boldsymbol{z}} \Bigg] \\
&= \sum_{k \in \{1, \dots, n, n+j\} }  \PP \Big(Z_{n+j} = z_{n+k}\given [\boldsymbol{z}],\bar\cZ_j = \bar{\boldsymbol{z}} \Big)\cdot   \frac{ l_{n+k} \ind\{s(x_{n+k}) \leq t_{[\bz],\bar\bz}\} \cdot (w(x_{n+j})+\sum_{i=1}^n w(x_i))  }{w(x_{n+j}) \cdot l_{n+j} \ind\{s(x_{n+j}) \leq t_{[\bz],\bar\bz}\}+  \sum_{i=1}^n  w(x_i) \cdot l_i\ind\{s(x_{i}) \leq t_{[\bz],\bar\bz}\}} \\
&= \sum_{k \in \{1, \dots, n, n+j\} } \frac{w(x_k)}{\sum_{i=1}^n w(x_i) + w(x_{n+j})}\cdot   \frac{ l_{n+k} \ind\{s(x_{n+k}) \leq t_{[\bz],\bar\bz}\} \cdot (w(x_{n+j})+\sum_{i=1}^n w(x_i))  }{w(x_{n+j}) \cdot l_{n+j} \ind\{s(x_{n+j}) \leq t_{[\bz],\bar\bz}\}+  \sum_{i=1}^n  w(x_i) \cdot l_i\ind\{s(x_{i}) \leq t_{[\bz],\bar\bz}\}} \\
&= \sum_{k \in \{1, \dots, n, n+j\} } \frac{ w(x_k)\cdot l_{n+k} \ind\{s(x_{n+k}) \leq t_{[\bz],\bar\bz}\} }{w(x_{n+j}) \cdot l_{n+j} \ind\{s(x_{n+j}) \leq t_{[\bz],\bar\bz}\}+  \sum_{i=1}^n  w(x_i) \cdot l_i\ind\{s(x_{i}) \leq t_{[\bz],\bar\bz}\}} = 1,
\$
and the proof is complete by applying the tower property.
\end{proof}

\section{Proof of additional results}

\subsection{Proof of Proposition~\ref{prop:w_single_equiv}} \label{sec:prop_w_single_equiv}

\begin{proof}[Proof of Proposition~\ref{prop:w_single_equiv}] \label{app:prop_w_single_equiv}
We proceed by mimicking the proof of Proposition~\ref{prop:single_equiv}. In the weighted case, the following equivalence continues to hold:
\$
    E_{\gamma, n+1} \geq 1/\alpha \iff s(X_{n+1}) \leq t_{\gamma}(\ell), \text{ and } \textrm{F}(t_{\gamma}(\ell);\ell) \leq \alpha \text{ for any } \ell \in [0,1].
\$
Assuming the RHS, we have for any $\ell$
\$
    \quad& \frac{\ind\{s(X_{n+1})\leq t_\gamma(\ell)\} \cdot \sum_{i=1}^{n+1} w(X_i)}{\sum_{i=1}^n w(X_i)\cdot L_i \ind\{s(X_{i})\leq t_{\gamma}(\ell)\} + w(X_{n+1}) \cdot \ell \ind\{s(X_{n+1})\leq t_\gamma(\ell)\} } \\
    &= \ind\{s(X_{n+1}) \leq t_\gamma(\ell)\} / \mathrm{F}(t_\gamma(\ell), \ell) = 1 / \mathrm{F}(t_\gamma(\ell); \ell) \geq 1/\alpha.
\$
which implies $E_{\gamma,n+j} \geq 1/\alpha$. Conversely, if the RHS does not hold, then either $s(X_{n+1}) > t_\gamma(\ell)$ for some $\ell$, in which case $E_{\gamma,n+j} = 0$, or $\textrm{F}(t_\gamma(\ell);\ell) > \alpha$ for some $\ell$, in which case $E_{\gamma}(\ell) \leq 1 / \textrm{F}(t_\gamma(\ell);\ell) < 1/\alpha$. The equivalence is therefore established.

We now continue to examine the two conditions. For the first condition, we observe that
\$
\forall~ \ell \in [0,1], s(X_{n+1}) \leq t_{\gamma}(\ell) \iff \forall~ \ell \in [0,1],\textrm{F}(s(X_{n+1}), \ell) \leq \gamma.
\$
The direction $\Leftarrow$ is by the definition of $t_\gamma$, and the direction $\Rightarrow$ follows from the monotonicity of $\textrm{F}$ in the first argument: $\textrm{F}(s(X_{n+1}), \ell) \leq \textrm{F}(t_\gamma(\ell),\ell) \leq \gamma$. Since $\textrm{F}$ is also monotone in the second argument, the RHS condition reduces to $\textrm{F}(s(X_{n+1}), 1) \leq \gamma$, which is in turn
\$
\frac{w(X_{n+1}) +\sum_{i=1}^n w(X_i) L_i \ind {\{s(X_i)\leq s(X_{n+1})\}}}{\sum_{i=1}^{n+1} w(X_i)} \leq \gamma.
\$

For the second condition to hold, we must ensure that there is no $t \in \cM$ with $\textrm{F}(t;\ell) \in (\alpha,\gamma]$. This is automatic if $\gamma \leq \alpha$; otherwise, assuming the first condition, this reduces to
\$
\textrm{F}(t;\ell) = \frac{\ell \cdot w(X_{n+1}) +\sum_{i=1}^n w(X_i) L_i \ind  \{s(X_i)\leq t\} }{\sum_{i=1}^{n+1} w(X_i)} \notin (\alpha,\gamma\big], ~\forall t\in \cM, \ell\in [0,1].
\$
The proof is complete after combining all the demonstrated equivalences.
\end{proof}

% \tiancomment{it seems that the proof of computation shortcuts are almost the same as in the unweighted case. shall we just skip the 2 proofs then (for mdr, sdr)? for validity of the e-value, we may include since they involves weighted exchangeability.} \yingcomment{yes!}

\subsection{Proof of Proposition~\ref{prop:mult_equiv_w}}
\label{app:subsec_mult_equiv_w}

\begin{proof}[Proof of Proposition~\ref{prop:mult_equiv_w}] \label{app:prop_w_mult_equiv}
    We use the same strategy as the proof of Proposition~\ref{prop:mult_equiv}. First, we observe that the equivalence
    \@
        E_{\gamma,n+j} \geq 1/\gamma \iff s(X_{n+j}) \leq t_{\gamma,n+j}(\ell) \text{ for any } \ell \in [0,1]
    \@
    continues to hold with weights, by the same reasoning as in the unweighted proof. In addition, we see that $t_{\gamma,n+j}$ is still a non-increasing function of $\ell$. As such, we have
    \$
        E_{\gamma,n+j} \geq 1/\gamma \iff s(X_{n+j}) \leq t_{\gamma,n+j}(1),
    \$
    justifying Lines 6 and 7 of Algorithm~\ref{alg:w_sel}. Now, assume above conditions hold, i.e. $s(X_{n+j}) \leq t_{\gamma,n+j}(1)$. Then in this case,
    \$
        E_{\gamma,n+j}(\ell) = \frac{\sum_{i=1}^n w(X_i) + w(X_{n+j})}{\ell \cdot w(X_{n+1}) + \sum_{i=1}^n L_i \cdot w(X_i) \ind\{s(X_i) \leq t_{\gamma,n+j}(\ell) \}}.
    \$
    We now define the set of $\ell$'s such that $t_{\gamma,n+j}(\ell) = t$ by
    \$
        \cL(t) := \{\ell \in[0,1]: t_{\gamma,n+j}(\ell) = t \}.
    \$
    Since we have $s(X_{n+j}) \leq t_{\gamma,n+j}(\ell)$, for any $t$ that $\cL(t) \neq \emptyset$, we must have $t \in \cM^+ := \{s(X_i): i \in [n+m], s(X_i) \geq s(X_{n+j})\}$. We can then express $E_{\gamma,n+j}$ in terms of potential values of $t_{\gamma,n+j}(\ell)$:
    \$
        E_{\gamma,n+j} = \inf_{t \in \cM^+, \cL(t) \neq \emptyset} \frac{\sum_{i=1}^n w(X_i) + w(X_{n+j})}{\sup \cL(t) \cdot  w(X_{n+j}) + \sum_{i=1}^n L_i \cdot w(X_i) \ind\{s(X_i) \leq t\} }.
    \$
    By monotonicity, $t_{\gamma,n+j}(1) \leq t_{\gamma,n+j}(\ell) \leq t_{\gamma,n+j}(0)$ for any $\ell \in [0,1]$. Hence if $t_{\gamma,n+j}(0) = t_{\gamma,n+j}(1)$, we would have $\{t: \cL(t) \neq \emptyset\} = \{t_{\gamma,n+j}(0)\}$. In this case,
    \$
        E_{\gamma,n+j} = \inf_{t \in \cM^+, \cL(t) \neq \emptyset} \frac{\sum_{i=1}^n w(X_i) + w(X_{n+j})}{w(X_{n+j}) + \sum_{i=1}^n L_i \cdot w(X_i) \ind\{s(X_i) \leq t\} },
    \$
    which corresponds to Lines 8 and 9 of Algorithm~\ref{alg:w_sel}.

    Finally, for the general case, following the steps in the proof of Proposition~\ref{prop:mult_equiv}, we can show that
    \$
        \cL(t) = \{\ell \in [0,1]: \fr_{n+j}(t;\ell) \leq \gamma \} \cap \bigcap_{t' > t, t' \in \cM} \{\ell \in [0,1]: \fr_{n+j}(t'; \ell) > \gamma\}.
    \$
    Since $\fr_{n+j}$ is a monotone function of $\ell$, the sets in above expression must be intervals. By computing the endpoints of these intervals, we see that
    \$
        \cL(t) &= [0, \bar{\ell}(t)] \cap \bigcap_{t' > t, t \in \cM, \bar{\ell}(t') > 0} [\bar{\ell}(t'), 1] = \Big[\max_{t' > t,t' \in \cM, \fr_{n+j}(t'; 0) \leq \gamma} \bar{\ell}(t'), \bar{\ell}(t) \Big],
    \$
    where
    \$
         \bar{\ell}(t)=\frac{\gamma}{m} \cdot \frac{\sum_{i=1}^n w(X_i) + w(X_{n+j}) }{w(X_{n+j})} \big(1 + \sum_{\ell \neq j}\ind\{s(X_{n+j}) \leq t\} \big) - \sum_{i=1}^n \frac{w(X_i)}{w(X_{n+j})} L_i \ind\{s(X_{i}) \leq t\}.
    \$
    Therefore, the set of $t$ with $\cL(t) \neq \emptyset$ is reduced to
    \$
        \cM^* = \cM^+ \cap [t_{\gamma, n+j}(1), t_{\gamma, n+j}(0)] \cap \Big\{t: \fr_{n+j}(t;0) \leq \gamma, \text{ and } \max_{t' > t,t' \in \cM, \fr_{n+j}(t'; 0) \leq \gamma} \bar{\ell}(t') \leq \bar{\ell}(t)  \Big\},
    \$
    and we obtain the simplified computation by considering all $t \in \cM^*$:
    \$
        E_{\gamma,n+j} = \inf_{t\in \cM^*} \frac{\sum_{i=1}^n w(X_i) + w(X_{n+j})}{\bar\ell(t) + \sum_{i=1}^n L_i \cdot w(X_i) \ind\{s(X_i) \leq t\} }.
    \$
    By above, we just showed the correctness of Algorithm~\ref{alg:w_sel}. For the computation complexity part, it is straightforward to check that the pseudocode listed in Algorithm~\ref{alg:sel_pseudocode} works for the weighted case with the updated definition:
    \$
        A[i] &= \sum_{i=1}^n L_k w(X_k) \ind\{S_k \leq M[i]\}.
    \$
    Consequently, Algorithm~\ref{alg:w_sel} can execute in at most $\cO((n+m)m + (n+m) \log(n+m))$ time as well, concluding the proof of the proposition.
\end{proof}

\subsection{Proof of Theorem~\ref{thm:dr_mdr} (MDR double robustness)}
\label{subsec:dr_mdr}

\begin{proof}[Proof of Theorem~\ref{thm:dr_mdr}]
    For each test point $j$, we define 
    \$
    \textrm{F}_{n+j}(t;\ell) =  \frac{ \sum_{i=1}^n \hat{w}_i L_i \ind\{s(X_{i})\leq t\} + \hat{w}_{n+j}\cdot \ell \ind\{s(X_{n+1})\leq t\} }{\sum_{i=1}^{n } \hat{w}_i + \hat{w}_{n+j} },
    \$
    and so the e-values are correspondingly obtained by (slight simplifying the notations by dropping $\gamma$)
    \$
    E_{n+j} = \inf_{\ell \in[0,1]} ~\Bigg\{ \frac{  \ind\{s(X_{n+j})\leq t_{n+j}(\ell)\} \cdot (\hat{w}_{n+j}+\sum_{i=1}^{n } \hat{w}_i) }{\sum_{i=1}^n \hat{w}_i\cdot L_i \ind\{s(X_{i})\leq t_{n+j}(\ell)\} + \hat{w}_{n+j}\cdot \ell\ind\{s(X_{n+j})\leq t_{n+j}(\ell)\} } \Bigg\},
    \$
    where $t_{n+j}(\ell) = \sup\{t\in \cM\colon \textnormal{F}_{n+j}(t,\ell)\leq \alpha\}$.  
    Now define 
    \$
    \bar{\textnormal{F}}_{n+j}(t) := \textnormal{F}_{n+j}(t;L_{n+j}),\quad  \hat{t}_{n+j} = t_{n+j}(L_{n+j}) = \sup\{t\in \RR\colon \bar{\textnormal{F}}_{n+j}(t)\leq \alpha\}
    \$
    for the unknown risk $L_{n+j} = \cL(f,X_{n+j},Y_{n+j})$. Then by definition, $E_{n+j}\leq \bar{E}_{n+j}$ holds deterministically, where  we define 
    \$
    \bar{E}_{n+j}:= \frac{  \ind\{s(X_{n+j})\leq \hat{t}_{n+j} \} \cdot (\hat{w}_{n+j}+\sum_{i=1}^{n } \hat{w}_i) }{\sum_{i=1}^n \hat{w}_i\cdot L_i \ind\{s(X_{i})\leq \hat{t}_{n+j}\} + \hat{w}_{n+j}\cdot L_{n+j} \cdot \ind\{s(X_{n+j})\leq \hat{t}_{n+j}\} }.
    \$
    This leads to an upper bound on the MDR:
    \$
    \mdr_{n,m} &= \EE[L_{n+j}\ind\{E_{n+j}\geq 1/\alpha\}] \\
   & \leq \EE[L_{n+j}\ind\{\bar{E}_{n+j}\geq 1/\alpha\}] \\ 
   & = \EE\bigg[L_{n+j} \ind\Big\{ \frac{\ind\{s(X_{n+j})\leq \hat{t}_{n+j}\} }{\bar{\mathrm{F}}_{n+j}(\hat{t}_{n+j})} \geq 1/\alpha\Big\} \bigg],
    \$
    where the last equality follows from the definition of $\bar{E}_{n+j}$. Rearranging, we then have (for each $j$)
    \@\label{eq:mdr_bd1}
    \mdr_{n,m}\leq \EE \big[L_{n+j} \ind\{s(X_{n+j})\leq \hat{t}_{n+j}\} \ind\{\bar{F}(\hat{t}_{n+j})\leq \alpha \} \big] \leq \EE \big[L_{n+j} \ind\{s(X_{n+j})\leq \hat{t}_{n+j}\}  \big],
    \@
    as $\bar{F}(\hat{t}_{n+j})\leq \alpha $ always holds since $\hat{t}_{n+j}$ is searched over a finite set $\cM$. 
    Here we denote the random variable $L = \cL(f,X,Y)$. The expectation in~\eqref{eq:mdr_bd1} is over all the randomness (including the training process), so $\mdr_{n,m}$ can be viewed as an unknown, deterministic scalar (as $s(\cdot)$ is viewed as fixed). 
    
    Let $t^* = \bar{F}^{-1}(\alpha)$ be as in Theorem~\ref{thm:dr_mdr}, where we define 
    \$
    \bar{F}(t) := \frac{\EE_P[\bar{w}(X)l(X) \ind\{s(X)\leq t\}]}{\EE_P[\bar{w}(X)]} = \frac{G(t)}{\EE_P[\bar{w}(X)]},
    \$
    and the expectation is with respect to a new copy $X\sim P$, viewing $s(\cdot)$ as fixed. Thus $t^*\in \RR$ depends on the score $s(\cdot)$ only. 
    The proof of Claim~\ref{claim:t_conv_mdr} is right after this proof. 
    \begin{claim}\label{claim:t_conv_mdr}
    Under the conditions above, the random variable $\hat\delta:=\sup_{j\in [m]}|\hat{t}_{n+j}-t^*| = o_P(1)$. 
    \end{claim}
    With Claim~\ref{claim:t_conv_mdr}, continuing with~\eqref{eq:mdr_bd1} we know 
    \$
    \mdr_{n,m} \leq \EE \big[L_{n+j} \ind\{s(X_{n+j})\leq t^*\}  \big] + \EE \big[L_{n+j} \ind\{ t^*< s(X_{n+j})\leq \hat{t}_{n+j}\}  \big],
    \$
    where since $L_{n+j}\in [0,1]$, we know that for any $\epsilon>0$, 
    \$
    \EE \big[L_{n+j} \ind\{ t^*< s(X_{n+j})\leq \hat{t}_{n+j}\}  \big]
    &\leq \PP(t^*< s(X_{n+j})\leq \hat{t}_{n+j}) \\ 
    &\leq \PP( |\hat{t}_{n+j} -t^*|>\epsilon) + \PP( t^*<s(X_{n+j})\leq t^*+\epsilon)  \\ 
    &= o(1) + \PP_Q( t^*<s(X )\leq t^*+\epsilon).
    \$
    Since $s(X)$ has no point mass,  
    taking the sup-limit on both sides, and by the arbitrariness of $\epsilon>0$, we know 
    \@\label{eq:mdr_bd2}
    \limsup_{n,m\to \infty} \mdr_{n,m} \leq \EE_Q\big[L  \ind\{s(X )\leq t^*\}  \big] .
    \@
    In the next, we prove the upper bound for the right-handed side of~\eqref{eq:mdr_bd2} under either of the two conditions:
    \begin{itemize}
        \item First, if $\bar{w}(\cdot)=w(\cdot)$, by the covariate shift assumption it is straightforward to see that
    \$
    \bar{F}(t) = \EE_Q[l(X)\ind\{s(X)\leq t\}] = \EE_Q[L \ind\{s(X )\leq t\}],
    \$
    so the RHS of~\eqref{eq:mdr_bd2} is equal to $\bar{F}(t^*) = \alpha$ since $t^* = \bar{F}^{-1}(\alpha)$. 
    \item  Second, suppose $\bar{l}(\cdot) = l(\cdot)$. Then by the triangle inequality, 
    \@\label{eq:conv_cdf_hatl}
    &\sup_{t\in \RR} \bigg| \frac{1}{m}\sum_{j=1}^m \hat{l}(X_{n+j}) \ind\{s(X_{n+j})\leq  {t}\} - \EE_Q[\bar{l}(X)\ind\{s(X)\leq t\}] \bigg| \notag \\ 
    &\leq \sup_{t\in \RR} \bigg| \frac{1}{m}\sum_{j=1}^m \hat{l}(X_{n+j}) \ind\{s(X_{n+j})\leq  {t}\} - \EE_Q[\hat{l}(X)\ind\{s(X)\leq t\}] \bigg| \notag \\ 
    &\qquad + \sup_{t\in \RR} \big| \EE_Q[\hat{l}(X)\ind\{s(X)\leq t\}]- \EE_Q[\bar{l}(X)\ind\{s(X)\leq t\}]\big| \notag \\ 
    &\leq O_P(1/\sqrt{m}) + \EE_Q[|\hat{l}(X)-\bar{l}(X)|] = o_P(1).
    \@
    In the expectations above, both $\hat{l}(\cdot)$ and $s(\cdot)$ are viewed as fixed functions, and the expectation is over a new independent draw $X\sim Q$. 
    In addition, the $O_P(1/\sqrt{m})$ term is obtained by the following arguments. By Lemma~\ref{lem:weighted_cdf_conv}, we know that 
    \$
    \EE\Bigg[  \sup_{t\in \RR} \bigg| \frac{1}{m}\sum_{j=1}^m \hat{l}(X_{n+j}) \ind\{s(X_{n+j})\leq  {t}\} - \EE_Q[\hat{l}(X)\ind\{s(X)\leq t\}] \bigg|    \Bigggiven \hat{l}(\cdot),s(\cdot)\Bigg] \leq \frac{CM}{\sqrt{n}},
    \$
    where $M=\sup_{x}\hat{l}(x)$. Then applying the tower property and  Markov's inequality we obtain the $O_P(1/\sqrt{m})$ bound.

    Since $s(X)$ has no point mass and the map $t\mapsto \EE_Q[\bar{l}(X)\ind\{s(X)\leq t\}]$ is strictly increasing at $t^\dagger := \sup\{t\colon \EE_Q[\bar{l}(X)\ind\{s(X)\leq t\}\leq \alpha\}$, we have $\hat{t} =t^\dagger + o_P(1)$ for the cutoff $\hat{t}$ in Assumption~\ref{assump:balance_mdr}. 
    % Denote $H(t):=\EE_P[\bar{w}(X) l(X) \ind\{s(X)\leq t\}]$. 
    
    Since $\frac{1}{n}(\hat{w}_i - \bar{w}(X_i))^2 = o_P(1)$ and $\|\hat{l}(\cdot)-l(\cdot)\|_{L_2}=o_P(1)$, we have 
    \$
    &\sup_{t\in \RR} \bigg| \frac{1}{n}\sum_{i=1}^n \hat{w}_i \hat{l}(X_i)\ind\{s(X_i)\leq  {t}\} - \EE_P[\bar{w}(X) l(X) \ind\{s(X)\leq t\}] \bigg| \notag \\  
    &\leq \underbrace{\sup_{t\in \RR} \bigg| \frac{1}{n}\sum_{i=1}^n \hat{w}_i \hat{l}(X_i)\ind\{s(X_i)\leq  {t}\} - \frac{1}{n}\sum_{i=1}^n \bar{w}(X_i) l(X_i)\ind\{s(X_i)\leq  {t}\} \bigg|}_{(a)}\notag \\ 
    &\qquad + \underbrace{\sup_{t\in \RR} \bigg| \frac{1}{n}\sum_{i=1}^n \bar{w}(X_i) l(X_i)\ind\{s(X_i)\leq  {t}\} - \EE_P[\bar{w}(X) l(X) \ind\{s(X)\leq t\}] \bigg|}_{(b)}.  
    \$
    First, invoking Lemma~\ref{lem:weighted_cdf_conv} for $f = \bar{w}(\cdot)l(\cdot)$ and Markov's inequality and tower property we know $(b) = o_P(1/\sqrt{m})$. On the other hand, 
    \@\label{eq:wl_bd1} 
    (a) & = \sup_{t\in \RR} \bigg| \frac{1}{n}\sum_{i=1}^n \big( \hat{w}_i \hat{l}(X_i)  -  \bar{w}(X_i) l(X_i) \big) \ind\{s(X_i)\leq  {t}\} \bigg|  \notag\\ 
     &\leq \sup_{t\in \RR} \bigg| \frac{1}{n}\sum_{i=1}^n \hat{w}_i \big( \hat{l}(X_i) - l(X_i) \big) \ind\{s(X_i)\leq  {t}\} \bigg| + \sup_{t\in \RR} \bigg| \frac{1}{n}\sum_{i=1}^n \bar{l}(X_i)\big( \hat{w}_i   - \bar{w}(X_i) \big) \ind\{s(X_i)\leq  {t}\} \bigg|   \notag\\ 
     &\leq  \frac{1}{n}\sum_{i=1}^n  \hat{w}_i \big| \hat{l}(X_i) - l(X_i)  \big|  +   \frac{1}{n}\sum_{i=1}^n \bar{l}(X_i) \big| \hat{w}_i   - \bar{w}(X_i) \big|  .
    \@
    where we repeatedly apply the triangle inequality. By the Cauchy-Schwarz inequality, 
    \$
    \frac{1}{n}\sum_{i=1}^n  \hat{w}_i \big| \hat{l}(X_i) - l(X_i)  \big| \leq \frac{1}{n}\sqrt{\sum_{i=1}^n \hat{w}_i^2} \cdot \sqrt{\sum_{i=1}^n \big| \hat{l}(X_i) - l(X_i)  \big|^2} = \frac{O(\sqrt{n})O_P(\sqrt{n} \|\hat{l}(\cdot)-l(\cdot)\|_{L_2})}{n}   = o_P(1),
    \$
    and due to the boundedness of $\bar{l}(X)=l(X)\in [0,1]$, by the Cauchy-Schwarz inequality, 
    \$
     \frac{1}{n}\sum_{i=1}^n \bar{l}(X_i) \big| \hat{w}_i   - \bar{w}(X_i) \big|  \leq  \frac{1}{n}\sum_{i=1}^n  \big| \hat{w}_i   - \bar{w}(X_i) \big| 
     \leq \sqrt{\frac{1}{n}\sum_{i=1}^n (\hat{w}_i-\bar{w}(X_i))^2} = o_P(1).
    \$
    Putting the above two inequalities together with~\eqref{eq:wl_bd1}, we obtain 
    \@\label{eq:cov_weighted_cdf}
(a) =\sup_{t\in \RR} \bigg| \frac{1}{n}\sum_{i=1}^n \hat{w}_i \hat{l}(X_i)\ind\{s(X_i)\leq  {t}\} - \EE_P[\bar{w}(X) l(X) \ind\{s(X)\leq t\}] \bigg| = o_P(1),
    \@
    and therefore 
    \@\label{eq:emp_w_cdf_bd}
    \sup_{t\in \RR} \bigg| \frac{1}{n}\sum_{i=1}^n \hat{w}_i \hat{l}(X_i)\ind\{s(X_i)\leq  {t}\} - \EE_P[\bar{w}(X) l(X) \ind\{s(X)\leq t\}] \bigg| = o_P(1).
    \@
    Since $\hat{t} = t^\dagger +o_P(1)$, applying~\eqref{eq:conv_cdf_hatl} and~\eqref{eq:emp_w_cdf_bd} to $\hat{t}$ we know 
    \$ 
     \frac{1}{n}\sum_{i=1}^n \hat{w}_i \hat{l}(X_i)\ind\{s(X_i)\leq \hat{t}\}& =  \EE_P[\bar{w}(X) l(X) \ind\{s(X)\leq \hat{t}\}]+ o_P(1) \\ 
     &= \EE_P[\bar{w}(X) l(X) \ind\{s(X)\leq t^\dagger\}] +o_P(1), \\ 
    \frac{1}{m}\sum_{j=1}^m \hat{l}(X_{n+j}) \ind\{s(X_{n+j})\leq \hat{t}\} 
    & = \EE_Q[l(X)\ind\{s(X)\leq \hat{t}\}] + o_P(1) \\ & = \EE_Q[l(X)\ind\{s(X)\leq t^\dagger\}] + o_P(1).
    \$
    Putting this together with Assumption~\ref{assump:balance_mdr}, and by the continuity of $t\mapsto \EE_P[\bar{w}(X)l(X)\ind\{s(X)\leq t\}]$ and $t\mapsto \EE_Q[l(X)\ind\{s(X)\leq t\}]$, we have 
    \$
    \EE_P[\bar{w}(X) l(X) \ind\{s(X)\leq t^\dagger\}] = \EE_Q[l(X)\ind\{s(X)\leq t^\dagger\}] + o_P(1).
    \$
    Similarly, the second balancing condition yields $\EE[\bar{w}(X)]=1$. 
    This further implies $t^\dagger = t^*$ due to the continuity and monotonicity of $G(t)$ at $t=t^*$.  This implies 
    \$
    \EE\big[L_{n+j} \ind\{s(X_{n+j})\leq t^*\}  \big]
    = \EE_Q[L\ind\{s(X)\leq t^\dagger\}] 
    = \EE_Q[l(X)\ind\{s(X)\leq t^\dagger\}]  \leq \alpha 
    \$
    by the definition of $t^\dagger$. This, together with~\eqref{eq:mdr_bd2}, completes the proof for the second case. 
    \end{itemize}
    We therefore complete the proof of Theorem~\ref{thm:dr_mdr}. To see how this implies Theorem~\ref{thm:w_mdr_asymp}, the convergence of $\bar{w}_{n,m}$ to the true weight $w$ implies that the weight is correctly specified. Consequently, the required continuity and monotonicity of the two mappings agrees and reduces to the given condition. Take $\hat{l}(X_i) = 1$ as a constant, Assumption~\ref{assump:balance_mdr} is automatically satisfied. The theorem therefore applies, establishing asymptotic MDR control in the setting of Theorem~\ref{thm:w_mdr_asymp}.
\end{proof}

    \begin{proof}[Proof of Claim~\ref{claim:t_conv_mdr}]
        It holds deterministically that for any $j\in [m]$, 
    \@\label{eq:bd_hat_F}
     \frac{ \sum_{i=1}^n \hat{w}_i L_i \ind\{s(X_{i})\leq t\}   }{\sum_{i=1}^{n } \hat{w}_i +M }\leq \bar{\textnormal{F}}_{n+j}(t ) \leq  \frac{ \sum_{i=1}^n \hat{w}_i L_i \ind\{s(X_{i})\leq t\} + M }{\sum_{i=1}^{n } \hat{w}_i  }.
    \@
    Under the convergence conditions in Theorem~\ref{thm:dr_mdr}, by Cauchy-Schwarz inequality, we know 
    \$
    \sup_{t\in \RR} \bigg| \frac{1}{n}\sum_{i=1}^n \hat{w}_i L_i \ind\{s(X_{i})\leq t\} - \frac{1}{n}\sum_{i=1}^n \bar{w}(X_i) L_i \ind\{s(X_{i})\leq t\} \bigg| \leq \sqrt{\frac{1}{n}\sum_{i=1}^n (\hat{w}_i - \bar{w}(X_i))^2 } =o_P(1). 
    \$
    and similarly 
    $
    \frac{1}{n} \sum_{i=1}^{n } \hat{w}_i = \EE_P[\bar{w}(X)]+o_P(1).
    $
    In addition, invoking Lemma~\ref{lem:weighted_cdf_conv} we know \$
\sup_{t\in \RR} \bigg| \frac{1}{n} \sum_{i=1}^n \bar{w}(X_i) L_i \ind\{s(X_{i})\leq t\} - \EE_P[\bar{w}(X)l(X) \ind\{s(X)\leq t\}] \bigg| = o_P(1).
    \$
    Thus, taking $n\to \infty$ in~\eqref{eq:bd_hat_F} we know 
    % Taking $n\to\infty$ in~\eqref{eq:bd_hat_F} yields
    % \$
    % \sup_{\ell\in [0,1],j\in [m]}|\mathrm{F}_{n+j}(t;\ell) - \bar{F}(t) | \stackrel{P}{\to} 0 
    % \$
    % for any fixed $t\in \RR$. 
    % Since $\mathrm{F}(t;\ell)$ is increasing in $t$, through a standard finite-grid argument similar to that in the uniform law of large numbers we have 
    \$
    \sup_{ t\in \RR, j\in [m]}|\bar{\textnormal{F}}_{n+j}(t ) - \bar{F}(t) | \stackrel{P}{\to} 0 .
    \$
    Since $\bar{F}(t)$ is strictly increasing around $t^* = \bar{F}^{-1}(\alpha)$, we know 
    $
    \sup_{j\in [m]} | \hat{t}_{n+j} - t^*| \stackrel{P}{\to} 0. 
    $
    % \yingcomment{This inversion is why we need a fixed $s(\cdot)$, otherwise we'd need $s$ to converge to some $\bar{s}$ and replace things by $\bar{s}$.}
    \end{proof}

\subsection{Proof of Theorem~\ref{thm:dr_sdr} (SDR double robustness)}
\label{app:subsec_proof_dr_sdr}
\begin{proof}[Proof of Theorem~\ref{thm:dr_sdr}]
Take $\gamma=\alpha$. The e-values used are defined as (simplifying the notations)
\begin{align} \label{eq:def_eval_sel_hatw}
E_{  n+j} := \inf_{\ell \in [0,1]} \bigg\{ \frac{\ind\{s(X_{n+j}) \leq t_{ n+j}(\ell)\} \cdot ( \hat{w}_{n+j}+\sum_{i=1}^n \hat{w}_i) }{ \hat{w}_{n+j} \cdot \ell\ind\{s(X_{n+j}) \leq t_{ n+j}(\ell)\}+  \sum_{i=1}^n \hat{w}_i \cdot L_i\ind\{s(X_{i}) \leq t_{ n+j}(\ell)\}} \bigg\}, \qquad
\end{align}
where $t_{ n+j}(\ell) = \max\big\{t\colon  {\textrm{FR}_{n+j}}(t;\ell) \leq \alpha\big\}$, and
\$
{\textrm{FR}_{n+j}}(t;\ell) = \frac{ \hat{w}_{n+j} \cdot \ell\ind\{s(X_{n+j})\leq t\} + \sum_{i=1}^n \hat{w}_i \cdot L_i\ind\{s(X_i)\leq t\}}{1+\sum_{\ell\neq j} \ind\{s(X_{n+\ell}) \leq t\}}\cdot \frac{m}{\hat{w}_{n+j}+\sum_{i=1}^n \hat{w}_i}.
\$
Plugging in $\ell = L_{n+j}=\cL(f,X_{n+j},Y_{n+j})$, we know 
\$
E_{n+j}\leq \bar{E}_{n+j} := \frac{\ind\{s(X_{n+j}) \leq \hat{t}_{n+j}\} \cdot ( \hat{w}_{n+j}+\sum_{i=1}^n \hat{w}_i) }{ \hat{w}_{n+j} \cdot L_{n+j}\ind\{s(X_{n+j}) \leq \hat{t}_{n+j}\}+  \sum_{i=1}^n \hat{w}_i \cdot L_i\ind\{s(X_{i}) \leq \hat{t}_{n+j}\}},
\$
where 
\$
&\hat{t}_{n+j}:= t_{n+j}(L_{n+j}) = \max\{t\in \cM\colon \bar{\textrm{F}}_{n+j}(t) \leq \alpha\},\quad \\ 
&\bar{\textrm{F}}_{n+j}(t) :=  \frac{ \hat{w}_{n+j} \cdot L_{n+j} \ind\{s(X_{n+j})\leq t\} + \sum_{i=1}^n \hat{w}_i \cdot L_i\ind\{s(X_i)\leq t\}}{1+\sum_{k\neq j} \ind\{s(X_{n+k}) \leq t\}}\cdot \frac{m}{\hat{w}_{n+j}+\sum_{i=1}^n \hat{w}_i}.
\$
By construction, 
\@\label{eq:barE_sel_eq}
\bar{E}_{n+j} &= \frac{\ind\{s(X_{n+j})\leq \hat{t}_{n+j}\}}{\bar{F}_{n+j}(\hat{t}_{n+j})} \cdot \frac{m}{1+\sum_{\ell\neq j} \ind\{s(X_{n+\ell}) \leq \hat{t}_{n+j} \}}\notag \\ 
&\leq  \frac{\ind\{s(X_{n+j})\leq \hat{t}_{n+j}\}}{1+\sum_{\ell\neq j} \ind\{s(X_{n+\ell}) \leq \hat{t}_{n+j} \}} \cdot \frac{m}{\alpha}.
\@
Here the second inequality holds because $\bar{F}_{n+j}(\hat{t}_{n+j})\leq \alpha$ since $\hat{t}_{n+j}$ searches over the finite set $\cM$.

By construction, and since $\sup_i|\hat{w}_i|\leq M$, it holds deterministically and uniformly over all $j\in [m]$ that 
\@\label{eq:Fnj_dtm_bd}
 \frac{ \sum_{i=1}^n \hat{w}_i \cdot L_i\ind\{s(X_i)\leq t\}}{  \frac{1}{m}\sum_{k=1}^m \ind\{s(X_{n+k}) \leq t\} +  \frac{1}{m}}\cdot \frac{1}{M+\sum_{i=1}^n \hat{w}_i}.
 \leq 
\bar{\textrm{F}}_{n+j}(t) \leq   \frac{ M + \sum_{i=1}^n \hat{w}_i \cdot L_i\ind\{s(X_i)\leq t\}}{\frac{1}{m}\sum_{k=1}^m \ind\{s(X_{n+k}) \leq t\} }\cdot \frac{1}{ \sum_{i=1}^n \hat{w}_i}.
\@
Now define 
\$
& G(t):= \EE_P[\bar{w}(X)L\ind\{s(X)\leq t\}] , \quad H_Q(t)= \PP_Q(s(X) \leq t) = \PP(s(X_{n+j})\leq t), \\ 
& \bar{F}(t) = \frac{G(t)}{H_Q(t) \EE_P[\bar{w}(X)]},\quad 
t^* = \sup\{t\colon \bar{F}(t)\leq \alpha\}.
\$
The given convergence conditions imply
\$
\sup_{t\in \RR} \bigg| \frac{1}{n}\sum_{i=1}^n \hat{w}_i \cdot L_i\ind\{s(X_i)\leq t\} - \frac{1}{n}\sum_{i=1}^n \bar{w}(X_i) \cdot L_i\ind\{s(X_i)\leq t\} \bigg|
\leq \sqrt{\frac{1}{n}\sum_{i=1}^n (\hat{w}_i - \bar{w}(X_i))^2} = o_P(1) 
\$
and similarly  $ |\frac{1}{n}\sum_{i=1}^n \hat{w}_i - \EE_P[\bar{w}(X)]| = o_P(1)$. 
In addition, $\sup_{t\in \RR}|\frac{1}{m}\sum_{k=1}^m \ind\{s(X_{n+k}) \leq t\} - H_Q(t) | = o_P(1)$  due to the uniform law of large numbers or Lemma~\ref{lem:weighted_cdf_conv}.  Therefore, combining these results with~\eqref{eq:Fnj_dtm_bd}, and since  $H_Q(t^*)>0$, there exists a constant $\delta>0$ such that for any $\epsilon \in (0, \delta)$, 
\@\label{eq:conv_Fj}
\sup_{t \geq t^*-\epsilon ,j\in [m]} \big| \bar{\textrm{F}}_{n+j}(t) - \bar{F}(t) \big| = o_P(1). 
\@
Recall that $\bar{F}(t)$ is continuous at $t^* = \sup\{t\colon \bar{F}(t)\leq \alpha\}$, and for any sufficiently small $\epsilon>0$, there exists some $t_\epsilon\in (t^*-\epsilon,t^*)$ such that $\bar{F}(t_\epsilon)<\alpha$. Thus, by~\eqref{eq:conv_Fj} we know 
\$
\PP\Big(\inf_{j\in [m]} \hat{t}_{n+j} \geq t_\epsilon \Big) \geq \PP\Big( \sup_{j\in [m]} \bar{\mathrm{F}}_{n+j}(t_\epsilon) \leq (\alpha + \bar{F}(t_\epsilon))/2 \Big) \to 1
\$
as $n,m\to \infty$. 
On the other hand, by the definition of $t^*$ and the right-continuity of $\bar{\mathrm{F}}$, for any $\epsilon>0$ there exists some $\delta>0$ so that $\bar{\mathrm{F}}(t)>\alpha+\delta$ for all $t'>t+\epsilon$. Thus by~\eqref{eq:conv_Fj} we know 
\$
\PP\Big(\sup_{j\in [m]} \hat{t}_{n+j} \leq t+\epsilon \Big) \geq \PP\Big( \inf_{j\in [m]} \inf_{t'\geq t+\epsilon} \bar{\mathrm{F}}_{n+j}(t+\epsilon) \geq \alpha + \delta/2 \Big) \to 1
\$
as $n,m\to \infty$. Putting the two directions together, and by the arbitrariness of $\epsilon>0$, we know 
\@\label{eq:cov_hatt}
\sup_{j\in [m]} |\hat{t}_{n+j} - t^*| = o_P(1). 
\@
For any $\epsilon>0$, we define the event 
\$
\cE_{\epsilon} = \bigg\{ \sup_{j\in [m]} |\hat{t}_{n+j} - t^*|  > \epsilon  \bigg\} &\cup \bigg\{ \sup_{t\in \RR} \Big|\frac{1}{m}\sum_{j=1}^m L_{n+j}\ind\{s(X_{n+j})\leq t\} - \EE_Q[L\ind\{s(X)\leq t\}]\Big|  > \epsilon  \bigg\} \\ 
& \cup \bigg\{ \sup_{t\in \RR} \Big|\frac{1}{m}\sum_{j=1}^m \ind\{s(X_{n+j})\leq t\} - \PP_Q(s(X)\leq t)\Big|  > \epsilon  \bigg\} .
\$
which satisfies $\PP(\cE_{\epsilon})\to 0$ for any fixed $\epsilon>0$ as $n,m\to \infty$ by~\eqref{eq:cov_hatt} and the uniform law of large numbers or Lemma~\ref{lem:weighted_cdf_conv}.

By the definition of the eBH procedure (Theorem~\ref{thm:general_sel_risk}), we know that $\cR = \{j\in [m]\colon E_{n+j}\geq m/(\alpha \hat\tau)\}$ for $\hat\tau = |\cR|$. Thus the SDR can be bounded as 
\$
\sdr_{n,m} &= \EE\bigg[ \frac{\sum_{j=1}^m L_{n+j}\ind\{j\in \cR\}}{1\vee \hat\tau} \ind_{\cE_\epsilon}\bigg] + \EE\bigg[ \frac{\sum_{j=1}^m L_{n+j}\ind\{j\in \cR\}}{1\vee \hat\tau} \ind_{\cE_\epsilon^c}\bigg] \\
&\leq  \EE\bigg[ \frac{\sum_{j=1}^m L_{n+j}\ind\{E_{n+j}\geq m/(\alpha \hat\tau)\}}{1\vee \hat\tau}\ind_{\cE_\epsilon^c}\bigg] + \PP(\cE_\epsilon )\\ 
&\leq    \EE\bigg[ \frac{\sum_{j=1}^m L_{n+j}\ind\{\bar{E}_{n+j}\geq m/(\alpha \hat\tau)\}}{1\vee \hat\tau}\ind_{\cE_\epsilon^c}\bigg] + \PP(\cE_\epsilon )\\ 
&\leq  \sum_{j=1}^m \EE\Bigg[ L_{n+j}\frac{\ind\{\frac{\ind\{s(X_{n+j})\leq \hat{t}_{n+j}\}}{1+\sum_{k\neq j} \ind\{s(X_{n+k}) \leq \hat{t}_{n+j} \}}  \geq \frac{1}{  \hat\tau} \}}{1\vee \hat\tau}\ind_{\cE_\epsilon^c}\bigg] + \PP(\cE_\epsilon )\\ 
&\leq \sum_{j=1}^m  \EE\Bigg[ \frac{L_{n+j}\ind\{s(X_{n+j})\leq \hat{t}_{n+j}\} \ind\{ {1+\sum_{\ell\neq j} \ind \big\{s(X_{n+\ell}) \leq \hat{t}_{n+j} \} \leq  \hat\tau} \big\}}{1\vee \hat\tau}\ind_{\cE_\epsilon^c}\bigg] + \PP(\cE_\epsilon ) \\ 
&\leq \sum_{j=1}^m \EE\Bigg[ \frac{ L_{n+j}\ind\{s(X_{n+j})\leq \hat{t}_{n+j}\}  }{1+\sum_{k\neq j} \ind \big\{s(X_{n+k}) \leq \hat{t}_{n+j} \}}\ind_{\cE_\epsilon^c}\bigg] + \PP(\cE_\epsilon )\\
&=\sum_{j=1}^m  \EE\Bigg[ \frac{ L_{n+j}\ind\{s(X_{n+j})\leq \hat{t}_{n+j}\}  }{ 1\vee \sum_{k=1}^m \ind \big\{s(X_{n+k}) \leq \hat{t}_{n+j} \}}\ind_{\cE_\epsilon^c}\bigg] + \PP(\cE_\epsilon ),
\$
where the second inequality uses $E_{n+j}\leq \bar{E}_{n+j}$ and the fact that $L_{n+j}\leq 1$ hence the ratio in the expectation is upper bounded by $1$, the third inequality uses~\eqref{eq:barE_sel_eq}, and the last two inequalities follow from certain re-arrangements. Here on the event $\cE_\epsilon^c$, it holds simultaneously for all $j\in [m]$ that 
\$
\frac{ L_{n+j}\ind\{s(X_{n+j})\leq t^*-\epsilon \}  }{ 1\vee \sum_{k=1}^m \ind \big\{s(X_{n+k}) \leq t^*+\epsilon \}} \leq \frac{ L_{n+j}\ind\{s(X_{n+j})\leq \hat{t}_{n+j}\}  }{ 1\vee \sum_{k=1}^m \ind \big\{s(X_{n+k}) \leq \hat{t}_{n+j} \}} 
\leq \frac{ L_{n+j}\ind\{s(X_{n+j})\leq t^*+\epsilon \}  }{ 1\vee \sum_{k=1}^m \ind \big\{s(X_{n+k}) \leq t^*-\epsilon \}},
\$
which implies 
\$
\sdr_{n,m} \leq   \EE\Bigg[ \frac{ \sum_{j=1}^m L_{n+j}\ind\{s(X_{n+j})\leq t^* +\epsilon\}  }{ 1\vee \sum_{k=1}^m \ind \big\{s(X_{n+k}) \leq t^*-\epsilon \}}\ind_{\cE_\epsilon^c}\Bigg] + \PP(\cE_\epsilon ).
\$
Since $\PP_Q(s(X)\leq t^*)>0$, we know   $\epsilon < \PP_Q(s(X)\leq t^*-\epsilon)$ holds for sufficiently small $\epsilon>0$. Thus, taking $\epsilon>0$ sufficiently small, we have 
\$
\sdr_{n,m} \leq   \EE\bigg[ \frac{ \EE_Q[L\ind\{s(X)\leq t^*+\epsilon \}] + \epsilon }{  \PP_Q(s(X)\leq t^* - \epsilon ) - \epsilon}\ind_{\cE_\epsilon^c}\bigg] + \PP(\cE_\epsilon ) 
\leq \frac{ \EE_Q[L\ind\{s(X)\leq t^*+\epsilon \}] + \epsilon }{  \PP_Q(s(X)\leq t^* - \epsilon ) - \epsilon} + \PP(\cE_\epsilon ) . 
\$
By the arbitrariness of $\epsilon>0$ and the continuity of $s(X)$, we know 
\$
\limsup_{n,m\to\infty} \sdr_{n,m}  &\leq \frac{\EE_Q[L\ind\{s(X)\leq t^*  \}]  }{  \PP_Q(s(X)\leq t^*   )  } \\ &= \bar{F}(t^*) \frac{\EE_Q[L\ind\{s(X)\leq t^*  \}]\cdot \EE_P[\bar{w}(X)]}{\EE_P[L\bar{w}(X)\ind\{s(X)\leq t^*  \}]} 
\leq \alpha\cdot \frac{\EE_Q[L\ind\{s(X)\leq t^*  \}]\cdot \EE_P[\bar{w}(X)]}{\EE_P[L\bar{w}(X)\ind\{s(X)\leq t^*  \}]} .
\$
We now proceed to show that the above quantity is upper bounded by $\alpha$, under either of the two conditions.
\begin{itemize}
    \item First, if $\bar{w}(\cdot)=w(\cdot)$, then by definition we know $\EE_P[\bar{w}(X)]=1$, and $\EE_Q[L\ind\{s(X)\leq t^*\}] = \EE_P[L \bar{w}(X)\ind\{s(X)\leq t^*\}]$. This implies 
    \$
    \frac{\EE_Q[L\ind\{s(X)\leq t^*  \}]\cdot \EE_P[\bar{w}(X)]}{\EE_P[L\bar{w}(X)\ind\{s(X)\leq t^*  \}]} = 1
    \$
    and thus the desired result. 
    \item Second, suppose $\bar{l}(\cdot)=l(\cdot)$. Recall the balancing cutoff 
    \$
    \hat{t} = \sup\bigg\{t\colon \frac{\frac{1}{n}\sum_{i=1}^n \hat{w}_i \hat{l}(X_i)\ind\{s(X_i)\leq  {t}\}}{1\vee \sum_{j=1}^m  \ind\{s(X_{n+j})\leq  {t}\}} \leq \alpha \bigg\}.
    \$
    Following the same arguments as in the proof of Theorem~\ref{thm:dr_mdr} for~\eqref{eq:cov_weighted_cdf}, the given conditions imply 
        \@\label{eq:cov_w_hatl_sdr}
\sup_{t\in \RR} \bigg| \frac{1}{n}\sum_{i=1}^n \hat{w}_i \hat{l}(X_i)\ind\{s(X_i)\leq  {t}\} - \EE_P[\bar{w}(X) l(X) \ind\{s(X)\leq t\}] \bigg| = o_P(1) , 
    \@
    and $\sup_{t\in \RR}|\frac{1}{m}\sum_{j=1}^m  \ind\{s(X_{n+j})\leq  {t}\} - \PP_Q(s(X)\leq t)| =o_P(1)$. Also, taking $m,n\to\infty$ in the balancing conditions of Assumption~\ref{assump:balance_sdr} yields $\EE_P[\bar{w}(X)]=1$. Thus we know 
    \$
    \sup_{t\in \RR} \bigg| \frac{\frac{1}{n}\sum_{i=1}^n \hat{w}_i \hat{l}(X_i)\ind\{s(X_i)\leq  {t}\}}{1\vee \sum_{j=1}^m  \ind\{s(X_{n+j})\leq  {t}\}} - \bar{F}(t) \bigg| = o_P(1). 
    \$
    Since $\bar{F}(t)$ is continuous at $t^* = \sup\{t\colon \bar{F}(t)\leq \alpha\}$, and for any sufficiently small $\epsilon>0$, there exists some $t\in (t^*-\epsilon,t^*)$ such that $\bar{F}(t)<\alpha$, with similar arguments as those in the proof of~\eqref{eq:cov_hatt} we know 
    \@\label{eq:cov_hatt_dagger}
    \hat{t} =t^*+o_P(1).
    \@
    
    With the similar arguments as in the proof of~\eqref{eq:conv_cdf_hatl} in  Theorem~\ref{thm:dr_mdr} we can show 
    \@\label{eq:cov_hatl_sdr}
    &\sup_{t\in \RR} \bigg| \frac{1}{m}\sum_{j=1}^m \hat{l}(X_{n+j}) \ind\{s(X_{n+j})\leq  {t}\} - \EE_Q[\bar{l}(X)\ind\{s(X)\leq t\}] \bigg| = o_P(1).
    \@
    which, together with~\eqref{eq:cov_w_hatl_sdr} and the balancing conditions in Assumption~\ref{assump:balance_sdr}, leads to 
    \$
     \EE_P[\bar{w}(X) l(X) \ind\{s(X)\leq \hat{t}\}]  =   \EE_Q[  l(X) \ind\{s(X)\leq \hat{t}\}] + o_P(1),\quad \EE[\bar{w}(X)]=1,
    \$
    where   $\hat{t}$ shall be viewed as fixed and $X$ as an independent copy. By~\eqref{eq:cov_hatt_dagger} and the continuity of $s(X)$ we then have 
    \$
    \EE_P[\bar{w}(X) l(X) \ind\{s(X)\leq t^*\}]  =   \EE_Q[  l(X) \ind\{s(X)\leq t^*\}] +o_P(1),
    \$
    which implies 
    \$ 
    \frac{\EE_Q[L\ind\{s(X)\leq t^*  \}]\cdot \EE_P[\bar{w}(X)]}{\EE_P[L\bar{w}(X)\ind\{s(X)\leq t^*  \}]} = 1 +o_P(1),
    \$
    and thus the desired result.
\end{itemize}
We therefore conclude the proof of Theorem~\ref{thm:dr_sdr}. To see how this implies Theorem~\ref{thm:w_sdr_asymp}, the convergence of $\bar{w}_{n,m}$ to the true weight $w$ implies that the weight is correctly specified. As $\bar{w} = w$, the given condition on $F(t)$ exactly translates to the condition $\bar{F}(t)$ in the current theorem. In addition, Assumption~\ref{assump:balance_sdr} is automatically satisfied taking $\hat\ell(X_i) = 1$, since the weight estimates are consistent. The theorem therefore applies, establishing asymptotic SDR control in the setting of Theorem~\ref{thm:w_sdr_asymp}.
\end{proof}

\section{Additional details and results for numerical experiments} 

\subsection{Additional results for Section~\ref{subsec:app_drug}} \label{app:drug_more}

In this part, we present the analysis results on three additional drug discovery tasks under distribution shift. The results for datasets \texttt{clearance\_hepatocyte}, \texttt{clearance\_microsome} and \texttt{ppbr\_az} are shown in Figures~\ref{fig:drug_az}, \ref{fig:drug_microsome} and \ref{fig:drug_ppbr}, where the reward function is   diversity. 
Figures~\ref{fig:drug_caco_wang_activity} to~\ref{fig:drug_ppbr_activity} show the corresponding results for the four datasets with the activity reward function.

\begin{figure}[htbp]
    \centering
    \includegraphics[width=0.85\linewidth]{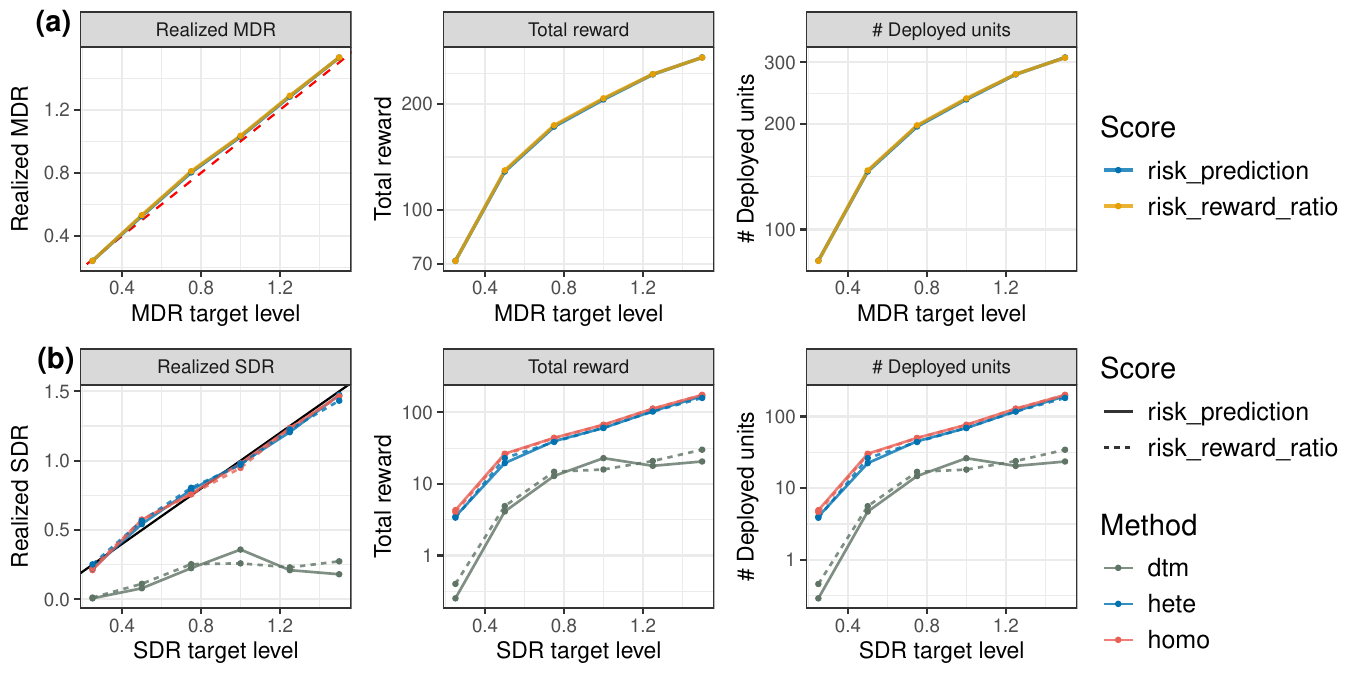}
    \caption{{\small MDR (a) and SDR (b) control for drug discovery with the \texttt{clearance\_hepatocyte} dataset in Therapeutic Data Commons with estimated covariate shift and diversity reward. Details are otherwise the same as Figure~\ref{fig:real-drug}.}}
    \label{fig:drug_az}
\end{figure}

\begin{figure}[htbp] 
    \centering
    \includegraphics[width=0.85\linewidth]{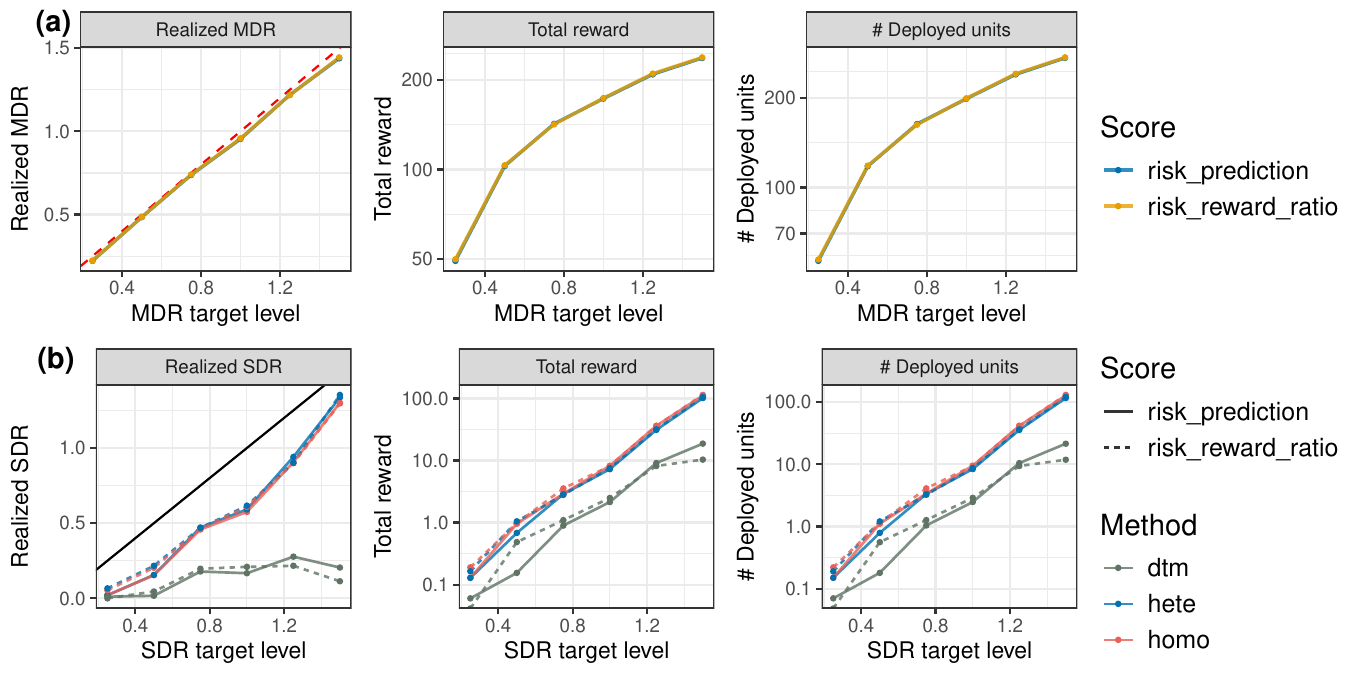}
    \caption{{\small MDR (a) and SDR (b) control for drug discovery with the \texttt{clearance\_microsome} dataset in Therapeutic Data Commons with estimated covariate shift and diversity reward. Details are otherwise the same as Figure~\ref{fig:real-drug}.}}
    \label{fig:drug_microsome}
\end{figure}

\begin{figure}[htbp] 
    \centering
    \includegraphics[width=0.85\linewidth]{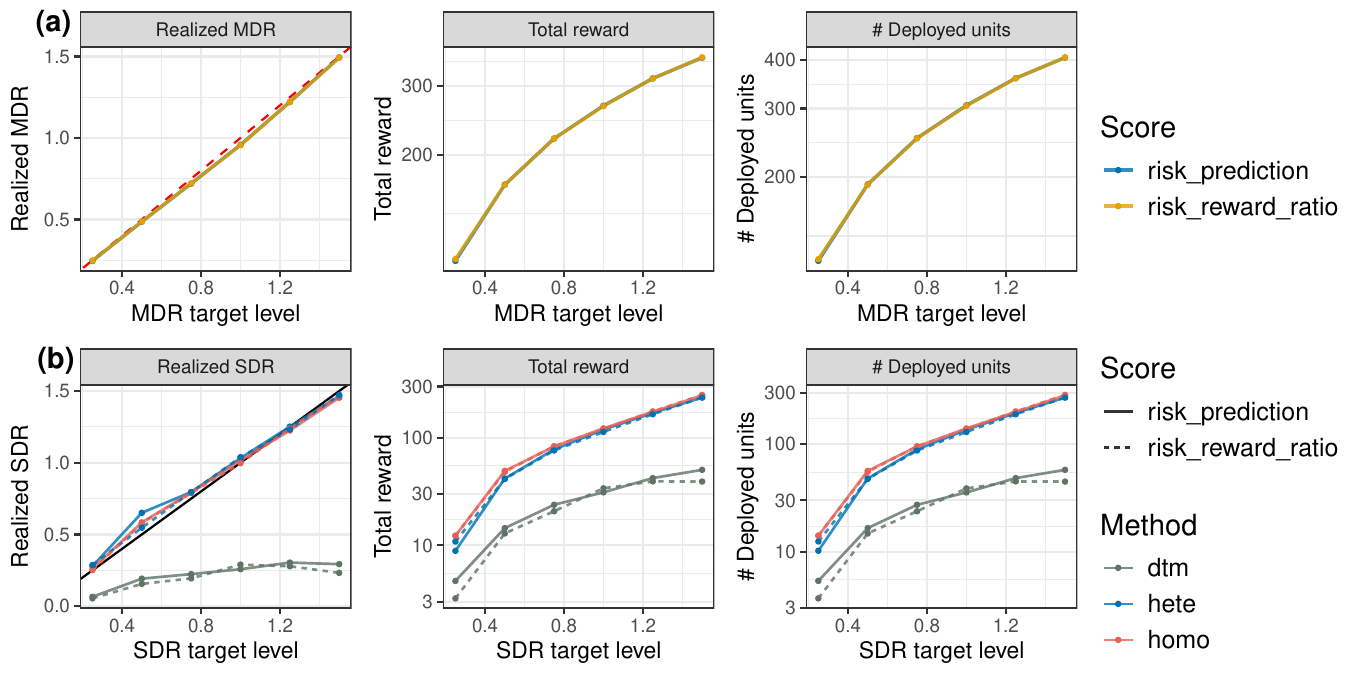}
    \caption{{\small MDR (a) and SDR (b) control for drug discovery with the \texttt{ppbr\_az} dataset in Therapeutic Data Commons with estimated covariate shift and diversity reward. Details are otherwise the same as Figure~\ref{fig:real-drug}.}}
    \label{fig:drug_ppbr}
\end{figure}

\begin{figure}[htbp]
    \centering
    \includegraphics[width=0.85\linewidth]{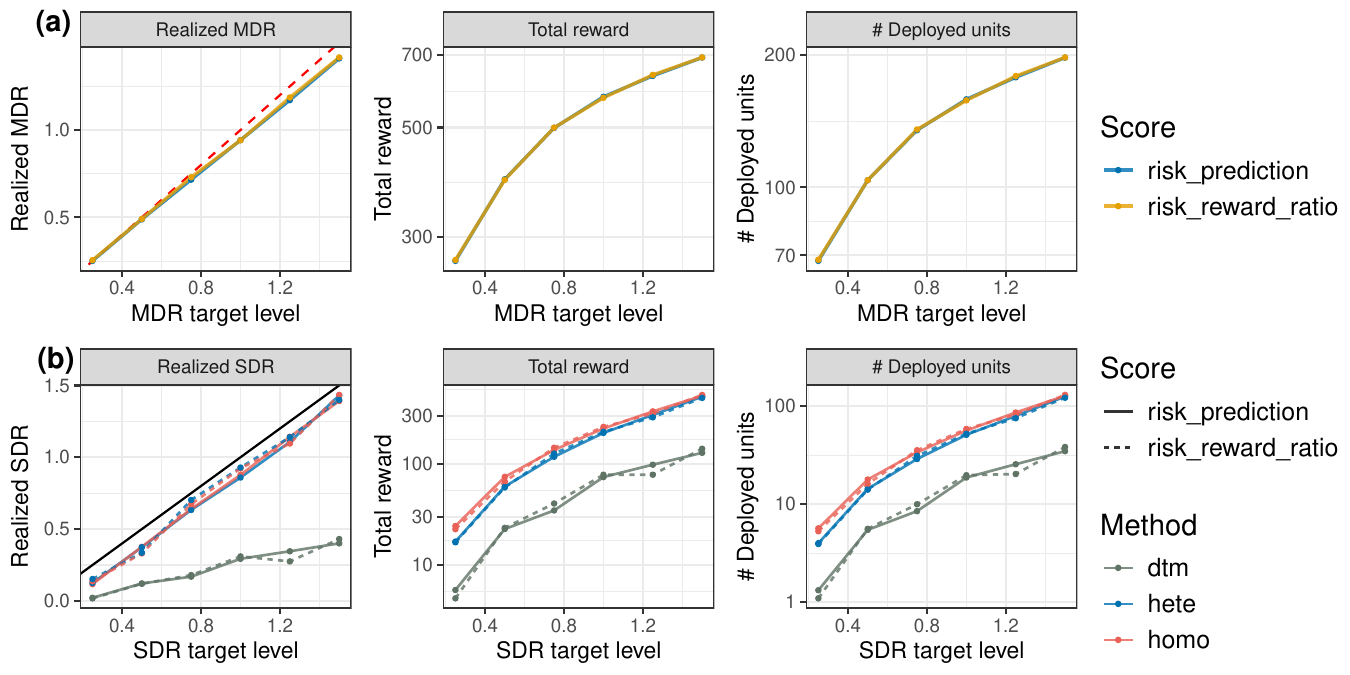}
    \caption{{\small MDR (a) and SDR (b) control for drug discovery with the \texttt{caco\_wang} dataset in Therapeutic Data Commons with estimated covariate shift and activity reward. Details are otherwise the same as Figure~\ref{fig:real-drug}.}}
    \label{fig:drug_caco_wang_activity}
\end{figure}

\begin{figure}[htbp]
    \centering
    \includegraphics[width=0.85\linewidth]{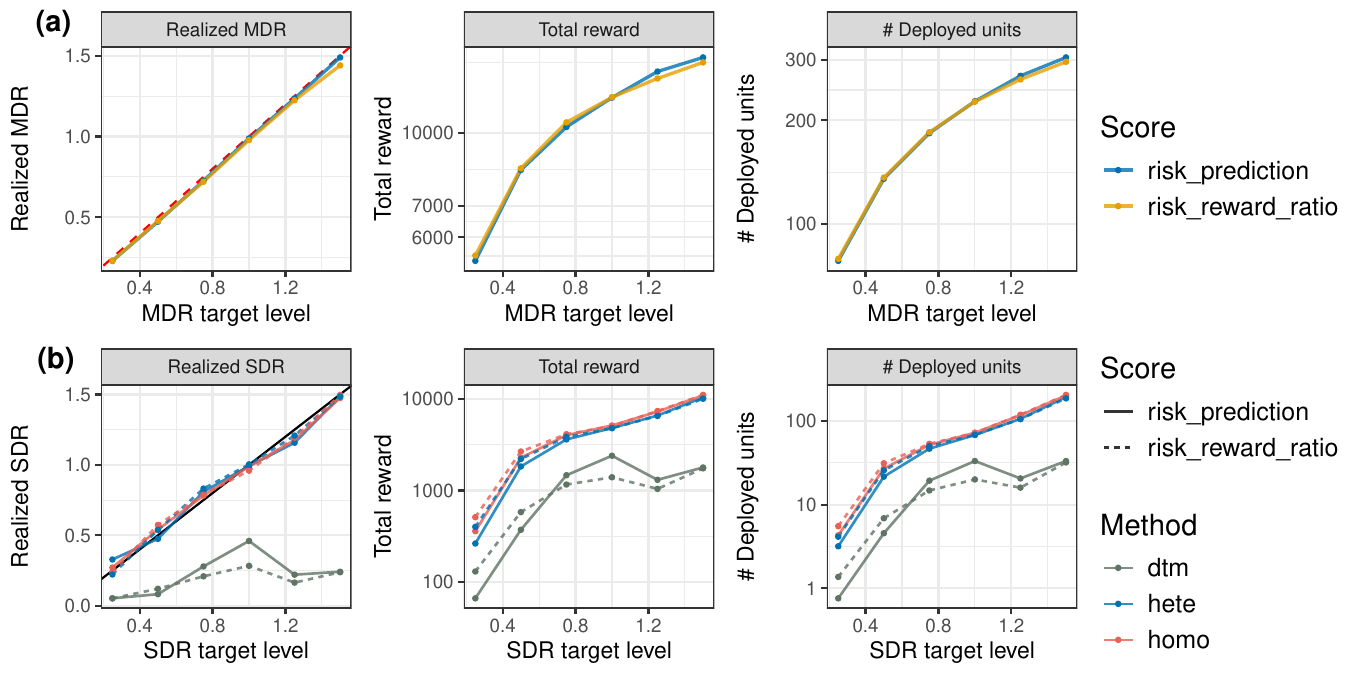}
    \caption{{\small MDR (a) and SDR (b) control for drug discovery with the \texttt{clearance\_hepatocyte} dataset in Therapeutic Data Commons with estimated covariate shift and activity reward. Details are otherwise the same as Figure~\ref{fig:real-drug}.}}
    \label{fig:drug_az_activity}
\end{figure}

\begin{figure}[htbp] 
    \centering
    \includegraphics[width=0.85\linewidth]{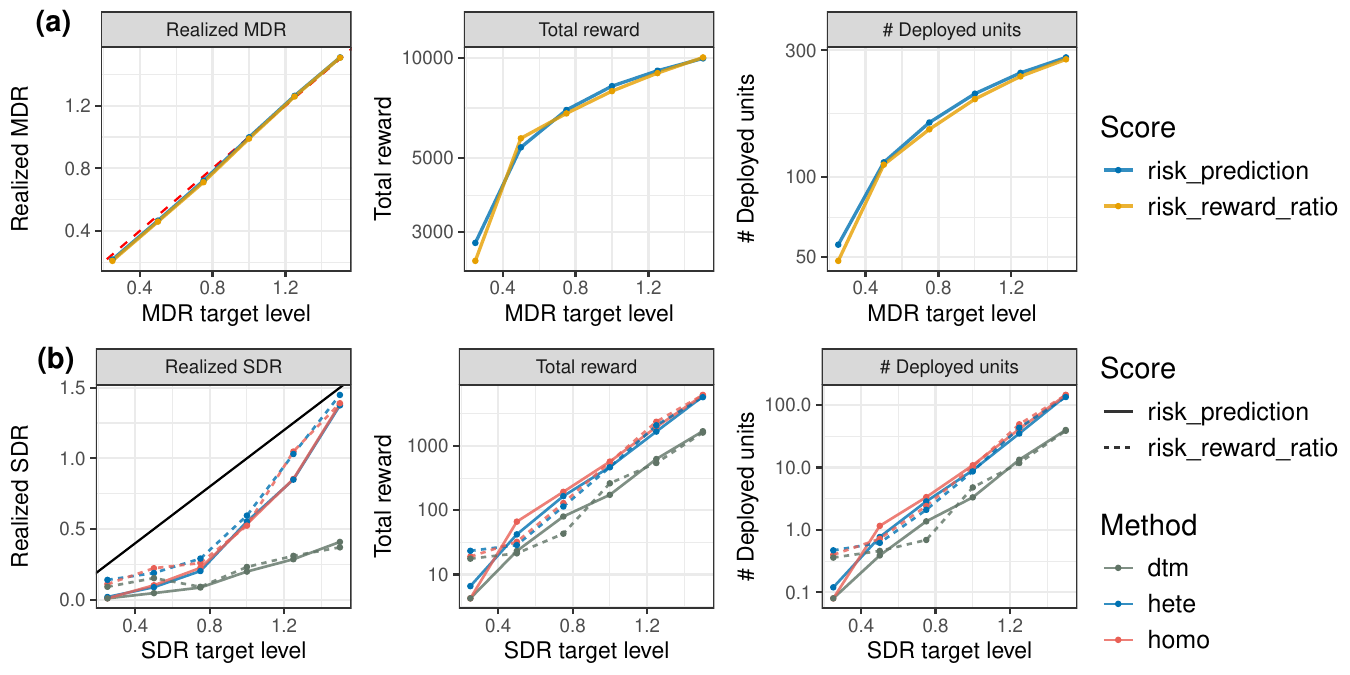}
    \caption{{\small MDR (a) and SDR (b) control for drug discovery with the \texttt{clearance\_microsome} dataset in Therapeutic Data Commons with estimated covariate shift and  activity reward. Details are otherwise the same as Figure~\ref{fig:real-drug}.}}
    \label{fig:drug_microsome_activity}
\end{figure}

\begin{figure}[htbp] 
    \centering
    \includegraphics[width=0.85\linewidth]{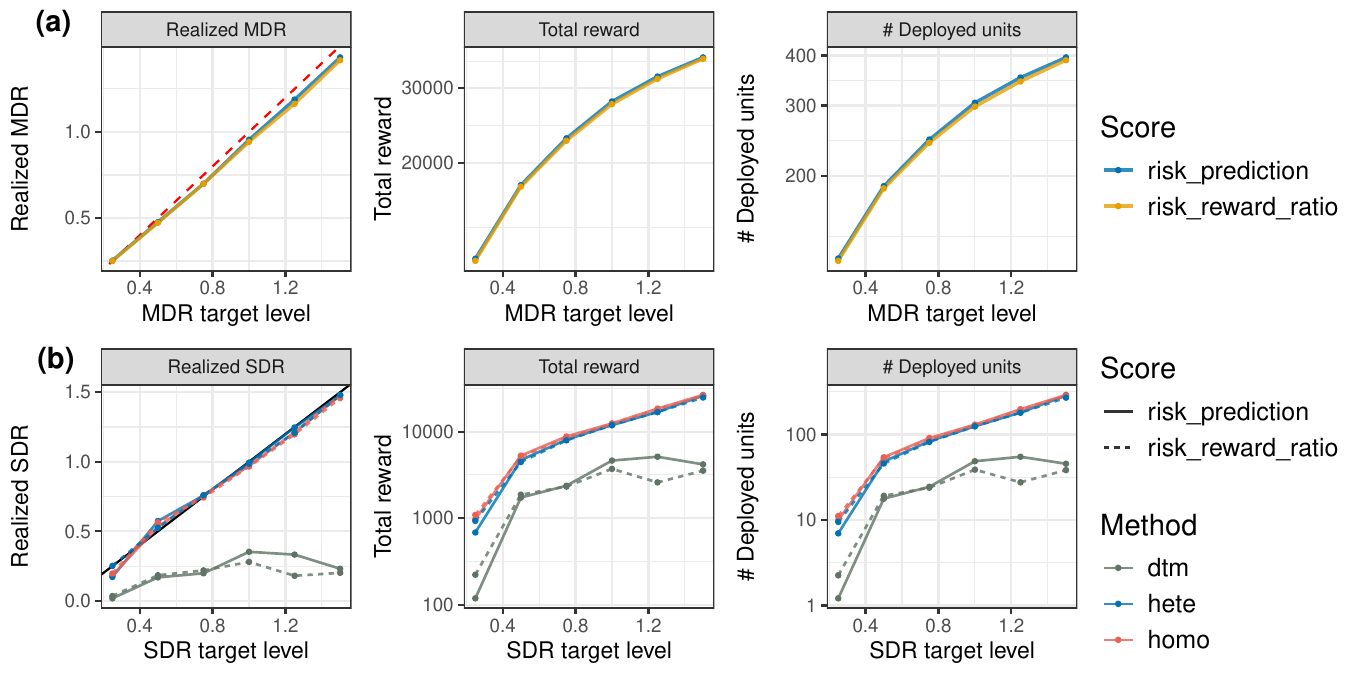}
    \caption{{\small MDR (a) and SDR (b) control for drug discovery with the \texttt{ppbr\_az} dataset in Therapeutic Data Commons with estimated covariate shift and activity reward. Details are otherwise the same as Figure~\ref{fig:real-drug}.}}
    \label{fig:drug_ppbr_activity}
\end{figure}

\subsection{Experiment setups for Section~\ref{subsec:app_llm}} \label{app:llm_detail}

In our LLM abstension application (Section~\ref{subsec:app_llm}), we use the same subset (p10, p11 and p12 folders) of the MIMIC-CXR dataset~\citep{johnson2019mimic} as in~\cite{gui2024conformal}, which is accessed from the PhysioNet project page \url{https://physionet.org/content/mimic-cxr/2.0.0/} under the PhysioNet Credentialed Health Data License 1.5.0. In our experiments, we draw a subset of images in the test folder determined by the same split as \cite{gui2024conformal}. In this way, the randomness is purely from randomly splitting the data into labeled data and test samples. 

The foundation model for generating the radiology reports is the one fine-tuned in \cite{gui2024conformal}. We include the details here for completeness. Specifically, this vision-language model combines the Vision Transformer \texttt{google/vit\-base-patch16-224-in21k} pre-trained on ImageNet-21k\footnote{\url{https://huggingface.co/google/vit-base-patch16-224-in21k}} as the image encoder and GPT2 as the text decoder.  
Each raw image is resized to $224\times 224$ pixels. 
The model  is fine-tuned on a hold-out dataset with a sample size of $43,300$ for $10$ epochs with a batch size of $8$, and other hyperparameters are set to default values. 
When generating reports, all the parameters are kept the same as the conformal alignment paper; we refer the readers to \cite[Appendix C.2]{gui2024conformal} for these details. 

We use exactly the same procedures as \cite{gui2024conformal} to compute $12$ features which (heuristically) measure the uncertainty of LLM-generated outputs:
\vspace{0.5em}
\begin{itemize}%[leftmargin=*, topsep=0pt] 
    \item \emph{Input uncertainty scores} (\texttt{Lexical\_Sim}, \texttt{Num\_Sets}, \texttt{SE}). Following~\cite{kuhn2023semantic}, we compute a set of features that measure the uncertainty of each LLM input through similarity among multiple answers. The features include lexical similarity (\texttt{Lexical\_Sim}), the \texttt{rouge-L} similarity among the  answers. In addition, we use a natural language inference (NLI) classifier to categorize the $M$ answers into semantic groups, and compute the number of semantic sets (\texttt{Num\_Sets}) and semantic entropy (\texttt{SE}).  Following \cite{kuhn2023semantic, lin2023generating}, 
    we use an off-the-shelf DeBERTa-large model \citep{he2020deberta} as the NLI predictor.
    \item \emph{Output confidence scores} (\texttt{EigV(J/E/C)}, \texttt{Deg(J/E/C)}, \texttt{Ecc(J/E/C)}). We also follow~\citep{lin2023generating} to compute features that measure the so-called output confidence: with $M$ generations, we compute the eigenvalues of the graph Laplacian (\texttt{EigV}), the pairwise distance of generations based on the degree matrix (\texttt{Deg}), and the Eccentricity (\texttt{Ecc}) which incorporates the embedding information of each generation. Note that each quantity is associated with a similarity measure; we follow the notations in \cite{lin2023generating} and use the suffix \texttt{J}/\texttt{E}/\texttt{C} to differentiate similarities based on the Jaccard metric, NLI prediction for the entailment class, and NLI prediction for the contradiction class, respectively. 
\end{itemize}
\vspace{0.5em}

A CheXbert model~\cite{smit2020chexbert} is employed to evaluate the factuality of LLM generated radiology reports. The model converts both the reference report from human experts and the generated report into two 14-dimensional vector, where each entry indicates the presence, absence, uncertainty or lack of mention for a medical condition. Based on these values, we set the risk as
\$
    \cL(f, X, Y) = \text{\# of type-I error} + 1/2 \cdot \text{\# of type-II error}
\$
where the type-I and type-II errors correspond to mismatched label values when the reference label is positive or otherwise. The confidence-weighted reward is defined as
\$
    r_1(X, Y) = 4 \cdot \text{\# of non-ambiguous matching labels} + \text{\# of other matching labels}
\$
where `non-ambiguous' means that the matching answer neither uncertain nor lack of mention. This reward encourages the selection of LLM outputs that make confident, definitive statements, thereby avoiding degenerate cases where the generated reports are dominated by uncertain or clinically uninformative conclusions.
Finally, a random forest model with default parameters is used for risk and reward prediction.

\subsection{Simulation setups for Section~\ref{sec:simu}} \label{app:simu_detail}

For each setting in the simulation studies, we draw covariates $X \sim \text{Unif}[-1,1]^d$, with the dimension taken to be $d=20$. We then form the responses as $Y = \mu(X) + \epsilon$, where the regression function $\mu$ and the noise distribution are detailed in Table~\ref{tab:simu_settings}. The same table also reports the definition of the risk function $\cL(f, X, Y)$ under each setting.

\renewcommand{\arraystretch}{2}
\begin{table}[h]
    \centering
    \small
    \begin{tabular}{c|c|c|c}
        \hline
        Setting & $\mu(\cdot)$ & $\epsilon_i$ & $\cL(\cdot)$ \\
        \hline
        1 & \makecell{$3 +\ind\{x_1x_2 > 0, x_4 > 0.5\} \cdot (x_4 + 0.5)$ \\ \hspace{0.5cm} $+ \ind\{x_1x_2 \leq 0, x_4 < -0.5\} \cdot (x_4 - 0.5) $} & $\textnormal{clip}( \sigma (5.5 - \mu(x)), -1.5, 1.5)$ & $\frac{1}{6} Y\ind\{Y > 2\}$ \\
        \hline
        2 & $2 +x_1x_2 + x_3^2 + e^{x_4 - 1}$ & $\textnormal{clip}( \sigma (6 - \mu(x)), -1, 1)$ & $\frac{1}{6} Y\ind\{Y > 2\}$ \\
        \hline
        3 & \makecell{$3 +\ind\{x_1x_2 > 0, x_4 > 0.5\} \cdot (x_4 + 0.5)$ \\ \hspace{0.5cm} $+ \ind\{x_1x_2 \leq 0, x_4 < -0.5\} \cdot (x_4 - 0.5) $} & $\textnormal{clip}( \sigma (5.5 - \mu(x)), -1.5, 1.5)$ & $\frac{1}{c} \text{clip}((Y-f(X))^2, 0, c)$ \\
        \hline
        4 & $2 +x_1x_2 + x_3^2 + e^{x_4 - 1}$ & $\textnormal{clip}( \sigma (6 - \mu(x)), -1, 1)$ & $\frac{1}{c} \text{clip}((Y-f(X))^2, 0, c)$ \\
        \hline
        5 & \makecell{$\ind\{x_1x_2 > 0, x_4 > 0.5\} \cdot (x_4 + 0.25)$ \\ \hspace{0.5cm} $+ \ind\{x_1x_2 \leq 0, x_4 < -0.5\} \cdot (x_4 - 0.25) $} & $\sigma (5.5 - \mu(x)) / 2$ & $\text{sigmoid}(-Y \cdot \tau)$ \\
        \hline
        6 & $x_1x_2 + x_3^2 + e^{x_4 - 1}$ & $\sigma (5.5 - \mu(x)) / 2)$ & $\text{sigmoid}(-Y \cdot \tau)$ \\
        \hline
    \end{tabular}
    % \captionsetup{font=small}
    \caption{{\small Details of the six data generating processes used in the simulation studies.}}
    \label{tab:simu_settings}
\end{table}

% maybe add the motivation of designing the 3 different risk functions?

In Table~\ref{tab:simu_settings}, we write the clipping operator as $\text{clip}(x, a, b) := \max\{a, \min\{b, x\}\}$ for $a, b \in \mathbb{R}, a \leq b$, and in setting 1-4, we apply it to the noise and predictor MSEs so that every risk value is confined to $[0,1]$, which is required for our procedure. In setting 3 and 4, the clipping constant $c$ is set to 0.6 and 0.4 respectively, corresponding to the approximate 0.95-th quantile of the MSE in these experiments (therefore, $c$ varies with different noise levels). Both settings also employ a pre-trained prediction model $f$, implemented as a random forest model (using the \texttt{scikit-learn} Python package) and fitted using an independent hold-out sample of 1000 observations. 
In setting 5 and 6, the sigmoid function is defined as $\text{sigmoid}(z) = 1/(1+e^{-z})$, and the temperature parameter $\tau$ is set to 10. A larger $\tau$ produces a closer approximation of the true indicator function. We use the parameter $\sigma$ to scale the noise level, and here $\sigma$ is fixed at $0.1$ in all settings. Finally, the risk and reward estimators $\hat{l}$ and $\hat{r}$ are instantiated as two random forest models and trained on an independent training dataset of size 1000. 

In the covariate shift setting, we apply an artificially crafted reweighting function $w$ to the covariates. Specifically, we define $w(x) = \text{sigmoid}(\theta^\top x)$, where $\theta_i = 0.1 \cdot \ind\{i \leq 5\}$.
The weights are estimated using probabilistic classification on an additional dataset of 2000 observations (1000 from each population).

\subsection{Details for baseline implementations in Section~\ref{subsec:simu_risk_control}}
\label{app:simu_baselines}

% TODO
Here we provide detail for the baseline methods introduced in Section~\ref{subsec:simu_risk_control}. For the MDR case, the two variants \texttt{Hoeffding} and \texttt{Rademacher} give different uniform bounds on $\mdr(t)$. The \texttt{Hoeffding} approach fixes a grid $\mathcal{G}$ consisting of $|\mathcal{G}| = 101$ evenly-spaced points between 0 and 1, and set $\epsilon_n = \sqrt{\log(2|\mathcal{G}|/\delta) / 2n}$ as the slack determined by Hoeffding's inequality. As such, $\widehat{\mdr}(t) + \epsilon_n$ is a uniform upper bound on $\mdr(t)$ over all $t \in \mathcal{G}$ with probability at least $1-\delta$. With $\hat{t} = \max\{t\in \cG\colon \widehat{\mdr}(t) + \epsilon_n \leq \alpha \}$, we have the PAC-type guarantee:
\begin{align*}
    \mdr(\hat{t}) \leq \alpha \quad \text{with probability} \geq 1-\delta.
\end{align*}
Similarly, the \texttt{Rademacher} approach bounds $\widehat{\mdr}(t)$ over all $t \in [0,1]$ by $\widehat{\mdr}(t) + 2 \widehat{\text{Rad}}(\mathcal{D}_{\text{calib}}) + 3\sqrt{\log(2/\delta)/2n}$. Here, $\widehat{\text{Rad}}(\mathcal{D}_{\text{calib}})$ denotes the empirical Rademacher complexity of the function class $\{t \mapsto L_i \ind\{s(X_i) \leq t\}\}$ for all $(X_i, L_i) \in \mathcal{D}_{\text{calib}}$; it is evaluated by empirically sampling $k=100$ Rademacher random variables. The slack term $3\sqrt{\log(2/\delta)/2n}$ is added once to account for the estimation of the empirical MDR and twice for that of the empirical Rademacher complexity. With this uniform upper bound on $t \in [0,1]$, we set the grid to all predicted values $\mathcal{G} = \{s(X_i)\}_{i=1}^n$ for tightness. It is straightforward to see that such approach also ensures above PAC-type guarantee.

For SDR control, the two variants are constructed similarly. Both variants bound the numerator $\EE[\cL(f, X,Y) \ind\{s(X) \leq t\}]$ and denominator $\PP(s(X) \leq t)$ separately. For the \texttt{Hoeffding} variant, the numerator is upper bounded by
\begin{align*}
    A_h(t) := \frac{1}{n} \sum_{i=1}^n L_i \ind\{s(X_i) \leq t\} + \sqrt{\frac{1}{2n} \log(4|\mathcal{G}|/\delta)}
\end{align*}
and the denominator is lower bounded by
\begin{align*}
    B_h(t) := \frac{1}{n} \ind\{s(X_i) \leq t\} - \sqrt{\frac{1}{2n} \log(4|\mathcal{G}|/\delta)}.
\end{align*}
Setting $\mathcal{G}$ to be a fixed evenly-spaced grid of size $|\mathcal{G}| = 100$, above bounds hold uniformly over $t\in \mathcal{G}$ with probability at least $1-\delta/2$. Therefore, with probability at least $1-\delta$, $\widehat{\sdr}^+(t) := A_h(t)/B_h(t)$ if $B_h(t) > 0$ and $\infty$ otherwise is an uniform upper bound on $\sdr^*(t)$. Now, for the \texttt{Rademacher} approach, we set $\mathcal{G} = \{s(X_i)\}_{i=1}^n$, and the upper and lower bounds are
\begin{align*}
    A_r(t) &:= \frac{1}{n} \sum_{i=1}^n L_i \ind\{s(X_i) \leq t\} +2 \widehat{\text{Rad}}(\mathcal{D}_{\text{calib}}) + 3\sqrt{\frac{1}{2n} \log(4/\delta)}, \\
    B_r(t) &:= \frac{1}{n}  \ind\{s(X_i) \leq t\} - 2\widetilde{\text{Rad}}(\mathcal{D}_{\text{calib}}) + 3\sqrt{\frac{1}{2n} \log(4/\delta)}
\end{align*}
where $\widetilde{\text{Rad}}(\mathcal{D}_{\text{calib}})$ denotes the empirical Rademacher complexity of the function class $\{t \mapsto \ind\{s(X_i) \leq t\}\}$. Again, $\widehat{\sdr}^+(t) := A_r(t)/B_r(t)$ if $B_r(t) >0$ and $\infty$ otherwise would be a valid uniform upper bound on $\sdr^*(t)$. As constructed, \texttt{Hoeffding} and \texttt{Rademacher} variants guarantees $\sdr^*(\hat{t}) \leq \alpha$ with probability at least $1-\delta$.

\subsection{Additional simulation results in Section~\ref{subsec:simu_weight}}
\label{app:simu_weight_res}

In this section, we present the omitted results for SCoRE under covariate shifts in Section~\ref{subsec:simu_weight}.

Figure~\ref{fig:weight_mdr_full} presents the complete results for SCoRE-MDR with estimated weights under the three covariate shift models. 
Figures~\ref{fig:weight_sdr_full_err},~\ref{fig:weight_sdr_full_nsel}, and~\ref{fig:weight_sdr_full_reward} present the realized SDR, number of selections, and total reward from SCoRE-SDR with estimated weights under the three models.  

\begin{figure}[h]
    \centering
    \includegraphics[width=0.9\linewidth]{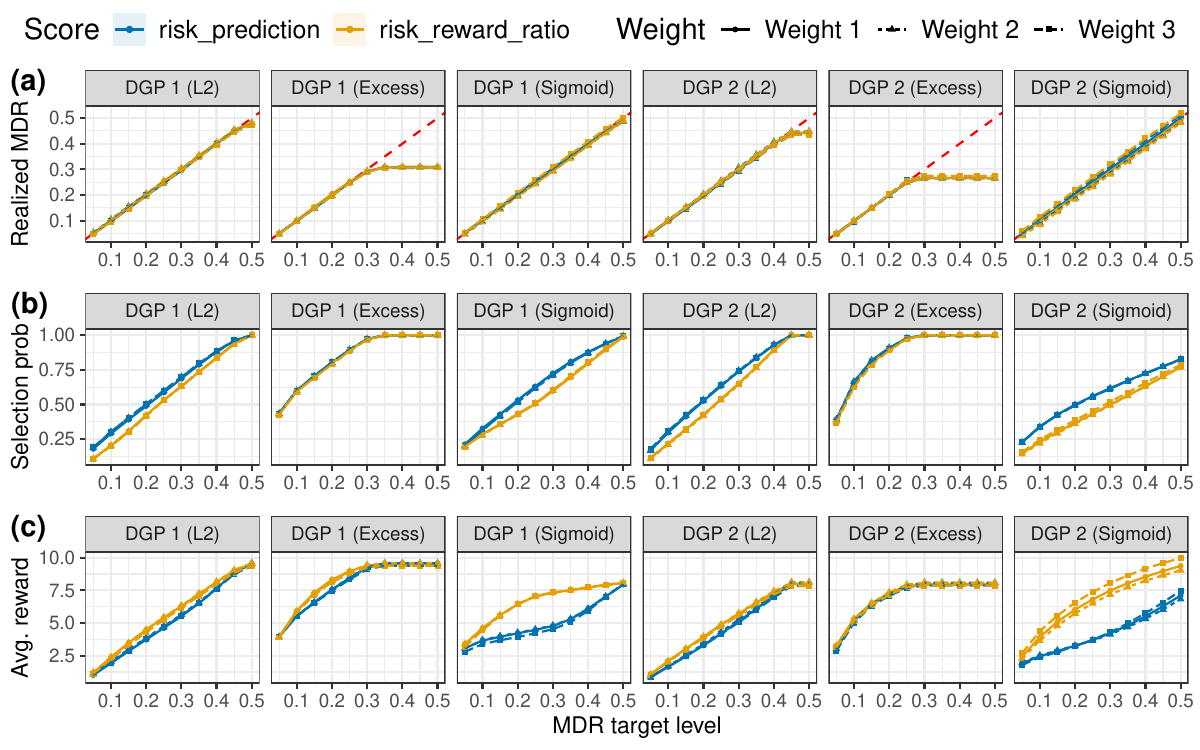}
    \caption{{\small Results for SCoRE-MDR with estimated weights under three covariate shift models (with the reward of Sigmoid risk re-scaled for easier visualization). Details are otherwise the same as in Figure~\ref{fig:simu_mdr}.}}
    \label{fig:weight_mdr_full}
\end{figure}

\begin{figure}[h]
    \centering
    \includegraphics[width=0.9\linewidth]{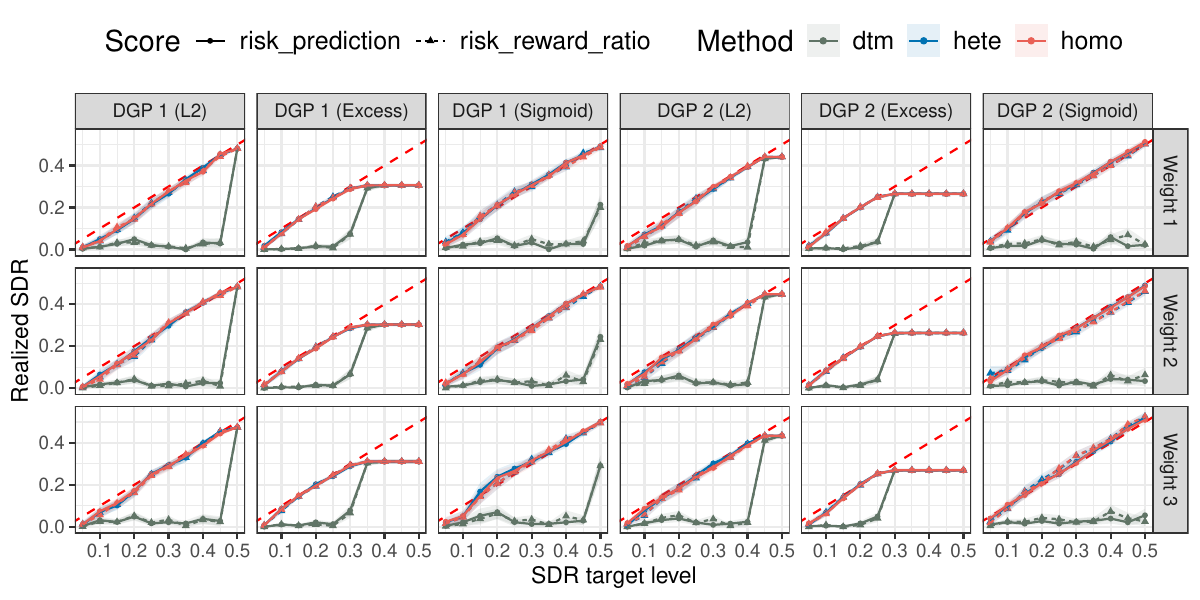}
    \caption{{\small Realized SDR of SCoRE-SDR  with estimated weights under three covariate shift models. Each row is a weight model. Details are otherwise the same as in Figure~\ref{fig:simu_sdr}.}}
    \label{fig:weight_sdr_full_err}
\end{figure}

\begin{figure}[H]
    \centering
    \includegraphics[width=0.9\linewidth]{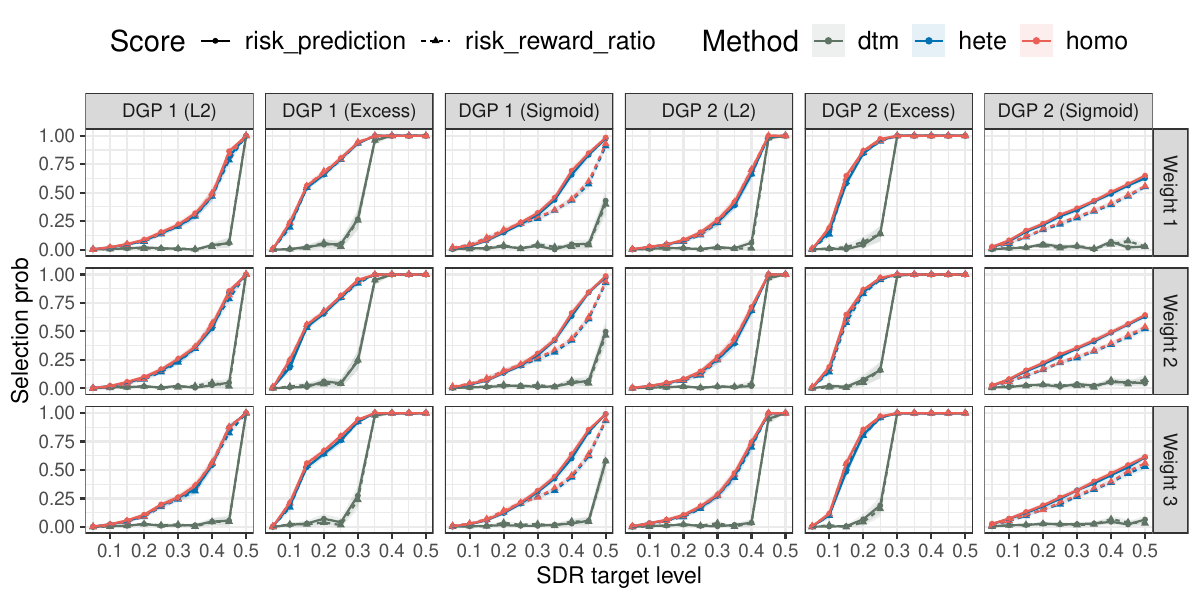}
    \caption{{\small Number of selection by SCoRE-SDR  with estimated weights under three covariate shift models. Each row is a weight model. Details are otherwise the same as in Figure~\ref{fig:simu_sdr}.}}
    \label{fig:weight_sdr_full_nsel}
\end{figure}

\begin{figure}[H]
    \centering
    \includegraphics[width=0.9\linewidth]{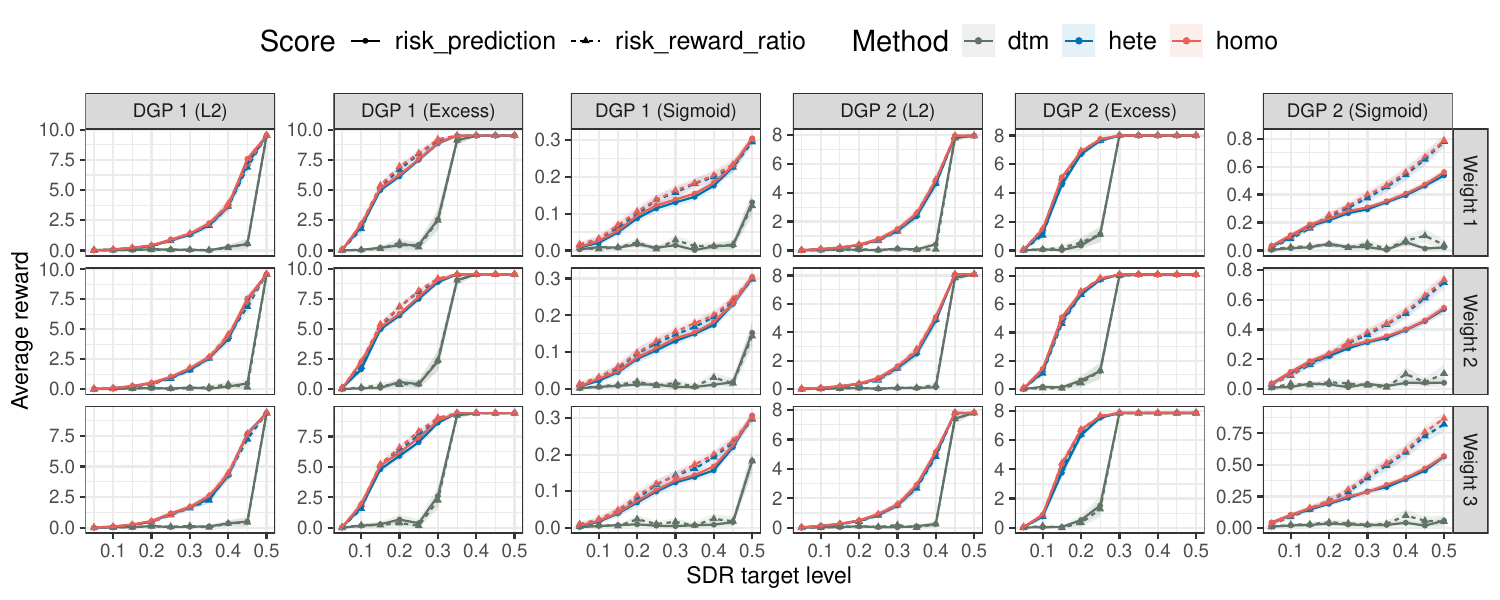}
    \caption{{\small Average total reward of SCoRE-SDR  with estimated weights under three covariate shift models. Each row is a weight model. Details are otherwise the same as in Figure~\ref{fig:simu_sdr}.}}
    \label{fig:weight_sdr_full_reward}
\end{figure}

% \clearpage

\section{Auxiliary lemmas}

    \begin{lemma}\label{lem:weighted_cdf_conv}
        Let $f\colon \cX\to [0,M]$ be any fixed, bounded function, and $s\colon \cX\to \RR$ be a fixed function so that $s(X)$ has no point mass for $X\sim Q$. Let $\{X_i\}_{i=1}^m$ be i.i.d.~samples from $Q$ and independent of $f$ and $s$.  Then there exists a universal constant $C>0$ such that 
        \$
        \EE\Bigg[ \sup_{t\in \RR} \bigg| \frac{1}{m} \sum_{i=1}^m f(X_i)\ind\{s(X_i)\leq t\} - \EE_Q[f(X)\ind\{s(X)\leq t\}] \bigg|  \Bigg] \leq \frac{CM}{\sqrt{m}}. 
        \$
    \end{lemma}
    \begin{proof}[Proof of Lemma~\ref{lem:weighted_cdf_conv}]
    Define 
    \$
    f_t(x) := f(x)\mathbf 1\{s(x)\le t\},\qquad \mathcal F := \{f_t: t\in\mathbb R\}.
    \$
        Consider the function class $\mathcal H:=\{h_t(x)=\mathbf 1\{s(x)\le t\}: t\in\mathbb R\}$ which is well-known to be a VC class. Hence there exist constants $A,v<\infty$ (e.g., $A=\sqrt{2}$, $v=2$) such that the covering number obeys
\[
N\big(\varepsilon,\mathcal H,L_2(\PP)\big)\le \Big(\frac{A}{\varepsilon}\Big)^v,\qquad 0<\varepsilon<1.
\] 
Due to the boundedness of $f(\cdot)$, it is straightforward to see that 
\[
N\big(\varepsilon,\mathcal F,L_2(\PP)\big)\le \Big(\frac{AM}{\varepsilon}\Big)^v,\qquad 0<\varepsilon<1,
\] 
so $\mathcal F$ is a VC-type class with envelope $F(x)\equiv M$.
By a standard maximal inequality for VC-type classes we obtain, for a universal constant $C_0>0$, that 
\[
\mathbb E\Big[\sup_{\tilde{f}\in\mathcal F}\big|\sqrt{m}(P_m-P)f\big| \Big]\le C_0 \|F\|_{L_2(Q)}=C_0 M,
\]
where $P_m(\tilde{f})=\frac{1}{m}\sum_{i=1}^n \tilde{f}(X_i)$ and $P(\tilde{f})=\EE[f(X)]$.  
Dividing by $\sqrt m$ yields the displayed expectation bound.
    \end{proof}
    % Lemma~\ref{lem:cond_conv_op1} is from \cite{jin2024tailored}, Lemma I.5, whose proof is omitted here.
    % \begin{lemma}\label{lem:cond_conv_op1}
    %     Let $\{\cF_n\}_{n\geq 1}$ be a sequence of $\sigma$-algebra, and let $\{A_n\}_{n\geq 1}$ be a sequence of non-negative random variables. If $\EE(A_n\given \cF_n) =o_P(1)$, then $A_n = o_P(1)$. 
    % \end{lemma}

\end{document}